\documentclass[11pt, a4paper]{imsart}
 
\usepackage{epigraph}
\usepackage[margin=3.5cm]{geometry}
\RequirePackage{amsthm,amsmath,amsfonts,amssymb}
\RequirePackage[numbers]{natbib}

\startlocaldefs
\theoremstyle{plain}
\newtheorem{axiom}{Axiom}
\newtheorem{claim}[axiom]{Claim}
\newtheorem{theorem}{Theorem}[section]
\newtheorem{lemma}[theorem]{Lemma}
\newtheorem{conjecture}[theorem]{Conjecture}
\newtheorem{corollary}[theorem]{Corollary}
\newtheorem{proposition}[theorem]{Proposition}
\newtheorem{definition}[theorem]{Definition}
\newtheorem{fact}[theorem]{Fact}
\theoremstyle{remark}

\newtheorem{remark}[theorem]{Remark}

\input{insbox}
\usepackage{hyperref}
\usepackage{pgfplots}
\usepackage{tikz}
\usepackage{graphicx}
\usepackage{subcaption}
\usetikzlibrary{shapes,decorations,arrows,calc,arrows.meta,fit,positioning}
\usepgfplotslibrary{fillbetween}
\usepackage{graphicx}
\usetikzlibrary{decorations.pathreplacing, matrix}
\usepackage{tabularx}
\usepackage[figuresright]{rotating}
\usepackage{makecell}
\usepackage{mathrsfs}
\usepackage{dsfont}
\usepackage[boxed, linesnumbered, noend, noline]{algorithm2e}
\SetKwInput{KwData}{Input}
\SetKwInput{KwResult}{Output}
\usepackage{enumerate}
\usepackage{tikz}
\usetikzlibrary{shapes,decorations,arrows,calc,arrows.meta,fit,positioning}
\usetikzlibrary{shadows,shadings,shapes.symbols, patterns}
\tikzset{
    double color fill/.code 2 args={
        \pgfdeclareverticalshading[%
            tikz@axis@top,tikz@axis@middle,tikz@axis@bottom%
        ]{diagonalfill}{100bp}{%
            color(0bp)=(tikz@axis@bottom);
            color(50bp)=(tikz@axis@bottom);
            color(50bp)=(tikz@axis@middle);
            color(50bp)=(tikz@axis@top);
            color(100bp)=(tikz@axis@top)
        }
        \tikzset{shade, left color=#1, right color=#2, shading=diagonalfill}
    }
}

\usepackage[colorinlistoftodos,prependcaption,textsize=small]{todonotes}
\usepackage{xcolor}

\usepackage{bbm}

\newcommand{\e}{\mathrm{e}}

\newcommand{\vA}{\vec A}

\newcommand{\vf}{\vec f}
\newcommand{\vg}{\vec g}

\newcommand{\vD}{\vec D}

\newcommand{\vX}{\vec X}

\newcommand{\vL}{\vec L}

\renewcommand{\epsilon}{\eps}

\newcommand\vY{\vec Y}
\newcommand\vB{\vec B}

\newcommand\vm{\vec m}

\newcommand\vU{\vec U}

\newcommand\vV{\vec V}

\newcommand\vR{\vec R}

\renewcommand{\vec}[1]{\boldsymbol{#1}}

\newcommand\KL[2]{D_{\mathrm{KL}}\bc{{{#1}\|{#2}}}}

\newcommand\SIGMA{\vec\sigma}

\newcommand\cA{\mathcal{A}}
\newcommand\cB{\mathcal{B}}

\newcommand\cF{\mathcal{F}}
\newcommand\G{\mathcal{G}}
\newcommand\cE{\mathcal{E}}

\newcommand\cN{\mathcal{N}}

\newcommand\cS{\mathcal{S}}
\newcommand\cT{\mathcal{T}}

\newcommand\cM{\mathcal{M}}

\newcommand\cP{\mathcal{P}}

\def\cR{{\mathcal R}}
\def\cE{{\mathcal E}}

 \newtheorem{assumption}{Assumption}

\newcommand\vE{\vec E}

\newcommand\vZ{\vec Z}
\newcommand\vM{\vec M}

\newcommand\eul{\mathrm{e}}
\newcommand\eps{\varepsilon}

\newcommand\NN{\mathbb{N}}

\newcommand\Var{\mathrm{Var}}
\newcommand\Erw{\mathbb{E}}
\newcommand{\vecone}{\vec{1}}

\newcommand{\Po}{{\rm Po}}
\newcommand{\Bin}{{\rm Bin}}
\newcommand{\Mult}{{\rm Mult}}
\newcommand{\Be}{{\rm Be}}

\newcommand\bc[1]{\left({#1}\right)}
\newcommand\cbc[1]{\left\{{#1}\right\}}
\newcommand\bcfr[2]{\bc{\frac{#1}{#2}}}

\newcommand\brk[1]{\left\lbrack{#1}\right\rbrack}

\newcommand\abs[1]{\left|{#1}\right|}

\newcommand\RR{\mathbb{R}}

\newcommand{\whp}{w.h.p.}

\newcommand\pr{\mathbb{P}} 
\renewcommand\Pr{\pr}

\newcommand\Thm{Theorem}

\newcommand\Cor{Corollary}

 \def\G{{\vec G}}

\def\pr{{\mathbb P}}

\def\bfm{{\vec m}}

\def\vGamma{\vec{\Gamma}}

\newcommand{\remove}[1]{}

\newcommand\minf{m_{\mathrm{inf}}}
\newcommand\cinf{c_{\mathrm{inf}}^{\mathrm{TGT}}}

\newcommand{\nm}[1]{\textcolor{black}{#1}}
\newcommand{\cri}[1]{\textcolor{black}{#1}}

\pgfplotsset{compat=1.14}

\usepackage{float}

\newcommand{\invisible}[1]{}

\endlocaldefs

\begin{document}

\begin{frontmatter}
\title{An Information-Theoretic Analysis of Threshold Group Testing}

\begin{aug}



\author[A]{\fnms{Remco}~\snm{van der Hofstad}\ead[label=e2a]{r.w.v.d.hofstad@tue.nl}\orcid{0000-0003-1331-9697}}
\author[A]{\fnms{Noela}~\snm{Müller}\ead[label=e2b]{n.s.muller@tue.nl}\orcid{0009-0004-8182-311X}}
\and
\author[A]{\fnms{Connor}~\snm{Riddlesden}\ead[label=e2c]{c.d.riddlesden@tue.nl}\orcid{0009-0003-4516-7349}}

\address[A]{Eindhoven University of Technology, Department of Mathematics and Computer Science\printead[presep={,\ }]{e2a,e2b,e2c}}
\end{aug}

\begin{abstract}
We study the Threshold Group Testing (TGT) problem in the noiseless and non-adaptive setting, where the objective is to exactly recover a sparse binary vector from pooled tests, using as few tests as possible. In TGT, each test applied to a subset of items returns a positive outcome if the number of $1$'s (`\emph{defective items}') in that subset meets or exceeds a specified threshold, and has a negative outcome otherwise. We investigate how the complexity of TGT compares to that of Classical Group Testing (CGT), corresponding to the special case of the threshold equal to one, and analyse the impact of increasing the threshold on the required number of tests. 

Our main contribution is the derivation of a sharp information-theoretic phase transition at $c_{\mathrm{inf}}^{\mathrm{TGT}}k\log(n/k)$ (non-adaptive) tests for TGT within the constant-column test design. The threshold constant $c_{\mathrm{inf}}^{\mathrm{TGT}}$ is expressed as a function of the prevalence of defectives and the threshold value. 
Our upper bound is derived under an analytic assumption, and we verify that this assumption is satisfied for a threshold value of $2$.

The value of $c_{\mathrm{inf}}^{\mathrm{TGT}}$ reveals that TGT on the constant-column design has the same information-theoretic behaviour as CGT in the low-prevalence regime. Yet, strikingly, at higher prevalences, the threshold leads to a significant \emph{reduction} in the number of tests.  

On the other hand, we provide evidence that when the asymptotic proportion of defective items is positive, TGT actually becomes strictly harder than CGT (excluding trivial reductions).
\end{abstract}

\begin{keyword}[class=MSC]
\kwd[Primary ]{68W20} 
\kwd{62B10} 
\kwd[; Secondary ]{68P30} 
\end{keyword}

\begin{keyword}
\kwd{Threshold Group Testing}
\kwd{Statistical Inference}
\end{keyword}

\end{frontmatter}


\tableofcontents

\section{Introduction}\label{Sec:Intro}
In this article, we consider the \emph{threshold group testing} (TGT) problem. More precisely, we are given a vector $\SIGMA \in \{0,1\}^n$ of Hamming weight $k$ \footnote{Throughout the paper, we use boldface notation to denote random quantities.}. When taking a subset of the coordinates of $\SIGMA$ and conducting a measurement, then, for a given threshold $t\geq 1$, the output will be $1$ if the number of $1$-coordinates of the subset is at least $t$, and $0$ otherwise. This specific case of TGT is commonly referred to as \emph{threshold group testing without a gap}, since there is no in between test result for which the test is \emph{indeterminate}. The aim is to exactly reconstruct $\SIGMA$  using the minimum number of measurements.

The concept of testing through pooled measurements was first proposed by Dorfman in 1943 to tackle the problem of syphilis in the US army \citep{dorfman_1943}. Dorfman's problem corresponds to setting $t=1$, and is commonly referred to as \emph{classical} or \emph{binary group testing}. Classical group testing (CGT) is well understood both information theoretically and algorithmically \cite{AJS_book, coja_spiv}. 
An early and influential treatment of the generalisation of the problem to $t>1$ was given by Damaschke \citep{damaschke2006threshold}.

\subsection{Model} \label{SSec:Model}

We now give the details of our model. The number of items will be denoted by $n$. Each of the $n$ items $x_1, \ldots, x_n$ is assigned a label $\SIGMA_i:=\SIGMA(x_i) \in \cbc{0, 1}$, where items with label $1$ are referred to as \emph{defectives}. The vector $\SIGMA$ constitutes the unknown \emph{ground-truth}, and $n$ will tend to infinity below. 
 We are generally interested in an \emph{average-case analysis}, where $\SIGMA$ is randomly chosen. The choice of the distribution of $\SIGMA$ depends on the specific model parameters, and will be made precise in the appropriate places.

A \emph{pooling scheme} for $x_1, \ldots, x_n$  is defined by a collection of $m$ pools $a_1, \ldots, a_m$, where each $a_i$ is a subset of $\{x_1, \ldots, x_n\}$ of arbitrary cardinality. 
We visualise pooling schemes $G$ through bipartite graphs (see Figure \ref{fig:example}). In this graphical representation, the $n$ items constitute one set of vertices, and the $m$ pools $a_1, \ldots, a_m$ constitute the second set of vertices. An edge between an item node $x_i$ and a pool node $a_j$ indicates that item $x_i$ is part of pool $a_j$. 
Additionally, every pool node $a_j$ comes with a label vector $\hat{\SIGMA}_j$, which will be $1$ if the number of item nodes connected to $a_j$ with label $1$ is greater than or equal to $t$, and $0$ otherwise (see Figure \ref{fig:example}). 
The vector of measurement results of a given pooling scheme $G$ will be denoted by $\hat \SIGMA_{G}$.

\begin{figure}[ht]
\centering
\begin{tikzpicture}[scale=0.65]
\node[circle, draw, minimum width=0.66cm] (x0) at (0, 0) {$1$};
\node[circle, draw, minimum width=0.66cm] (x1) at (2,0) {$1$};
\node[circle, draw, minimum width=0.66cm] (x2) at (4, 0) {$0$};
\node[circle, draw, minimum width=0.66cm] (x3) at (6, 0) {$0$};
\node[circle, draw, minimum width=0.66cm] (x4) at (8, 0) {$1$}; 
\node[circle, draw, minimum width=0.66cm] (x5) at (10, 0) {$0$};
\node[circle, draw, minimum width=0.66cm] (x6) at (12, 0) {$0$};

\node[rectangle, draw, minimum width=0.5cm, minimum height=0.5cm] (a1) at (0, -2.5) {$0$};
\node[rectangle, draw, minimum width=0.5cm, minimum height=0.5cm] (a2) at (3,-2.5) {$1$};
\node[rectangle, draw, minimum width=0.5cm, minimum height=0.5cm] (a3) at (6, -2.5) {$1$};
\node[rectangle, draw, minimum width=0.5cm, minimum height=0.5cm] (a4) at (9, -2.5) {$0$};
\node[rectangle, draw, minimum width=0.5cm, minimum height=0.5cm] (a5) at (12, -2.5) {$0$};

\path[draw] (x0) -- (a1);
\path[draw] (x0) -- (a2);
\path[draw] (x0) -- (a3);
\path[draw] (x3) -- (a1);
\path[draw] (x1) -- (a3);
\path[draw] (x2) -- (a1);
\path[draw] (x2) -- (a3);
\path[draw] (x2) -- (a2);
\path[draw] (x3) -- (a3);
\path[draw] (x3) -- (a4);
\path[draw] (x3) -- (a5);
\path[draw] (x4) -- (a2);
\path[draw] (x4) -- (a5);
\path[draw] (x4) -- (a4);
\path[draw] (x5) -- (a2);
\path[draw] (x5)-- (a4);
\path[draw] (x6) -- (a3);
\path[draw] (x6) -- (a5);
\path[draw] (x6) -- (a4);
\end{tikzpicture}
\caption{Graphical representation of a pooling scheme for $t=2$: Here, the $n=7$ items are represented by circles, while the $m=5$ pools are represented by rectangles. Edges between items and pools are present whenever an item is an element of the corresponding pool. The label of an item represents its actual label, while the label of a pool represents the output of its measurement.}
\label{fig:example}
\end{figure}
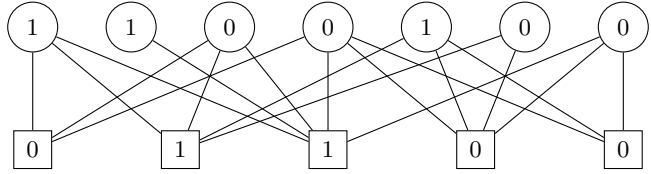

\subsection{Main Result and Motivating Question} \label{SSec:Quest}
Explicit results for both information theoretic bounds, and algorithms to match, have been found for classical group testing ($t=1$), see for example the comprehensive survey \cite{AJS_book}.  
Indeed, \nm{under certain assumptions (e.g. availability of known defectives or systematic use of multi-edges),} there are several ways to \textit{reduce} TGT for $t\geq 2$ to CGT, demonstrating that in principle, \nm{at least under such assumptions,} TGT is not harder than CGT. 
We mention two such reductions here.

\paragraph{Reduction via Dummy Items.}

This idea assumes that a small set of $t-1$ known defectives, denoted $V_{\mathrm{dummy}}$, is available. Given these, we may use any pooling scheme $\G$ that guarantees with high probability (w.h.p.)\footnote{Here, and below, we say that a sequence of events $(\cE_n)_n$ occurs \emph{with high probability} (\whp) if $\pr(\cE_n) \to 1$ as $n \rightarrow \infty$.} exact recovery of $k$ defectives using binary tests also for TGT: By 
 inserting the $t - 1$ known defectives into every test of $\G$, we effectively convert the threshold tests into binary tests, showing that when $\G$ allows w.h.p. recovery of $\SIGMA$ using binary tests, then so does TGT.



\paragraph{Reduction via Duplication.}

This idea assumes that each item may be included multiple times in each test.
If this is the case, again, we may use any pooling scheme $\G$ that guarantees exact recovery w.h.p. of $k$ defectives using binary tests also for TGT: Obtain $\tilde\G$ from $\G$ by duplicating each item exactly $t$ times in each test in which it appears in $\G$. Again, we effectively convert the threshold tests into binary tests, and we are done, since $\G$ allowed w.h.p. recovery of $\SIGMA$ using binary tests in the first place.

These constructions show that under certain design modifications or assumptions, threshold group testing can inherit both the algorithmic machinery and the performance guarantees of the classical setting.
Therefore it remains for us to question:
\renewcommand{\epigraphflush}{center}
\setlength\epigraphwidth{.9\textwidth}
\setlength\epigraphrule{0pt}

\epigraph{\textit{In \textbf{threshold group testing}, are reduction strategies optimal? Or are there cases where TGT becomes strictly easier? Does the threshold $t$ have any impact on information-theoretic bounds?}}{}

\subsubsection{The Sublinear Regime} \label{SSSec:Contrib}
We first address this question in the sublinear regime in which the label $0$ is asymptotically dominant. More specifically, if $k$ denotes the number of items with label $1$, then we assume that $k = \lfloor n^{\theta}\rfloor$ for $\theta \in (0,1)$.
We also first assume that $k$ is known and that $\SIGMA$ is chosen uniformly at random among all vectors having $k$ defectives. In practical settings, the distributional assumption of uniformity might be realised by a random permutation of the items prior to inference.

This model has been extensively studied for CGT. 
Let 
\begin{align*}
m_{\mathrm{inf}}^{\mathrm{CGT}} := \max\cbc{\frac{1}{\log 2}, \frac{\theta}{(1-\theta)\log^22}} n^{\theta} \log (n/k).
\end{align*}
Then, CGT is known to exhibit the following information-theoretic threshold behaviour:

\begin{theorem}[\cite{coja_spiv}] For any $\theta \in (0,1)$ and  $\varepsilon >0$ there exists $n_0 = n_0(\theta, \varepsilon)$ such that for all $n>n_0$, all pool designs $G$ with $m \leq (1-\eps)m_{\mathrm{inf}}^{\mathrm{CGT}}$ tests and any inference procedure $\cA_G:\{0,1\}^m \to \{0,1\}^n$, 
\begin{align*}
\pr\bc{\cA_G(\hat\SIGMA_G) = \SIGMA} < \varepsilon.
\end{align*}
On the other hand, there exist a randomised pool design $\vec G$ on $m \leq (1+\eps) m_{\mathrm{inf}}^{\mathrm{BGT}}$  pools and a polynomial time algorithm that, given $\vec G$ and the test results $\hat\SIGMA$, outputs $\SIGMA$ w.h.p.
\end{theorem}

 Notably, the same information-theoretic threshold can already be observed by restricting to randomised \emph{constant-column} pooling schemes $\vec{G}=\vec{G}(n,m,\Delta)$, where each item $x_1,\dots,x_n$ \emph{independently chooses $\Delta$ of the $m$ pools uniformly at random with replacement} \cite{aco_2019}. Therefore, this test design provides an important test case also for TGT. We restrict to such test designs in the sublinear regime. Since items draw their tests with replacement, there is ambiguity in the definition of a test outcomes. For an item $x$ and a test $a$, let $n_x(a)$ denote the number of times that $x$ appears in test $a$. We then choose the convention that the outcome of test $a$ is
$$\hat{\SIGMA}_a = \mathds{1}\cbc{\sum_{x \in \partial a} n_x(a)\SIGMA_y \geq t}.$$

Let $H: [0,1] \to \RR, H(x) = -x\log x - (1-x)\log(1-x)$ denote the binary entropy function with base $\eul$. Our main result in the sublinear regime then comes in terms of the following quantities:  For $r>0$, set
\begin{align*}
    c_{1}(r) := \frac{1}{H(\pr\bc{\Po(r)\leq t-1}) }, \qquad c_{2}(r,\theta) := \frac{\theta}{-(1-\theta) r \log\bc{1- \pr\bc{\Po(r)=t-1}}},
\end{align*} and
\begin{align} \label{def_minf}
c_{\mathrm{inf}}^{\mathrm{TGT}} = c_{\mathrm{inf}}^{\mathrm{TGT}} (\theta, t)&= \inf_{r>0}\max\cbc{c_{1}(r),  c_{2}(r,\theta) }.
\end{align}
\nm{In Lemma~\ref{claim_unique_r} below we will show that the function $r \mapsto \max\cbc{c_{1}(r),  c_{2}(r,\theta)}$ has a unique minimiser $r^*>0$, i.e.} 
\begin{align}\label{def_rstar}
    r^* := \arg \nm{\min_{r>0}} \max\cbc{c_{1}(r),  c_{2}(r,\theta)}.
\end{align}

For $t \geq 3$, our main result is conditional, and relies on the following analytic conjecture \nm{involving $r^*$. Below, we}  write $\vX\sim \Po(\lambda)$ when $\vX$ has a Poisson distribution with parameter $\lambda$:

\begin{conjecture} \label{Conj}
For any integer \nm{$t \geq 2$} and $\alpha \in [0,1]$, 
let the random variables $\vE^*$, $\vB^*$, $\vX_1$, and $\vX_2$ be independent with distributions
\begin{equation*}
    \vE^* \sim \Po(r^*), \qquad \vB^* \sim \Po(\alpha r^*), \qquad \vX_1, \vX_2 \sim \Po((1-\alpha)r^*).
\end{equation*}
Define the maximum of $\vX_1$ and $\vX_2$ to be $\vM^* = \max\{\vX_1, \vX_2\}$. Let the probabilities $q(r^*)$ and $p(r^*,\alpha)$ be given by
\begin{equation*}
    q(r^*) = \Pr(\vE^* \leq t-1) \quad \text{and} \quad p(r^*,\alpha) = \Pr(\vB^* + \vM^* \leq t-1).
\end{equation*}
Then, for all $c>c_{\mathrm{inf}}^{\mathrm{TGT}}$, \nm{the function}
\begin{align} \label{Eq:Conj}
    \alpha \mapsto 1 - \alpha + c \bc{q(r^*) \log\bcfr{p(r^*,\alpha)}{q(r^*)} + (1-q(r^*)) \log\bcfr{1-2q(r^*)+p(r^*,\alpha)}{1-q(r^*)}} 
\end{align}
\nm{is negative on $(0,1)$.} 
\end{conjecture}

\nm{Conjecture \ref{Conj} arises from a first moment computation in Section \ref{SSec:UpperSmall}. In Section \ref{App:CondUp}, we reduce it to an only slightly stronger condition that is simpler to state and give some numerical evidence for the case $t=3$ (for $t=2$, we were able to prove that Conjecture \ref{Conj} holds).}

\begin{theorem}\label{Thm_inf}
Suppose that $0<\theta<1$, $k = \lfloor n^{\theta}\rfloor$, $t \in \mathbb{N}^+$ and $\eps>0$, and assume that $k$ is known. 
\begin{enumerate}[\textup(a\textup)]
\item If $m<(1-\eps)c_{\mathrm{inf}}^{\mathrm{TGT}}k\log(n/k)$, then for any constant-column test design $\G=\G(n,m,\Delta)$, there does not exist an algorithm that, given $\G,\hat\SIGMA,k$, outputs $\SIGMA$ with a non-vanishing probability.
\item  Assume that either $t=2$, or that $t \geq 3$ and Conjecture~\ref{Conj} holds. If $m>(1+\eps) c_{\mathrm{inf}}^{\mathrm{TGT}}k\log(n/k)$, then there exist a constant-column test design $\G=\G(n,m,\Delta)$ and an algorithm that, given $\G,\hat\SIGMA, k$, outputs $\SIGMA$ w.h.p.
\end{enumerate}
\end{theorem}

\begin{figure}[H]
    \centering
    \includegraphics[scale=0.37]{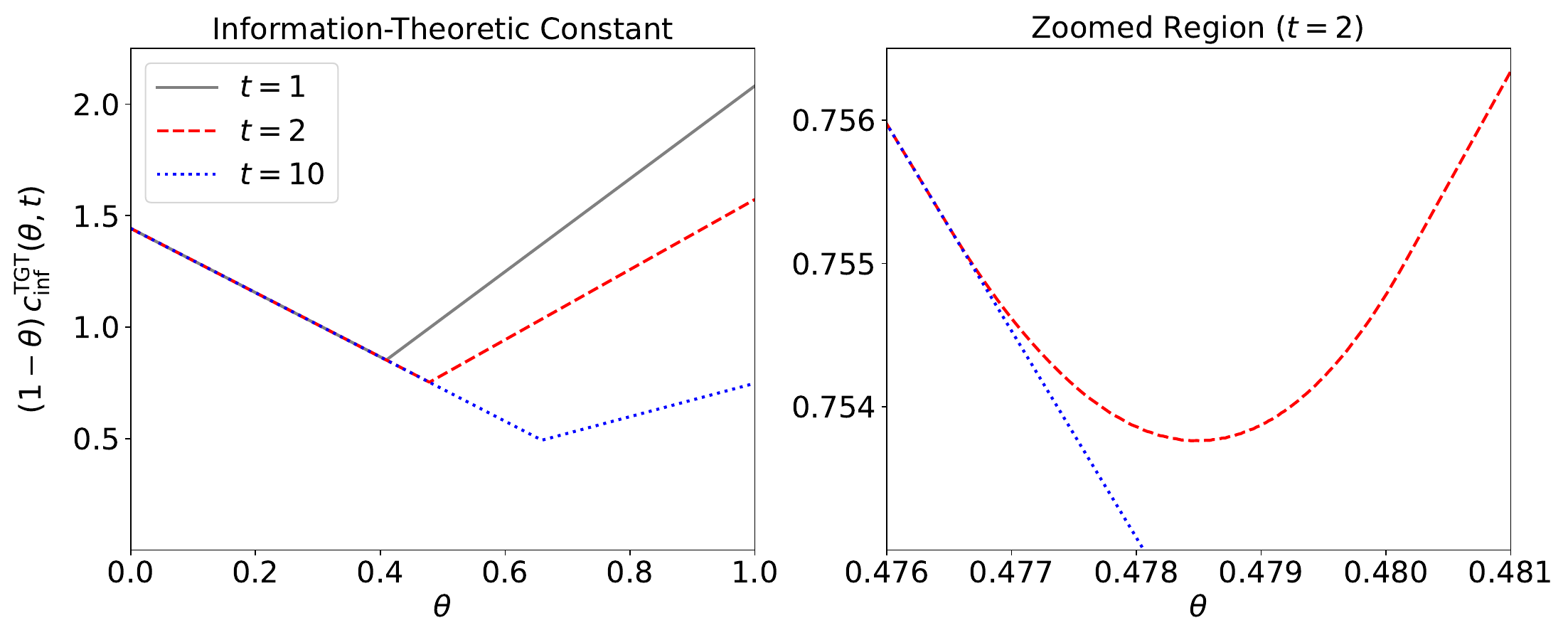}
    \caption{ Left: The difference between the information-theoretically optimal constant for binary group testing (grey, $t=1$) \cite{coja_spiv}, the constant from Theorem \ref{Thm_inf} for $t=2$ (red) and the constant from Theorem \ref{Thm_inf} for $t=10$ (blue). In all cases, the number of tests is of the form $m=(1-\theta) c_{\mathrm{inf}}^{\mathrm{TGT}} (\theta, t) k\log n$, where we used $(1-\theta)^{-1}\log(n/k)=\log{n}$. Right: A zoomed picture of the change-over point when $t=2$.}
    \label{fig:CthetaT}
\end{figure}

Figure \ref{fig:CthetaT} shows the threshold constant $c_{\mathrm{inf}}^{\mathrm{TGT}}$ for $t=1, 2$ and $10$ as functions of $\theta \in (0,1)$.  In particular, the figure displays that for small values of $\theta$, classical and threshold group testing are equally hard in the constant-column design, 
while for larger values of $\theta$, TGT can in fact succeed with \emph{fewer} tests than CGT, since the curve for $t \geq 2$ falls below the classical bound. In this sense, with an eye on our motivating question, Theorem~\ref{Thm_inf} states that in the sublinear regime, for large prevalences $\theta$, reduction strategies to CGT are {\em not} optimal. Moreover, the larger the threshold $t$, the more pronounced the improvement becomes.

Furthermore, Figure \ref{fig:CthetaT}
shows a hidden intricacy of threshold group testing on the constant column design when compared to classical group testing. While to the naked eye it appears that $(1-\theta)c_{\mathrm{inf}}^{\mathrm{TGT}}$ is piecewise linear, in fact, there is an intermediate region where is it not. This is shown on the right of Figure \ref{fig:CthetaT} for $t = 2$.

The underlying reason for this qualitative difference is as follows. Denote by $r_1$ and $r_2$ the minimisers of the two functions in \eqref{def_minf}:
\begin{align}
    r_1 &= r_1(t)=\arg \min_{r > 0} \frac{1}{H(\mathbb{P}(\Po(r) \le t-1))}, \label{Eq:r1} \\
    r_2 &=r_2(t)= \arg \min_{r > 0} \frac{\theta}{-(1-\theta)r \log(1 - \mathbb{P}(\Po(r) = t-1))}. \label{Eq:r2}
\end{align}
Taken individually, $r_1$ and $r_2$ depend only on the threshold $t$ and are independent of $\theta$.  Notably, $r_1=r_2=\log 2$ for $t=1$.
In contrast, $r_1 < r_2$ (see Lemma~\ref{lem_b}) for $t>1$, and there appears to be no closed-form expression for $r_1$ nor $r_2$. 
Approximations of $r_1, r_2$ up to $t=10$ are detailed in the following table:
\begin{table}[H]
\begin{tabular}{|ccccccccccc|}
\hline
$t$ & $1$      & $2$      & $3$      & $4$      & $5$      & $6$      & $7$      & $8$      & $9$      & $10$     \\ \hline
$r_1$ & $\log 2$ & $1.6783$ & $2.6741$ & $3.6721$ & $4.6709$ & $5.6702$ & $6.6696$ & $7.6692$ & $8.6690$ & $9.6687$ \\ \hline
$r_2$ & $\log 2$ & $1.8369$ & $2.8730$ & $3.8923$ & $4.9048$ & $5.9137$ & $6.9206$ & $7.9260$ & $8.9305$ & $9.9342$ \\ \hline
\end{tabular}
\caption{The values of $r_1$ and $r_2$ (rounded to four decimal points) for $t \in\{1,\ldots,10\}$.}
\label{tab:values}
\end{table}

For intervals of $\theta$ where the maximiser $r^*$ of \eqref{def_minf} coincides with one of $r_1, r_2$, the overall curve is linear. However, there is a small `tradeoff' region where the two constraints compete. In Section \ref{SSec:Prelim}, we will establish that \nm{$r^* \in [r_1, r_2]$ (we prove that $r_1 < r_2$ for $t > 1$ in Lemma~\ref{lem_b}).} In this region, as $r$ increases, $c_1$ begins to increase while $c_2$ is still decreasing, which results in the curvature demonstrated in Figure \ref{fig:CthetaT} for $t=2$. The same phenomenon appears in the noisy version of classical group testing \cite{coja2025noisy}.

\subsubsection{The Linear Regime}

We now turn our attention to TGT in the \emph{linear} regime, where the number of defectives $k$ scales linearly in the population size $n$. 
In this section, we assume an i.i.d. prior on $\SIGMA$, where for all $i \in [n]$, $\SIGMA(x_i)=1$ independently with fixed and known probability $\alpha\in(0,1)$. Indeed, in classical non-adaptive group testing, this regime is decidedly different from the sparse regime.
Here, Aldridge \cite{aldridge2018individual} showed that individual testing is optimal:

\begin{theorem}[\cite{aldridge2018individual}] \label{thm_ald_opt}
    Consider non-adaptive binary group testing with an i.i.d. prior where each of the $n$ items is independently defective with a given probability $\alpha \in (0,1)$ which is independent of $n$. Then there exists a constant $\eps > 0$, independent of $n$, such that $\pr\bc{\cA_G(\hat\SIGMA) = \SIGMA} \leq 1-\eps$ for any test design $G$ on $m<n$ tests and any algorithm $\cA_G\colon \{0,1\}^m \to \{0,1\}^n$.
\end{theorem}

Using any of the reduction techniques from TGT to CGT shows that also TGT can also, at least theoretically, be solved with $n$ tests. While we have seen that in the sublinear case, such reductions are not optimal for large values of $\theta$, and the simple constant-column design 
achieves better performance without any reduction approach, the following theorem establishes that in the linear regime, for a large class of reasonable test designs, reduction to the optimal test design of classical group testing yields a considerable improvement:
\begin{theorem}\label{Thm:LinMain}
    Consider non-adaptive threshold group testing with a threshold of $t\geq 2$ and an i.i.d. prior where each of the $n$ items is independently defective with a given probability $\alpha \in (0,1)$ which is independent of $n$. 
    Then, for any constant $M>0$, there exist constants $C>0$ and $\epsilon = \epsilon(\alpha) > 0$, independent of $n$, such that the following is true: Any inference procedure using a non-adaptive test design without multi-edges on at most $Cn\log n$ tests, in which each item participates in at most $\log^4 n$ tests and each test has degree at most $M$, has an average error probability of at least $1-\epsilon$.
\end{theorem}

\begin{remark}
Let $M,C>0$ be fixed. If $\vec G$ is a randomised test design on $m \leq Cn\log n$ tests where each test has degree at most $M$ and chooses its items uniformly without replacement, then \whp, each item participates in at most $\log^4 n$ tests and $\vec G$ satisfies the assumptions of Theorem~\ref{Thm:LinMain} (see Proposition~\ref{prop:bounded_item_degrees}). Thus, Theorem~\ref{Thm:LinMain} covers a natural class of test designs on $Cn\log n$ tests.
\end{remark}

Taken together, our comparison of TGT to CGT reveals the following qualitative difference between the sublinear and the linear regime: In the {\em sparse} setting, TGT initially matches the performance of CGT and, as the prevalence increases, can even outperform it. In contrast, in the \emph{linear regime}, based on our results, TGT appears to be more challenging than CGT, and varying the threshold does not lead to improvements (without resorting to the obvious reductions). \cri{If, indeed, TGT is harder than CGT in the linear regime, a transition must take place in the narrow intermediate regime $k=n^{1-o(1)}$. Characterising this transition more precisely is an interesting direction for future work; for CGT, this regime has been studied by Bay, Price, and Scarlett \cite{bay2022optimal}.}

\subsection{Previous Work} \label{SSec:Prev}

The worst-case version of the threshold group testing problem, in which the goal is to identify the ground truth with probability exactly equal to one, has been studied since the work of Damaschke from 2006\,\cite{damaschke2006threshold}. For 
 constant $t$, Damaschke gives an initial information-theoretic upper bound on the number of tests as $O(k \log n)$ for the adaptive two-round setting \citep{damaschke2006threshold} in the sublinear regime and of $O(n)$ tests for the adaptive two-round setting in the linear regime. \cri{The original work of Damaschke indeed considered the more general problem of TGT \emph{with a gap}: Here, given $\ell \leq u$, tests are negative when there are at most $\ell$ defectives and positive when there are at least $u$ defectives. The test outcomes may be arbitrary when the number of defectives lies strictly between $\ell$ and $u$. This setting has also been studied in \cite{ahlswede2011bounds,bui2021improved,cheraghchi2013improved}.}
 In the sublinear setting \nm{of TGT without a gap}, De Marco, Jurdzi\'nski, Kowalski, R\'o\.za\'nski and Stachowiak  \citep{demarco2020subquadratic} provide an upper bound of order $O(k^2\log(n/k))$
for the non-adaptive case and constant $t$. Neither Damaschke nor De Marco et al. consider the computational complexity. Most recently, Bui et al.~\citep{bui2024efficient, bui2019efficiently} provide the first known efficient algorithms for the problem, covering a range of thresholds and numbers of stages.

In this paper, we focus on the average-case setting where the number of defective items is assumed to be known. Key results include \citep{chan2013stochastic, reisizadeh2018sub}. The first results by Chan et al.~address an error probability of less than~$\epsilon$ \citep{chan2013stochastic}. They propose both a one-round and two-round procedure, using $(4\mathrm{e}^8 \log 2 /\pi^2) \log (1/\epsilon) k \sqrt{t} \log n + O\left( \log (1/\epsilon) k \sqrt{t} \right)$ and $16\mathrm{e}^2 k \log n + O\left( \log (1/\epsilon) k \right)$ tests respectively, with decoding complexity of $O\left( n \log n + n \log (1/\epsilon) \right)$. The results of Reisizadeh et al.~\citep{reisizadeh2018sub} provide a single-round, non-adaptive algorithm requiring $O\left( k \sqrt{t} \log^3 n \right)$ tests, but the decoder requires a precomputed look-up matrix of size $O(t \log n) \times \tbinom{n}{t}$ look-up matrix indexed by all $\tbinom{n}{t}$ possible defective sets. \nm{Finally, Laarhoven \cite{laarhoven2015asymptotics} derived exact thresholds in the very sparse regime where $k=O(1)$.}

\paragraph{\bf Applications of Threshold Group Testing.}
The applications of TGT pertain to the case in which an output only occurs when a threshold of the number of items displaying a certain characteristic is met. One such example is learning interpretable scorecards, which can be used in clinical prediction rules such as estimating the risk of a stroke \cite{gage2001validation, malioutov2017learning}. 
Chen and Fu \cite{chen2009nonadaptive} discuss the equivalence between TGT and graph search problems, where finding a hidden graph $H$ involves taking a collection of edges and determining whether the associated subgraph contains at least $t$ edges of $H$. TGT can also be used to identify triggers for reactive jamming attacks in multi-radio wireless sensors networks \cite{shin2009reactive}. This concept of multiple-packet receptions extends the use of TGT into engineering and practice in energy detectors \cite{chan2012carrier,georgiadis1982collision,ghez1989optimal}.

There exist further generalisations of TGT, the first again being proposed by Damaschke, in which the concentration of the defective items in each test plays an important role \cite{damaschke1997algorithmic}. This problem of \emph{concentration group testing} seems to be very natural, for instance, the standard PCR tests for COVID required a lower limit of 100 copies of viral RNA per millilitre of the sample to provide a positive test result \cite{arnaout2021limit}. Elsewhere, Tsybakov \cite{tsybakov1980resolution} produced a model more closely related to \emph{quantitative group testing} (QGT), in which again a threshold $t$ is imposed, but if we let the exact number of defectives in the pool be $r$, then the output of a test will be $r$ if $0\leq r \leq t$, and $t$ otherwise. This problem is commonly known as \emph{additive multiple channel access} and has been studied previously \cite{censor2015bounded, chan2012carrier,ghez1989optimal}.

\subsection{Discussion and Outlook} \label{SSec:Discussion}

While classical group testing ($t=1$) relies on the fact that any negative test conclusively identifies all items as non-defective, threshold group testing introduces \emph{ambiguity}: an item may consistently be hidden among other defectives in tests. Our results demonstrate that increasing the threshold $t$ fundamentally alters the identifiability landscape. Importantly, Theorem~\ref{Thm_inf} shows that, at least from an information-theoretic perspective, threshold group testing can actually require fewer tests than classical group testing in certain regimes. Here, the achievability result in Theorem~\ref{Thm_inf}(b) holds unconditionally for $t=2$, and for $t \geq 3$ under Conjecture~\ref{Conj}, which ensures that the relevant objective function does not exhibit irregularities that would break the first-moment analysis. The conjecture is strongly supported by numerical evidence. However, as the number of defectives transitions to the linear regime, Theorem~\ref{Thm:LinMain} demonstrates that, in the absence of \nm{multi-edges}, the opposite is true, and more tests seem to be required for threshold group testing.

We emphasise that our conclusions are mainly based upon the restriction to test designs that do not employ multi-edges (even though we theoretically allow multi-edges in the sublinear constant-column design, in the relevant scaling of $m$ and $\Delta$, there will be few of them). This restriction is crucial since allowing \emph{structured} multi-edges can fundamentally alter the problem. For example, as explained in Section~\ref{SSec:Quest}, a simple reduction via duplication can reduce TGT to CGT in all regimes. Similar uses of structured multi-edges also fundamentally change other group testing problems: An example is quantitative group testing, where using such a method can technically allow for a solution \emph{using a single test} (see \cite[Remark 1.3]{hahnklimroth_QGT}).

\subsubsection{Sublinear Regime}

We derive a sharp information-theoretic phase transition at $c_{\mathrm{inf}}^{\mathrm{TGT}}k\log(n/k)$ (non-adaptive) tests 
for TGT within the constant-column test design.
Our analysis quantifies how the required number of tests depends on the population size $n$, the prevalence $\theta$, and the threshold parameter $t$. In order to establish a sharp information-theoretic phase transition in the classical group testing problem, Coja-Oghlan et al.~\cite{aco_2019} drew inspiration from random constraint satisfaction problems. We employ similar techniques to account for the phenomenon of \emph{disguised} defectives, which represents a central challenge unique to the threshold setting.

\paragraph{Proof Strategy.}

The proof strategy of Theorem \ref{Thm_inf} for the large values of $\theta$ follows the approach developed by Coja-Oghlan et al. \cite{aco_2019} for the classical group testing problem. Following careful updates and adaptations of the proof strategy, we have generalised it for any constant $t$. For every result we adapt from \cite{aco_2019}, we cite the corresponding one at the beginning. As \cite{aco_2019}, our approach adapts moment bounds and expansion arguments from the random CSP literature regarding the rigidity of the solutions to manage these subtleties, ensuring uniqueness of the solution with high probability. A crucial element of the analysis is handling the so-called \emph{lottery phenomenon}, where rare configurations allow certain defectives to remain insufficiently constrained \cite{Barriers,Molloy}.

 Theorem~\ref{Thm_inf} is derived for constant-column test designs $\vec G$.
 This limitation parallels the classical group testing setting, where the result for fixed designs was first established in~\cite{aco_2019} and later extended to arbitrary designs in~\cite{coja_spiv}. We aimed to replicate this extension in our threshold setting, and for the most part the necessary adaptations were feasible. However, a key obstacle arises in \cite[Claim~3.11 (equation (3.14))]{coja_spiv}, which relies on the FKG inequality. In TGT, the presence of both increasing and decreasing events relating to disguised items violates the conditions required for applying the FKG inequality. Note that in Section \ref{Sec:LinearK}, when dealing with the linear regime, the issue with the FKG inequality arises again. However, in this setting, we were not aiming for the exact constant and work with restrictions on the item and test degrees, so that a suitable weakening of the appropriate events is sufficient to obtain Theorem~\ref{Thm:LinMain}. 

Nevertheless, we conjecture that our bound in Theorem~\ref{Thm_inf} represents the correct information-theoretic threshold even for arbitrary test designs, since the underlying heuristic argument is that enough tests are needed to eliminate disguised items. Establishing this rigorously remains an open problem.

\paragraph{\textbf{Algorithmic Achievability.}} In subsequent work \cite{2026TGTalg}, which includes the authors of the present paper, we show that the information-theoretic threshold identified here can also be achieved by a polynomial-time algorithm. More precisely, the \texttt{SPOT} algorithm achieves exact recovery w.h.p. using $(1+\eps) c_{\mathrm{inf}}^{\mathrm{TGT}}k\log(n/k)$ tests under a spatially coupled test design, thereby matching our information-theoretic threshold. This result is unconditional for $t$ as it does not require the analytic assumption used in our achievability result. Thus, together, our work and \cite{2026TGTalg} establish a sharp information-theoretic and algorithmic picture for threshold group testing, at least for the corresponding classes of test designs. The open question is whether the lower bound remains tight for arbitrary test designs.

\paragraph{\textbf{Concurrent Results.}} Following the completion of this work, another closely related paper by Tran, McMorrow and Scarlett appeared \cite{tran2026group}, that considers the more general setting of group testing with selectable thresholds, in which the threshold $t_a$ may be chosen separately for each test $a$. Their results include both converse bounds and algorithmic guarantees. When restricted to the fixed-threshold setting considered here, their converse applies to a broader class of test designs than our constant-column construction. In particular, their lower bound agrees with the bound established here, providing further evidence that $c_{\mathrm{inf}}^{\mathrm{TGT}}$ is the correct information-theoretic threshold beyond the constant-column setting. While their results do not fully establish this conjecture for arbitrary designs and all values of $\theta$, they substantially strengthen the evidence that the threshold identified in Theorem \ref{Thm_inf} is intrinsic to the problem rather than an artefact of the constant-column design.

\subsubsection{Linear Regime}
Theorem \ref{Thm:LinMain} establishes that when the number of defectives scales linearly in $n$, $t \geq 2$, and under the restrictions that items are not contained in too many tests, and tests do not contain too many items, the information-theoretic lower bound on tests is at least of order $n \log n$. We conjecture that this is, in fact, the correct result for arbitrary test designs without multi-edges, but have unfortunately been unable to show this.
This highlights a fundamental difference from the sublinear case, where substantial savings over classical group testing are achievable. 

Moreover, in TGT without multi-edges, there is no immediate analogue of individual testing (a test that gives definite information with probability one). Consequently the change from the optimality of group to individual testing that appears in classical group testing does not arise in the same manner for $t\geq2$. \nm{In classical group testing, the change naturally arises when the number of tests required to ensure that each defective item is involved in at least one test that uniquely identifies its status (see Definition~\ref{Def:Piv} for the formal criterion) becomes larger than the cost of individual testing.
} It is therefore natural to conjecture that, if a phase transition occurs in the linear regime for TGT from impossibility to achievability, it
is governed by the requirement that all defectives be covered by pivotal tests. However, due to the previously noted limitations in applying the FKG inequality, we have been unable to rigorously establish such a characterisation.

\subsubsection{Exploration of Conjecture \ref{Conj}}\label{App:CondUp}

In this subsection, we formulate and explore a slight strengthening of Conjecture \ref{Conj}, providing some numerical evidence as to its truth. 
The following proposition first shows that Conjecture~\ref{Conj} is implied by the hold of another, simpler inequality.
\begin{proposition}[Reduction of Conjecture \ref{Conj}] \label{prop:conj_reduction}
To establish Conjecture \ref{Conj} for any integer $t \geq 2$, it suffices to verify that for all $\alpha \in (0,1)$,
\begin{equation} \label{Eq:Conj_reduced}
    q(r^*)^{2-\alpha} > p(\alpha,r^*).
\end{equation}
\end{proposition}

 The proof of all Proposition \ref{prop:conj_reduction} is given in Appendix~\ref{appendix:conj2}.  
In Section~\ref{sec:conj2}, we will prove that \eqref{Eq:Conj_reduced} holds for $t=2$. Generalising this approach to higher values of $t$ becomes substantially more difficult, which is why we only provide numerical evidence for \eqref{Eq:Conj_reduced} to illustrate the case $t=3$.

Similarly to the setting in Conjecture~\ref{Conj}, for $r>0$ and $\alpha \in [0,1]$, let
\begin{align}
    q(r) &= \Pr(\Po(r) \leq t-1) \quad \text{and} \label{var:q} \\
    p(r,\alpha) &= \Pr(\Po(\alpha r) + \max\{\Po_1((1-\alpha)r),\Po_2((1-\alpha)r) \} \leq t-1), \label{var:p}
\end{align}
where the Poisson random variables in the definition of $p(r,\alpha)$ are assumed to be independent. 
As $r^* \in [t-1,t]$ (see Lemmas \ref{lem_a0} to \ref{claim_unique_r}), to check the hold of \eqref{Eq:Conj_reduced}, it is sufficient to see whether for all $\alpha \in (0,1)$ and $r\in[t-1, t]$,
\begin{align} \label{Eq:ToStudy}
q(r)^{2-\alpha} > p(\alpha,r).
\end{align}
When $t=3$, writing out the Poisson probabilities shows that condition \eqref{Eq:ToStudy} is equivalent to the function $h:(1,2)\times [t-1,t]$
\begin{align*}
h(\beta,r) &= \left(1 + r + \frac{r^2}{2}\right)^\beta \\ &- \bc{ \left(1 + (\beta-1)r + \frac{(\beta-1)^2 r^2}{2}\right)^2 + (2-\beta)r(1 + (\beta-1)r)^2 + \frac{(2-\beta)^2 r^2}{2} }
\end{align*}
 being positive. Positivity of $h$ is illustrated in Figure~\ref{fig:struggle}(a), where $\beta$ corresponds to $2-\alpha$. The rest of Figure~\ref{fig:struggle} illustrates the behaviour of the first five derivatives of $h$, highlighting why an analytical argument similar to the proof of Lemma \ref{Lem:T2} is not straightforward.
In each case, $500$ equally spaced sample points were selected in both intervals to generate values of $r$ and $\alpha$.


\begin{figure}[h!]
    \centering
    \includegraphics[scale=0.45]{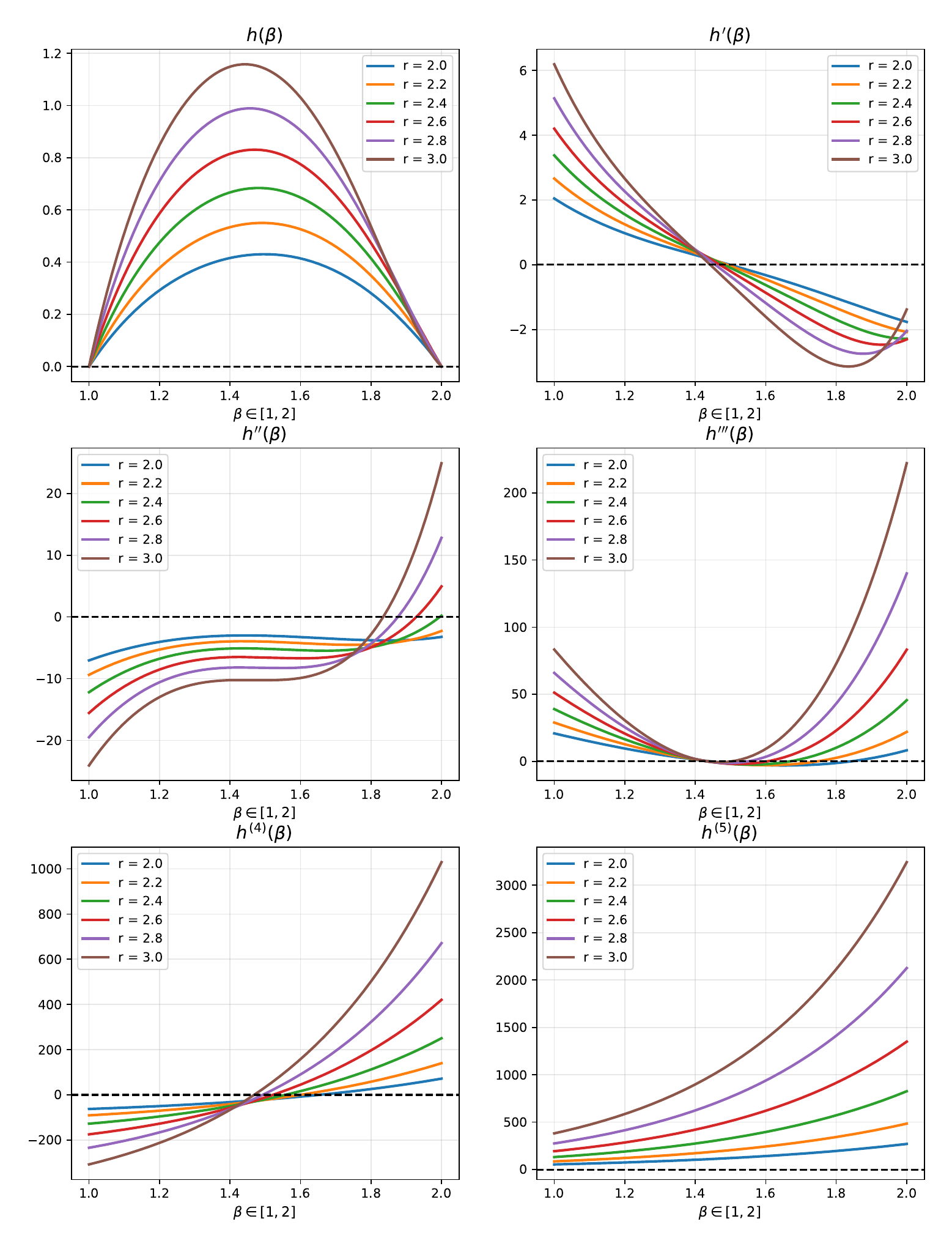}
    \caption{The difference function $h(\beta)$ and its first five derivatives for the $t=3$ case. The plots illustrate the behaviour of the function across the interval $\beta \in [1,2]$ for several selected values of $r \in [2,3]$.}
    \label{fig:struggle}
\end{figure}

For each integer $t \in \{4, \dots, 1000\}$, a similar numerical study as to the hold of \eqref{Eq:ToStudy} was conducted (not shown). The results of this study support our belief that condition \eqref{Eq:ToStudy} holds throughout the specified range.

\subsubsection{Future Directions}
Several avenues remain open for further research, the most immediate being the verification of Conjecture~\ref{Conj} and the proof of a matching information-theoretic lower bound for arbitrary test designs in the sublinear regime to establish a sharp threshold for arbitrary test designs. Extending our bounds to \emph{noisy} threshold group testing models would be of significant practical interest.

A further challenge is to establish the existence of a similar phase transition in the linear regime, if multi-edges are excluded.
Our results point towards the following compelling picture: in sparse regimes, threshold group testing can even outperform the classical setting, whereas in dense regimes, no asymptotic savings are possible. Clarifying the nature of this transition and constructing matching designs remain key directions for future work.

\subsection{Outline} \label{SSec:Outline}

The remainder of the paper is organised as follows. Section~\ref{Sec:Start} introduces recurring concepts underlying our analysis. We begin by discussing the issue of multiple label vectors consistent with the same test outcomes, which we refer to as \emph{competing solutions}, and formalise the notion of uniqueness. We then define \emph{disguised items}, whose labels can be interchanged without changing the test results, and describe the structure of the random pooling scheme used throughout.

The proof of our first main result, Theorem~\ref{Thm_inf}, is presented in two parts, (a) and (b), each supported by two auxiliary results. Part (a), developed in Section~\ref{Sec:Lower}, establishes the information-theoretic lower bound, showing that recovery fails with high probability when the number of tests falls below $c_{\mathrm{inf}}^{\mathrm{TGT}}k\log(n/k)$. 
Section~\ref{Sec:Lower} uses an adaptation of ideas from Coja-Oghlan et al. \cite[Section IV]{aco_2019}.

Part (b), in Section~\ref{Sec:Upper}, establishes our (conditional) information-theoretic upper bound, showing that an item-regular test design on more than $c_{\mathrm{inf}}^{\mathrm{TGT}}k\log(n/k)$ tests guarantees identifiability with high probability. The analysis is divided into a small-overlap and a large-overlap regime, using moment methods (Lemma~\ref{Lem:IntTc}) and expansion arguments (Proposition~\ref{Prop:LrgU}).
The distinction between $t=2$ and $t \geq 3$ arises in the optimisation step within the proof of Lemma~\ref{Lem:IntTc}: the case $t=2$ can be handled directly, whereas we require Conjecture~\ref{Conj} for $t \geq 3$.
Section~\ref{SSec:UpperLarge} specifically contains the adaptation of Section III from Coja-Oghlan et al. \cite{aco_2019}.

Finally, in Section~\ref{Sec:LinearK}, we prove Theorem~\ref{Thm:LinMain} on the linear regime. 
Our proof relies on a genie-based estimator argument similar to that used in~\cite{hintze2024noisy}.

\section{Getting Started} \label{Sec:Start}

The main arguments in the paper revolve around the existence of competing solutions that correspond to the same set of test results. Since our main results are expressed as high-probability statements, we begin in Section \ref{SSec:Prelim} by introducing the necessary asymptotic notation.
We then introduce the notion of item-regular designs in Section \ref{SSec:Pools}. These designs introduce slight dependencies between tests, which we show in Section \ref{SSec:Decoupling} can be decoupled through the use of auxiliary variables.
Next, in Section \ref{SSec:Uniq}, we develop the foundation of our proof strategy, focusing on how any potential algorithm handles competing solutions. To analyse when such competing solutions arise, we first study in Section \ref{SSec:ScalingDelta} the appropriate scaling of the item degree and its impact on the local structure of the test design.
Building on this scaling regime, Section \ref{SSec:Disg} introduces the notion of disguised items. In particular, we show how disguised items give rise to competing solutions, including the simplest case in which two distinct labels may be interchanged without affecting the test outcomes. This allows us to establish the conditions under which disguised items occur, which are formally characterised in Section \ref{SSec:DisguisedConditions}.

\subsection{Preliminaries} \label{SSec:Prelim}
Throughout this work, asymptotic notation is interpreted in the limit as $n \rightarrow \infty$. Specifically, $o(1)$ denotes any term tending to zero, while $\omega(1)$ denotes any term diverging to $\infty$. We write $f(n) = O(g(n))$ to denote that there exist a constant $b > 0$ and an integer $n_0 \in \mathbb{N}$ such that $f(n) \leq b \cdot g(n)$ for all $n \geq n_0$. Similarly, $f(n) = \Omega(g(n))$ indicates that there exist $b > 0$ and $n_0 \in \mathbb{N}$ such that $f(n) \geq b \cdot g(n)$ for all $n \geq n_0$. Finally, $f(n) = \Theta(g(n))$ indicates that $f(n) = O(g(n))$ and $f(n) = \Omega(g(n))$ simultaneously hold. Unless otherwise stated, all functions are assumed to be non-negative. We also use the falling factorial notation $(n)_k = n (n-1) \cdots (n-k+1)$. The Kullback-Leibler divergence of $\delta, p \in (0,1)$ is denoted by
$$D_{\mathrm{KL}}(\delta \,\|\, p) := \delta \log \bc{\frac{\delta}{p}} + (1 - \delta) \log \bc{\frac{1 - \delta}{1 - p}}.$$

We first remark that $r_1, r_2$ from \eqref{Eq:r1}, \eqref{Eq:r2} are well-defined for all $t \in \mathbb{N}$. The proofs of the following \nm{analytical} lemmas are deferred to Appendix~\ref{SSec:PrelimProofs}.

\begin{lemma}\label{lem_a0}
    For every $t \in \mathbb{N}$, the function $r \mapsto c_1(r)$ has a unique minimum $r_1$ on $(0,\infty)$. Moreover, \cri{$r_1 \in \brk{t-1, t-1 + \log 2}$}, $c_1$ is strictly decreasing on $(0,r_1)$ and strictly increasing on $(r_1,\infty)$. 
\end{lemma}

\begin{lemma}\label{lem_a}
    For every \nm{$t \in \NN_{\geq 2}$} and $\theta \in (0,1)$, the function $r \mapsto c_2(r,\theta)$ has a unique minimum $r_2$ on $(0,\infty)$. Moreover, $r_2 \in (t-1, t)$, $r \mapsto c_2(r,\theta)$  is strictly decreasing on $(0,r_2)$ and strictly increasing on $(r_2,\infty)$. 
\end{lemma}

\begin{lemma}\label{lem_b}
    For every $t \in \NN_{\geq 2}$ and $\theta \in (0,1)$, the unique minima $r_1$ and $r_2$ of the functions $r \mapsto c_1(r)$ and $r \mapsto c_2(r,\theta)$, respectively, are such that $r_1<r_2$. 
\end{lemma}

\begin{lemma}\label{claim_unique_r}
    For each $\theta \in (0,1)$, the infimum in \eqref{def_minf} is achieved in exactly one $r^* \in (0,\infty)$.
\end{lemma}

\begin{lemma} \label{claim_opt_eq}
   Let $\theta \in (0,1)$, $d >0$ and set $c_1^*(d) = \inf\{c>0: c \geq c_1(d/c)\}$ and $c_2^*(d,\theta) = \inf\{c>0: c \geq c_2(d/c,\theta)\}$. Then
   \begin{align}
      c_{\mathrm{inf}}^{\mathrm{TGT}} = \inf_{r>0}\max\cbc{c_{1}(r),  c_{2}(r,\theta) } =  \inf_{d>0}\max\cbc{c_{1}^*(d),  c_{2}^*(d,\theta) } =: c_{\mathrm{inf}}^{\mathrm{TGT}, *}.
   \end{align}
\end{lemma}

\subsection{Constant Column Test Design} \label{SSec:Pools}

We next define our pooling scheme $\vec{G}$, which can be represented as a random bipartite graph (cf. Section \ref{SSSec:Contrib}). Recall that we denote the constant column design on $n$ items and $m$ tests, where every item independently chooses $\Delta$ tests uniformly at random \textit{with replacement}, by $\vec{G}=\vec{G}(n,m,\Delta)$. Note that $\vec{G}$ may contain multi-edges.

For $x_i \in V$, we define $\partial x_i$ as the multiset of tests to which $x_i$ is assigned, so that $\abs{\partial x_i} = \Delta$.
Similarly, for each test $a_j$, let $\partial a_j$ denote the multiset of items included in that test, setting $\vGamma_j := \abs{\partial a_j}$. 
We represent the vector of test degrees by $\vGamma = \bc{\vGamma_j}_{j \in [m]}$, where $\vGamma_j\sim \Bin(n\Delta, 1/m)$ and  $\sum_{j=1}^m \vGamma_j = n \Delta$.

We introduce positive constants $c, d > 0$ defined by
\begin{align}\label{Eq:Param}
m = \lceil c k \log(n/k)\rceil, \quad \quad \quad \quad \quad \quad \Delta = \lceil d \log(n/k)\rceil.
\end{align}
While $m = c k \log(n/k)$ is the natural scaling of the information-theoretically optimal number of tests, in Lemma~\ref{Lemma_Delta_too_small_big}, we will justify that  $\Delta = \Theta(\log n)$ is the corresponding appropriate scaling of the item degrees for successful inference. 

Finally, let $\vGamma_{\min} = \min_{j \in [m]} \vGamma_j$ and $\vGamma_{\max} = \max_{j \in [m]} \vGamma_j$ capture the range of test degrees. The Chernoff bound for the Binomial distribution implies the following:

\begin{lemma}[{\cite[Lemma II.4]{aco_2019}}] \label{Lemma_GammaMinMax}
    Assume $m = ck \log(n/k)$ and $\Delta = d \log(n/k)$ for constant $c,d>0$. With probability at least $1-o(n^{-2})$, as $n \rightarrow \infty$, 
    \begin{align}\label{eq_minmax_G}
        \frac{\Delta n}{m}-\sqrt{\frac{\Delta n}{m}}\log n \leq \vGamma_{\min} \leq \vGamma_{\max} \leq \frac{\Delta n}{m}+\sqrt{\frac{\Delta n}{m}}\log n.
    \end{align} 
\end{lemma}

\begin{remark} \label{rmk:Gamma_omega}
    Since they only rely on the Chernoff bound for the Binomial distribution, the bounds in \eqref{eq_minmax_G} in fact hold true for more general choices of $m,\Delta$.
\end{remark}



\subsection{Decoupling Test Dependencies} \label{SSec:Decoupling}

In the constant column design, the numbers of defective items across different tests are not independent.
To deal with the resulting correlations, we follow the conditioning approach of  \cite[Appendix B]{aco_2019}. 
More precisely, let $\vY_j$ denote the number of edges in test $a_j$ that are incident to defective items. Given the test degree sequence  $\vGamma = \bc{\vGamma_j}_{j \in [m]}$, we define $(\vX_j)_{j \in [m]}$ to be a sequence of \emph{independent} random variables, where $\vX_j\sim\Bin(\vGamma_j, k/n)$. In addition, introduce the event
\begin{align*}
	\cE&=\cbc{\sum_{j\in[m]}\vX_j=k\Delta}.
\end{align*}
\nm{Given any valid degree sequence $\vGamma$ with $\sum_{i=1}^m \vGamma_i = n\Delta$, $\sum_{j\in[m]}\vX_j$ is distributed as a sum of $n\Delta$ i.i.d. $\Be(k/n)$ random variables. In particular, $\Erw[\sum_{j\in[m]}\vX_j] = k\Delta$. The de Moivre-Laplace Theorem (see \cite[Theorem 3.5.3]{durrett2019probability} for a more general local limit theorem) thus implies that  }
\begin{align}\label{eqEprob}
	\pr\bc{\cE \vert \vGamma} = \pr\bc{\Bin\bc{n\Delta,k/n}=k\Delta} =\Omega(1/\sqrt{\Delta k}).
\end{align}
The following lemma shows that, conditionally on $\vGamma$ and $\cE$, we can work with the $\vX_j$'s instead of the $\vY_j$'s:

\begin{lemma}[{\cite[Lemma B.2]{aco_2019}}]\label{Lemma_Elemma}
    Let $(\vY_1, \dots, \vY_m)$ and $(\vX_1, \dots, \vX_m)$ be as defined above. Then, for any integer sequence $(y_j)_{j\in[m]}$ with $0\leq y_j \leq \vGamma_j$ and $\sum_{j\in[m]}y_j=k\Delta$, 
$$
\Pr\Big(\forall j \in [m]: \vY_j = y_j \;\big|\; \vGamma\Big)
=
\Pr\Big(\forall j \in [m]: \vX_j = y_j \;\big|\; \vGamma, \mathcal{E}\Big).
$$
\end{lemma}

Finally, using the random variables $(\vY_j)_{j \in [m]}$, for $\ell \in [k]$, we let
\begin{align}\label{def_ms}
		\vm_{\ell}&=\sum_{i=1}^m\vecone\cbc{\vY_i=\ell}, \quad \bfm_{<t}=\sum_{i=1}^m\vecone\cbc{\vY_i< t}, \quad \bfm_{\geq t}=\sum_{i=1}^m\vecone\cbc{\vY_i\geq t}
	\end{align}
be the numbers of tests with exactly $\ell$, less than $t$ and more than $t-1$ defective items, respectively.
Observe that $m=\bfm_{<t} + \bfm_{\geq t}$.
 Using the shorthand $[(1 \pm \varepsilon)\mu] := [(1-\varepsilon)\mu,\,(1+\varepsilon)\mu]$, the next lemma pins down the likely ranges of various types of tests in our test design. 

\begin{lemma}[Generalised from {\cite[Lemma II.5]{aco_2019}}]\label{Lemma_m0}
    Assume $m = ck \log(n/k)$ and $\Delta = d \log(n/k)$ for some constants $c,d>0$. With probability at least $1-o(n^{-7})$,
    \begin{enumerate}
        \item $\vm_{t-1}\in \brk{\bc{1\pm n^{-\Omega(1)}}m \frac{\bc{d/c}^{t-1}}{(t-1)!}\exp(-d/c)}$;
        \item$\vm_{t}\in \brk{\bc{1\pm n^{-\Omega(1)}}m \frac{\bc{d/c}^{t}}{t!}\exp(-d/c)}$;
        \item$\vm_{<t}\in \brk{\bc{1\pm n^{-\Omega(1)}}m \bc{\sum_{j=0}^{t-1}\frac{\bc{d/c}^{j}}{j!}\exp(-d/c)}}$; and
        \item$\vm_{\geq t}\in \brk{\bc{1\pm n^{-\Omega(1)}}m \bc{1-\sum_{j=0}^{t-1} \frac{\bc{d/c}^{j}}{j!}\exp(-d/c)}}$.
    \end{enumerate}
\end{lemma}

\begin{proof}


    For $\ell \in \{t-1, t\}$, let $\vm_{\ell}'$ be an auxiliary random variable estimating $\vm_{t-1}$ based on $\bc{\vX_i}$ rather than on $\bc{\vY_i}$, i.e.,
	\begin{align*}
		\vm_{\ell}'&=\sum_{i=1}^m\vecone\cbc{\vX_i=\ell}.
	\end{align*}
    Conditioning on $\vGamma$, we can bound the expected value of $\vm_{\ell}'$ in terms of $\vGamma_{\min}$ and $\vGamma_{\max}$. 
    For $\vGamma$ satisfying \eqref{eq_minmax_G}, 
	\begin{align}
		\Erw[\vm_{\ell}'\mid \vGamma]&=\sum_{i=1}^m \binom{\vGamma_i}{\ell} \bcfr{k}{n}^{\ell}\bc{1-k/n}^{\vGamma_i-\ell} \leq m\binom{\vGamma_{\max}}{\ell}\bcfr{k}{n}^{\ell}\bc{1-k/n}^{\vGamma_{\min}-\ell} \nonumber\\
        &= \bc{1+n^{-\Omega(1)} } m \frac{\bc{d/c}^{\ell}}{\ell!}\exp(-d/c). \label{Eq:ExpW1U}	\end{align}
	Analogously,
	\begin{align}
		\Erw[\vm_{\ell}'\mid \vGamma]&\geq m\binom{\vGamma_{\min}}{\ell}\bcfr{k}{n}^{\ell}\bc{1-k/n}^{\vGamma_{\max}} =\bc{1+n^{-\Omega(1)}} m \frac{\bc{d/c}^{\ell}}{\ell!}\exp(-d/c). \label{Eq:ExpW1L}
	\end{align}
    We conclude from \eqref{Eq:ExpW1U} and \eqref{Eq:ExpW1L} that for all $\vGamma$ that satisfy \eqref{eq_minmax_G},
$$
\Erw[\vm_{\ell}'\mid\vGamma]
= \bc{1+n^{-\Omega(1)}} m \frac{\bc{d/c}^{\ell}}{\ell!}\exp(-d/c).
$$
        The Chernoff bound on the Poisson-binomial random variable found in Theorem \ref{lem_chernoff_pb} for $\vm_{\ell}'$ then gives, again for $\vGamma$ satisfying \eqref{eq_minmax_G},
	\begin{align}\label{Eq:PrW1}
	    \pr\bc{\left.\vm_{\ell}'\in \brk{\bc{1\pm n^{-\Omega(1)}} m \frac{\bc{d/c}^{\ell}}{\ell!} \exp(-d/c)}\right|\vGamma}=1-o(n^{-9}).
	\end{align}
    We now use the fact that $\vm_{\ell}' $ given $ \vGamma$ and $\cE$ is equal in distribution to $\vm_{\ell}$ given $\vGamma$, to study the probability that $\vm_{\ell}$ is close to its expectation:
    \begin{align} 
        \pr&\bc{\vm_{\ell}\in \left. \brk{\bc{1\pm n^{-\Omega(1)}}m \frac{\bc{d/c}^{\ell}}{\ell!}\exp(-d/c)}\right|\vGamma} \label{Eq:CloseExp} \\
        &= \pr\bc{\vm_{\ell}'\in \left. \brk{\bc{1\pm n^{-\Omega(1)}}m \frac{\bc{d/c}^{\ell}}{\ell!}\exp(-d/c)}\right|\vGamma, \cE} \notag \\
        &= 1- \frac{\pr\bc{\vm_{\ell}'\not\in \left.\brk{\bc{1 \pm n^{-\Omega(1)}}m \frac{\bc{d/c}^{\ell}}{\ell!}\exp(-d/c)}, \cE\right| \vGamma}}{\Pr\bc{\cE \mid \vGamma}} \notag \\
        & \geq 1- \frac{\pr\bc{\vm_{\ell}'\not\in \left. \brk{\bc{1 \pm n^{-\Omega(1)}}m \frac{\bc{d/c}^{\ell}}{\ell!}\exp(-d/c)}\right|\vGamma}}{\Pr\bc{\cE \mid \vGamma}} = 1 -o\bc{n^{-7}}, \notag
    \end{align}
    which follows from \eqref{Eq:PrW1} and \eqref{eqEprob}.  The proofs for $\bfm_{<t}$ and $\bfm_{\geq t}$ are analogous.
\end{proof}

\subsection{Uniqueness and Recovery} \label{SSec:Uniq}
We study exact recovery of $\SIGMA$ from the measurement $\hat{\SIGMA}$ and the pooling scheme $\mathbf G$.  The key quantity governing recovery is the number of vectors consistent with the observations: 
Let $S_k(\mathbf{G},\hat{\SIGMA})$ be the set of all vectors $\sigma$ with Hamming weight $k$ that give rise to the test result $\hat{\SIGMA}$ on the pooling scheme $\mathbf{G}$. Further, denote $Z_k(\mathbf{G},\hat{\SIGMA}) = \abs{S_k(\mathbf{G},\hat{\SIGMA})}$. Proposition \ref{Prop:Posterior} shows that the posterior distribution of $\SIGMA$ given $\mathbf{G}$ and $\hat{\SIGMA}$ is the uniform distribution on $S_k(\mathbf{G},\hat{\SIGMA})$. The following proposition and subsequent corollary are direct generalisations of Proposition~II.1 and Corollary~II.2, respectively, in \citep{aco_2019}.

\begin{proposition}\label{Prop:Posterior}
    For all $\tau \in \{0,1\}^{\{x_1, \dots, x_n\}}$,
    \[ \Pr\bc{\SIGMA = \tau \mid \mathbf{G},\hat{\SIGMA}} = \frac{\vecone\{ \tau \in  S_k (\mathbf{G},\hat{\SIGMA})\}}{Z_k (\mathbf{G},\hat{\SIGMA})}.\]
\end{proposition}

As a consequence, the information-theoretically optimal inference algorithm cannot do better than outputting a uniform sample of $S_k (\mathbf{G},\hat{\SIGMA})$, and we obtain the following corollary.

\begin{corollary}\label{Cor_Nishi}
    \begin{enumerate}[(1)]
        \item If $Z_k (\mathbf{G},\hat{\SIGMA})= \omega(1) $ w.h.p., then for any algorithm $\cA$,
        \[ \pr\bc{\cA(\mathbf{G},\hat{\SIGMA},k)=\SIGMA}=o(1). \]
        \item If $Z_k (\mathbf{G},\hat{\SIGMA})= 1 $ w.h.p., then there is an algorithm $\cA$ such that
        \[ \pr\bc{\cA(\mathbf{G},\hat{\SIGMA},k)=\SIGMA}=1-o(1). \]
    \end{enumerate}
\end{corollary}
\subsection{Scaling of $\Delta$} \label{SSec:ScalingDelta}

In this section, we analyse how the choice of the design parameter $\Delta$ affects the balance between positive and negative tests.
Heuristically, having items appearing in too few tests yields mostly negative outcomes, and having items appearing in too many tests yields mostly positive outcomes; in either scenario, the observed test results are uninformative, preventing recovery of the ground truth.  Recall the random variable $Z_k(\mathbf{G},\hat{\SIGMA})$ from the previous section.
We now justify that for $m = ck \log \frac{n}{k}$, $c>0$, only constant-order choices of $d$ may lead to $Z_k(\mathbf{G},\hat{\SIGMA}) \not= \omega(1)$ w.h.p..
The proof of the following lemma is similar to the proof of \cite[Lemma II.6]{aco_2019}.

\begin{lemma}
\label{Lemma_Delta_too_small_big}
    \begin{enumerate}
        \item If $m=c k \log(n/k)$ and  $\Delta = o\bc{\log(n/k)}$, then $Z_k(\mathbf{G},\hat{\SIGMA}) = \omega(1)$ w.h.p.
        \item If $m=ck \log(n/k)$ and $\Delta = \omega\bc{\log(n/k)}$, then $Z_k(\mathbf{G},\hat{\SIGMA}) = \omega(1)$ w.h.p.
    \end{enumerate}
\end{lemma}

\begin{proof}
We first consider the case $\Delta = \omega(\log(n/k))$ and use a similar argument as in Lemma \ref{Lemma_m0} to study the number of negative tests $\bfm_{<t}$. As before, we introduce the
auxiliary quantity
$$\bfm_{< t}' = \sum_{i=1}^m \vecone\{\vX_i \leq t-1\}.$$
Then
$$\Erw[\bfm_{< t}' \mid \vGamma] = \sum_{i=1}^m \Pr\bc{\Bin(\vGamma_i,k/n) \leq t-1 \mid \vGamma}.$$
Moreover, for $\vGamma$ satisfying \eqref{eq_minmax_G},
$$\Erw\brk{\Bin(\vGamma_i,k/n) \mid \vGamma} = \vGamma_i k/n = \omega(1).$$
The Chernoff bound  now further implies that for $n$ large enough,
\begin{align}
    \Pr\bc{\Bin(\vGamma_i,k/n) \leq t-1 \mid \vGamma}
    \leq \exp\bc{-\frac{\vGamma_i k}{8n}}.
\end{align}
Summing over $i$, and still assuming that $\vGamma$ satisfies \eqref{eq_minmax_G},
$$\Erw[\bfm_{< t}' \mid \vGamma] = \sum_{i=1}^m \Pr\bc{\Bin(\vGamma_i,k/n) \leq t-1\mid \vGamma}
\leq m \exp\bc{-\tfrac{k}{8n}\min_i \vGamma_i}= o(m).$$
An application of the Chernoff bound shows that $\bfm_{< t}'$ is concentrated around its mean. Again, the transition from $\bfm_{< t}'$ to $\bfm_{<t}$ follows the same argument as \eqref{Eq:CloseExp} in Lemma \ref{Lemma_m0}, using that $\pr\bc{\cE\vert\vGamma}^{-1} = O\bc{\sqrt{\Delta k}} = O\bc{n^{1+\theta}}$ \nm{for any reasonable choice of item degrees $\omega(\log(n/k)) = \Delta \leq n$}. Finally, by Remark \ref{rmk:Gamma_omega}, $\vGamma$ satisfies \eqref{eq_minmax_G} with sufficiently high probability, and we obtain that
\[
\bfm_{<t} = o(m)
\]
with probability $1 - o(n^{-2})$.
Therefore w.h.p., almost all tests are positive. The rest of the proof now follows as in \cite[Lemma II.6]{aco_2019} by observing that w.h.p. over the choice of $\G$, asymptotically most configurations $\sigma \in \{0,1\}^n$ give rise to at most $o(m)$ negative tests. Going through all such `bad' configurations, there will be at most 
$n\binom{m}{o(m)}= o(\binom{n}{k})$ ones that share their test outcomes with at most $n-1$ other configurations. Consequently, $Z_k(\G,\hat\SIGMA)\geq n$ w.h.p. over the choice of $\G$ and $\SIGMA$.

The case $\Delta=o(\log(n/k))$ on the other hand is an immediate consequence of \cite[Lemma II.6]{aco_2019}, which states that $\bfm_{\geq 1} = o(m)$. Since $\bfm_{\geq t}\leq \bfm_{\geq 1}$, also $\bfm_{\geq t} = o(m)$. Interchanging the positive and negative tests in the previous counting argument gives the claim for $\Delta=o(\log(n/k))$.
\end{proof}

\subsection{Disguised Items in the Constant Column Design} \label{SSec:Disg}
Our \nm{converse} arguments revolve around the notion of \emph{disguised items} - those items whose labels, given $\vec{G}, \hat{\SIGMA}$, can be changed without changing $\hat{\SIGMA}$. We now provide the corresponding definitions. 

\begin{definition}[Pivotal Tests] \label{Def:Piv}
Let $ a_j $ be a test and \( x_i \in \partial a_j \). We say that \( a_j \) is a \emph{pivotal test for item \( x_i \)} if

$$\sum_{y \in \partial a_j \setminus \{x_i\}} \SIGMA_y = t - 1.$$
Here, the notation $\partial a_j \setminus \{x_i\}$ refers to the removal of all occurrences of item $x_i$ from the multiset $\partial a_j$.

\end{definition}
This notion accounts for the fact that changing the label of $x_i$ would change the results $\hat{\SIGMA}_{j}$ of all its pivotal tests. 

\begin{definition}[Disguised Items] \label{Def:Disguised}
    We say that item $x_i$ is \emph{disguised} if  $x_i$ appears exactly once in each $a_j \in \partial x_i$ and
    there exists no $a_j \in \partial x_i$ which is pivotal for $x_i$. The set of all disguised items is denoted by $\vV^+$.
\end{definition}

\begin{remark}
For simplicity of analysis, in Definition \ref{Def:Disguised}, we disregard items that appear more than once within a test. Consequently, our notion of disguised items is a   \emph{lower bound} on the actual number of items whose labels could be changed without affecting the test results. 
\end{remark}

\subsection{Conditions Under Which Items Are Disguised} \label{SSec:DisguisedConditions}
 We continue with a study on the effect of $c$ and $d$ on the number of defective and non-defective disguised items. 
Recall that $\vV^+$ denotes the set of all disguised items, and further define
\begin{align}
\vV_0^+ = \{ x_i \in \vV^+ \mid \SIGMA_{x_i}=0 \} \quad \quad \textrm{and}\quad \quad \vV_1^+ = \vV^+ \setminus \vV_0^+ .
\end{align}
\nm{Moreover, recall $r^*$ from \eqref{def_rstar} and Claim \ref{claim_unique_r}. }
The following proposition is similar in spirit to \cite[Proposition II.3]{aco_2019}.

\begin{proposition} 
\label{Lemma_V_all}
    Let $c,d = \Theta(1)$. Then, the following statements hold w.h.p.:
    \begin{enumerate}[\textup{(}i\textup{)}]
        \item If $k\bc{1-(d/c)^{t-1}\exp(-d/c)/(t-1)!}^{\Delta} \geq n^{\Omega(1)}$, then $\abs{\vV_1^+} = n^{\Omega(1)}$.
        \item If $k\bc{1-(d/c)^{t-1}\exp(-d/c)/(t-1)!}^{\Delta} =o(1)$, \hspace{1pt} then $\abs{\vV_1^+} = o(1)$.
        \item If $k\bc{1-(d/c)^{t-1}\exp(-d/c)/(t-1)!}^{\Delta} \geq n^{\Omega(1)}$, then $\abs{\vV_0^+} = n^{\Omega(1)}$. 
    \end{enumerate}
    \begin{enumerate}[\textup{(}i\textup{)}]\addtocounter{enumi}{3}
        \item Fix $d>0$ and let $c_2^*(d,\theta) = \inf\{c>0: c \geq c_2(d/c,\theta)\}$. 
        If $c<c_2^*(d,\theta)$, 
              then $\abs{\vV_0^+}, \abs{\vV_1^+} = n^{\Omega(1)}$. \label{Lemma_V++}
        \item Suppose that $c > c_{\mathrm{inf}}^{\mathrm{TGT}}(n,\theta,t)$ and $d = r^*(\theta) c$. Then $\abs{\vV_1^+} = o(1)$.
    \end{enumerate}
\end{proposition}

\begin{proof}
    
    \textbf{Proof of (i):\\}  
Recall from Definition~\ref{Def:Disguised} that item $x_i$ is \emph{disguised} if it is included exactly once in each of its incident tests $a_j \in \partial x_i$ and none of them is pivotal for $x_i$. 
To estimate $|\vV_1^+|$, we first work without the multiplicity restriction. 
Let
\begin{align*}
    \vU_1 = \sum_{i=1}^n \mathds{1}\cbc{\SIGMA_i=1, \forall a \in \partial x_i: \vY_a \not=t} = \sum_{x \in \vV_1} \vecone\left\{ \sum_{a \in \partial x} \vecone \{\vY_a =t\} =0\right\}
\end{align*}
be the number of defective items that do not participate in any test with exactly $t$ defective items.
Since every disguised defective item is counted by the corresponding indicator in $\vU_1$, 
$\vU_1$ provides an upper bound on the number
of disguised defectives.
 
In the present test design, every item has fixed degree $\Delta=d \log (n/k)$. 
Given $\vGamma$ and $\vm_t$, for an item to be in $\vU_1$, it must choose its $\Delta$ edges from the other $k\Delta-t \vm_t$ which are not incident to a pivotal test. 
By Lemma~\ref{Lemma_m0}, with probability $1-o(n^{-7})$,
	\color{black}
	\begin{align}
		\Erw[\vU_1\mid\vGamma,\vm_t]&=k\binom{k\Delta-t \vm_t}\Delta\binom{k\Delta}\Delta^{-1} = k \frac{(k\Delta-t \vm_t)_{{\Delta}}}{(k\Delta)_{{\Delta}}} \nonumber \\ 
		&=\bc{1+n^{-\Omega(1)}}k(1-t \vm_t/k\Delta)^\Delta \nonumber \\ 
        &=\bc{1+n^{-\Omega(1)}}k\bc{1-\frac{(d/c)^{t-1}}{(t-1)!}\exp(-d/c)}^\Delta. \label{eq_exp_U1}
	\end{align}
	\color{black}
	Similarly, w.h.p.
	\begin{align*}
		\Erw[\vU_1^2\mid\vGamma,\vm_t]&=k\binom{k\Delta-t \vm_t}\Delta\binom{k\Delta}\Delta^{-1}+ k(k-1)\binom{k\Delta-t \vm_t}{2\Delta}\binom{k\Delta}{2\Delta}^{-1}\\
        &=\Erw[\vU_1\mid\vGamma,\vm_t]+\bc{1+n^{-\Omega(1)}}\Erw[\vU_1\mid\vGamma,\vm_t]^2.
	\end{align*}
In the present case, we assume that $k\bc{1-(d/c)^{t-1}\exp(-d/c)/(t-1)!}^{\Delta} \geq n^{\Omega(1)}$.	Therefore, by Chebyshev's inequality, \whp{}
	\begin{align}\label{eqLemma_V1plus_1}
		\vU_1&\in \brk{\bc{1\pm n^{-\Omega(1)} 
        }k\bc{1-\frac{1}{(t-1)!} \bcfr{d}{c}^{t-1}\exp(-d/c)}^\Delta}.
	\end{align}	
Thus, \eqref{eqLemma_V1plus_1} implies there exists $\kappa>0$ such that $\vU_1 \geq n^\kappa$ \whp\

    Let now $\vR_1$ be the number of defective items that are contained in at least one test with multiplicity greater than one such that $\vU_1 \geq \abs{\vV_1^+} \geq \vU_1-\vR_1$. 
    The probability that an item appears in a specific test at least twice is upper-bounded by $(\Delta/m)^2$. 
    Taking the union bound over all defective items and all tests yields
    $$\Erw[\vR_1]\leq k m \bc{\frac{\Delta}{m}}^2 = O(\log n).$$
By Markov’s inequality,
$$\Pr\bc{\vR_1 \geq \tfrac12 n^\kappa} \leq \frac{2\,\Erw[\vR_1]}{n^\kappa} = o(1). $$
Hence, w.h.p. we simultaneously have
$\vU_1 \geq n^\kappa$ and $\vR_1 \le \tfrac12 n^\kappa$, and therefore
$$\abs{\vV_1^+} \geq \vU_1 - \vR_1 \geq \tfrac12 n^\kappa =n^{\Omega(1)} \quad \text{\whp}$$

    \noindent \textbf{Proof of (ii):\\}
    Let $\vU_1$ be as in
    part (i) such that $\vU_1\geq\abs{\vV_1^+}$. If $k\bc{1-\tfrac{(d/c)^{t-1}}{(t-1)!} \exp(-d/c)}^\Delta=o(1)$, $\abs{\vV_1^+}=o(1)$ \whp~ follows from \eqref{eq_exp_U1} and Markov's inequality. 

\noindent    \textbf{Proof of (iii):\\}
    The proof of (iii) follows an almost identical argument to the proof of (i), but with the adaptation that we first study the tests with exactly $t-1$ defective items. Therefore, we relegate this proof to Appendix \ref{app:Condition3}.

\noindent \textbf{Proof of (iv):  \\}
Fix $d>0$. Then 
\begin{align*}
    k\bc{1-(d/c)^{t-1}\exp(-d/c)/(t-1)!}^{\Delta} = n^{\theta + (1-\theta)d\log\bc{(1 - \exp(-d/c) \frac{(d/c)^{t-1}}{(t-1)!} }}.
\end{align*}
We will show that for $c<c_2^*(d,\theta)$ the exponent is positive, whereby the claim follows from parts (i) and (iii). For this, observe that
\begin{align*}
    \theta + (1-\theta)d\log\bc{(1 - \exp(-d/c) \frac{(d/c)^{t-1}}{(t-1)!} } >0 \quad \Longleftrightarrow \quad c < c_2(d/c,\theta).
\end{align*}
By definition of $c_2^*(d,\theta)$, for $c<c_2^*(d,\theta)$, the right hand side is satisfied, so that the claim follows.

\noindent \textbf{Proof of (v):  \\} Suppose that $c > c_{\mathrm{inf}}^{\mathrm{TGT}}(n,\theta,t)$ and $d = r^*(\theta) c$. Then 
\[
k \bc{1 - \frac{\bc{\frac{d}{c}}^{t-1} \exp(-d/c)}{(t-1)!}}^{ d \log(n/k)} = o(1)
\]
is implied by $\theta + (1-\theta)r^*c\log(1-\pr\bc{\Po(r^*) = t-1}) <0$, which is satisfied since $c > c_{\mathrm{inf}}^{\mathrm{TGT}}$.
By part (ii), we conclude that $\abs{\vV_1^+} = o(1)$. \qedhere

\end{proof}

\section{Information-Theoretic Lower Bound} \label{Sec:Lower}
We first outline the derivation of \Thm~\ref{Thm_inf}(a), 
asserting that \whp{} the ground truth $\SIGMA$ cannot be uniquely identified when $m < (1 - \eps)c_{\mathrm{inf}}^{\mathrm{TGT}}k\log\frac{n}{k}$. 

\begin{proposition}[\Thm~\ref{Thm_inf}(a)]
\label{Prop_IT_lower}
If $m<(1-\eps)c_{\mathrm{inf}}^{\mathrm{TGT}}k\log\frac{n}{k}$, then there does not exist any algorithm that, given $\G,\hat\SIGMA,k$, outputs $\SIGMA$ with a non-vanishing probability.
\end{proposition}

The form \eqref{def_minf} of $c_{\mathrm{inf}}^{\mathrm{TGT}}$ as a maximum over two different expressions already hints at two different mechanisms being at work. In this sense, we establish the information-theoretic lower bound by two different considerations.

Section \ref{SSec:LowerCounting} employs a rather classical information-theoretic argument \cite{AJS_book}, but adapted to the constant-column design.
Section~\ref{SSec:LowerDisguised} on the other hand builds on ideas of Coja-Oghlan et al.~\cite{aco_2019}: 
We use Proposition \ref{Lemma_V_all} to show that if the number of tests $m$ is less than $c_{\mathrm{inf}}^{\mathrm{TGT}}k\log\frac{n}{k}$, both disguised defectives and disguised non-defectives 
occur abundantly \whp{} We will use these items to construct assignments that are consistent with the test results, but different from $\SIGMA$. 
The two parts together yield our information-theoretic lower bound,
as concluded in  Section~\ref{SSec:Proof_Prop_IT_lower}.

\subsection{Information-Theoretic Lower Bound for Constant-Column Design} \label{SSec:LowerCounting}

We begin by modifying a standard counting argument to our $\Delta$-regular test setting.
While the general strategy is well-known (see e.g. \cite{AJS_book}), the application to our test design appears to be new. \nm{In the statement of the following lemma, we represent a test design $G$ on $n$ items and $m$ tests by its bi-adjacency matrix $B(G) \in \mathbb{N}_0^{m\times n}$.}

\begin{lemma}\label{Lem:Counting}
     Let $\theta \in (0,1)$, $d>0$, $\G(n,m,\Delta)$ be the item-regular test design with $\Delta = d \log(n/k)$ and $c_1^*=c_1^*(d)=\inf\cbc{c >0: c\geq c_1(d/c)}$. Then for any $\eps >0$ it holds true that if
     \begin{align}
         m \leq \bc{1-\eps} c_1^*(d) k \log (n/k),
     \end{align}
    then every function \nm{$\cA: \mathbb{N}_0^{m\times n}\times\{0,1\}^m \to \{0,1\}^n$ has $\pr\bc{\cA(\G, \hat \SIGMA) = \SIGMA} = o(1)$.} 
\end{lemma}

\begin{proof}
\nm{W.l.o.g. assume that 
$ m = c k \log (n/k)$ for a constant $c \leq \bc{1-\eps} c_1^*(d)$.}
Let the random variables $\vY=(\vY_1, \ldots, \vY_m)$ and $\vX=(\vX_1,\ldots, \vX_m)$ be defined as in Lemma~\ref{Lemma_Elemma}. As the test results can be read off from the numbers of defective items per test, we can write $\hat\SIGMA =f(\vY)$ for an appropriate function $f$. Let $\SIGMA':=f(\vX)$ for the same function $f$. Thus Lemma~\ref{Lemma_Elemma} implies that $\hat\SIGMA \stackrel{d}{=} \SIGMA' \mid \cE$ given $\vGamma$. Hence, for any $A \subseteq \{0,1\}^m$,
\begin{align}
\pr(\hat\SIGMA \in A \mid \vGamma)
= \pr(\SIGMA' \in A \mid \vGamma, \cE)
= \frac{\pr(\SIGMA' \in A, \cE \mid \vGamma)}{\pr(\cE \mid \vGamma)}
\le \frac{\pr(\SIGMA' \in A \mid \vGamma)}{\pr(\cE \mid \vGamma)}.
\end{align}
Using $\pr(\cE \mid \vGamma) = \Omega(1/\sqrt{\Delta k})$ \nm{as in \eqref{eqEprob}}, we obtain
\begin{align} \label{eq_approx_indep}
\pr(\hat\SIGMA \in A \mid \vGamma)
\le O(\sqrt{\Delta k}) \pr(\SIGMA' \in A \mid \vGamma).
\end{align}

Conditioned on $\vGamma$, the random variables $\vX_1, \ldots, \vX_m$ are independent, hence
\begin{align}
\pr(\SIGMA' = \hat\sigma \mid \vGamma)
= \prod_{i=1}^m \pr\bc{\vecone\{\vX_i < t\} = \hat\sigma_i \mid \vGamma_i}, \qquad \hat \sigma \in \{0,1\}^m.
\end{align}
Conditionally on $\vGamma$, let $\vZ=(\vZ_1,\ldots,\vZ_m)$ be a random vector with independent $\Be(\vec p_i)$ components, where
$$ \vec p_i = \pr\bc{\Bin(\vGamma_i, k/n) \le t-1 \mid \vGamma_i}.$$ 
Then
\begin{align}
\pr(\SIGMA' = \hat\sigma \mid \vGamma)
= \prod_{i=1}^m \pr\bc{\vec Z_i = \hat\sigma_i \mid \vGamma_i}.
\end{align}
Let $p_{\vZ_i\vert\vGamma_i}(\cdot) = \pr\bc{\vec Z_i = \cdot \mid \vGamma_i}$ denote the p.m.f. of $\vZ_i$ conditioned on $\vGamma_i$. Then for each outcome of $\vGamma$, $(-\log p_{\vZ_i\vert \vGamma_i}(\vZ_i))_{i \in [m]}$ is a sequence of independent random variables, where  
$-\log p_{\vZ_i\vert \vGamma_i}(\vZ_i)$ takes values in an interval of length $\vL_i = |\log(\vec p_i/(1 - \vec p_i))|$.
Hoeffding's inequality thus yields
\begin{align*}
 \pr\!\bc{\abs{\Erw\!\brk{\sum_{i=1}^m \log p_{\vZ_i\vert \vGamma_i}(\vZ_i)\Big\vert \vGamma} - \sum_{i=1}^m \log p_{\vZ_i\vert \vGamma_i}(\vZ_i) } > \eps'm \Big\vert \vGamma} \le 2 \exp\bc{- \frac{2(\eps')^2m^2}{\sum_{i=1}^m \vec L_i^2}}.
\end{align*}
Moreover, $\sum_{i=1}^m \vec L_i^2 = O(m)$ for each $\vGamma$ satisfying \eqref{eq_minmax_G}, so that 
\begin{align}\label{eq_log_conc}
 \pr\bc{\abs{\Erw\brk{\sum_{i=1}^m \log p_{\vZ_i\vert \vGamma_i}(\vZ_i)\Big\vert \vGamma} - \sum_{i=1}^m \log p_{\vZ_i\vert \vGamma_i}(\vZ_i) } > \eps'm \Big\vert \vGamma} \le  \exp\bc{- \Omega(m)}.
\end{align}
Next, set
$$p = \pr\bc{\Po(d/\nm{c}) \le t-1}.$$
Then, by the triangle inequality and mean value theorem,
\begin{align} \label{eq_approx_H}
    \abs{\Erw\brk{\sum_{i=1}^m -\log p_{\vZ_i\vert \vGamma_i}(\vZ_i) \Big\vert \vGamma} - m H(p) } & \leq \sum_{i=1}^m \abs{H(\vec p_i) -  H(p) } \nonumber \\
    & \leq  \sum_{i=1}^m \abs{H'(\vec \xi_i) \cdot \bc{\vec p_i - p}}, 
\end{align}
where $\vec \xi_i \in \bc{\vec p_i \wedge p, \vec p_i \vee p}$.
By \cite[Theorem 2.10]{Hofstad_2016}, for fixed $\Gamma$, the total variation distance between a binomial and Poisson variable is upper bounded by
$$\mathrm{d}_{\mathrm{TV}}\bc{\Bin(\Gamma, k/n), \Po(\Gamma k/n)} \le \Gamma (k/n)^2.$$
Again, it follows from \eqref{eq_approx_H} that for $\vGamma$ satisfying \eqref{eq_minmax_G}
\begin{align} \label{eq_approx_H1}
    \abs{\Erw\brk{\sum_{i=1}^m -\log p_{\vZ_i\vert \vGamma_i}(\vZ_i) \Big\vert \vGamma} - m H(p) } = O\bc{m \cdot \nm{\sqrt{\frac{\Delta n}{m}\cdot \bc{\frac{k}{n}}^2}\log n}} = o(m).
\end{align}
Equations \eqref{eq_log_conc} and \eqref{eq_approx_H1} together imply that for $\vGamma$ satisfying \eqref{eq_minmax_G} and any $\eps'>0$,
\begin{align}\label{eq_log_conc2}
 \pr\bc{\abs{\sum_{i=1}^m -\log p_{\vZ_i\vert \vGamma_i}(\vZ_i) - mH(p)} > \eps'm \Big\vert \vGamma} \le  \exp\bc{- \Omega(m)}.
\end{align}

Finally, let $ \cA_{\eps'}(\vGamma) = \cbc{\hat\sigma \in \{0,1\}^m:\abs{\sum_{i=1}^m -\log p_{\vZ_i\vert \vGamma_i}(\hat\sigma) - mH(p)} \leq \eps'm}$.
As for all $\hat \sigma \in \cA_{\eps'}(\vGamma)$, $\pr\bc{\vZ = \hat \sigma\vert \vGamma} \geq \exp(-m(H(p)\nm{+}\eps'))$, we obtain that, for all $\vGamma$,
\begin{align}
|\cA_{\eps'}(\vGamma)| \le \exp\bc{m(H(p)+\eps')}.
\end{align}

\nm{Recall now that $m = \nm{ck\log(n/k)}$ with $c \leq \bc{1-\eps} c_1^*(d) < c_1^*(d)$.} Using the definition of $c_1^*(d)$, we have \nm{$c/c_1(d/c) < 1$}. 
\nm{For some small $\delta >0$, we may thus choose $\eps'$ small enough such that} 
we obtain
\begin{align}
m(H(p)+\eps') \nm{= \bc{\frac{c}{c_1(d/c)} + c\eps'} k\log\frac{n}{k} }
\le (1-\nm{\delta}) k \log(n/k).
\end{align}
Hence for $\vGamma$ satisfying \eqref{eq_minmax_G}, \whp \ $\SIGMA'$ takes at most $\exp((1-\nm{\delta})k\log(n/k))$ distinct values. Thanks to \eqref{eq_approx_indep}, this transfers to $\hat\SIGMA$.

Finally, fix $\vGamma$ that satisfies \eqref{eq_minmax_G}, which is independent of $\SIGMA$. By Lemma~\ref{SmallBinomial}, the number of potential values of the ground truth $\SIGMA$ is $\binom{n}{k} = \exp(k\log(n/k)(1+o(1)))$. By the pigeonhole principle, \whp\, there exists a subset of $\{0,1\}^n$ of size at least $\exp\bc{\nm{\delta} k\log(n/k)(1-o(1))}$ that yields the same test results.
Therefore, as any inference rule \nm{$\cA$} can only be correct on at most one of these values,
$$\pr(\nm{\cA}(\G,\hat\SIGMA)=\SIGMA)
\le \exp(-\nm{\delta}k\log(n/k)(1-o(1))) + o(1) = o(1). \qedhere$$
\end{proof}

\subsection{Disguised Items} \label{SSec:LowerDisguised}

For the second part of the lower bound, we inspect the threshold upon which the item-regular pooling scheme $\G(n,m,\Delta)$ displays disguised items. 
An argument of this form also appears in {\cite[Corollary IV.2]{aco_2019}}:

\begin{corollary} \label{Cor:Lower}
     Let $\theta \in (0,1)$, $d>0$, $\G(n,m,\Delta)$ be the item-regular test design with $\Delta = d \log(n/k)$ and $c_2^*=c_2^*(d, \theta)=\inf\cbc{c>0: c\geq c_2(d/c,\theta)}$. Then for any $\eps >0$ it holds true that if
     \begin{align}
         m \leq \bc{1-\eps} c_2^*(d) k \log (n/k),
     \end{align}
    every function $\cA_{\G}:\{0,1\}^m \to \{0,1\}^n$ has $\pr\bc{\cA_{\G}(\G, \hat \SIGMA) = \SIGMA} = o(1)$.
\end{corollary}
\begin{proof}
Suppose that $m$ is as in the statement of the corollary, such that $c$ from \eqref{Eq:Param} satisfies $c < c_2(d/c,\theta)$. By Proposition~\ref{Lemma_V_all}(\ref{Lemma_V++}), w.h.p., $|\vV_0^+|, |\vV_1^+| = n^{\Omega(1)}$. The existence of large sets $\vV_0^+$ and $\vV_1^+$ guarantees that, starting from the ground truth $\SIGMA$, we can select one item from each set and swap their labels to obtain a competing solution $\SIGMA'$ that produces the same test results $\hat{\SIGMA}$ under the pooling scheme $\vec{G}$: For any $x_i \in \vV_0^+$ and $x_j \in \vV_1^+$, the assignment $\SIGMA'$ defined by interchanging the labels of $x_i$ and $x_j$ (i.e., setting $\SIGMA'_{x_i} = 1$,  $\SIGMA'_{x_j} = 0$ and $\SIGMA'_{x_{\ell}} = \SIGMA_{x_{\ell}}$ for $\ell \notin\{i,j\}$) has Hamming weight $k$, and yields the same test outcomes as $\SIGMA$.  
Hence, 
$$Z_k(\G,\hat\SIGMA) \geq |\vV_0^+ \times \vV_1^+| = n^{\Omega(1)} \quad \whp \qedhere$$ 
\end{proof}

\subsection{Proof of Proposition~\ref{Prop_IT_lower}} \label{SSec:Proof_Prop_IT_lower}
Let $\G = \G(n,m,\Delta)$ be the constant-column test design with item-degree $\Delta$. By Lemma~\ref{Lemma_Delta_too_small_big}, it is sufficient to restrict to $\Delta$ of the form $\Delta = d \log\frac n k$ for $d>0$.  For fixed $d$, Lemma~\ref{Lem:Counting} and Corollary~\ref{Cor:Lower} imply that no inference procedure recovers the ground truth correctly with non-vanishing probability if $c < \max\{c_1^*(d), c_2^*(d,\theta)\}$. Thus, for $c<\inf_d \max\{c_1^*(d), c_2^*(d,\theta)\}$ there is no item-regular pooling scheme on $m=ck\log(n/k)$ tests that recovers the ground truth with non-vanishing probability. By Claim \ref{claim_opt_eq}, this expression equals $ c_{\mathrm{inf}}^{\mathrm{TGT}}$.  
\qed
\section{Information-Theoretic Upper Bound} \label{Sec:Upper}

We now address the achievability result from Theorem \ref{Thm_inf}(b), the derivation of the information-theoretic upper bound. 
Our objective is to prove that for a certain choice of $\Delta$, $Z_k(\G,\hat\SIGMA) = 1$ \whp, indicating that $\SIGMA$ is the sole assignment compatible with the observed test outcomes.

\begin{proposition}[\Thm~\ref{Thm_inf}(b)]\label{prop_upper_main}
Fix $0<\theta<1$ and $\eps>0$. Assume that either $t=2$, or that $t \geq 3$ and Conjecture~\ref{Conj} holds. If $m>(1+\eps) c_{\mathrm{inf}}^{\mathrm{TGT}}k\log(n/k)$, then there exist a constant-column test design $\G=\G(n,m,\Delta)$ and an algorithm that, given $\G,\hat\SIGMA, k$, outputs $\SIGMA$ w.h.p.
\end{proposition}

We now outline the philosophy behind the proof of Proposition~\ref{prop_upper_main}. 
We define the \emph{overlap} 
\begin{align}
    \langle \sigma, \tau\rangle = \sum_{i=1}^n \vecone\cbc{\sigma_{x_i} = \tau_{x_i} =1} 
\end{align}
between two assignments $\sigma, \tau \in \{0,1\}^n$.
Setting
\begin{align} \label{Eq:Zkl}
Z_{k,\ell}(\mathbf{G},\hat{\SIGMA}) = \abs{\{ \sigma \in S_k (\mathbf{G},\hat{\SIGMA}) : \langle\SIGMA, \sigma\rangle = \ell \}}, \qquad  \ell \in \{0, \dots, k\},  
\end{align}
this yields the decomposition
\begin{align*}
    Z_k(\mathbf{G},\hat{\SIGMA}) = \sum_{\ell = 0}^{k} Z_{k,\ell}(\mathbf{G},\hat{\SIGMA}).
\end{align*}
Our proof of Proposition~\ref{prop_upper_main} is structured in three parts, of which the first two are similar in spirit to \cite[Propositions III.1 and III.2]{aco_2019}. \cri{For this, recall $r^*$ from \eqref{def_rstar}}. 
\begin{lemma}[No Small Overlap \whp] \label{Lem:IntTc}
Fix $0<\theta<1$, $t \in \mathbb{N}^+$ and $\eps>0$. Suppose that $m=c k \log(n/k)$ with $c\geq (1+\eps)\cinf(n,\theta,t)$ and consider the constant column pooling scheme $\G(n,m,\Delta)$ with $\Delta = r^*c \log(n/k)$. 
 If   Conjecture~\ref{Conj} holds,  w.h.p.\ $Z_{k,\ell}(\mathbf{G},\hat{\SIGMA})=0$ for all $0 \leq \ell \leq (1- 1/\nm{\sqrt{\log n}})k$.
\end{lemma}

\begin{proposition}[No Large Overlap \whp] \label{Prop:LrgU}
Fix $0<\theta<1$, $t \in \mathbb{N}^+$ and $\eps>0$. Suppose that $m=c k \log(n/k)$ with $c\geq (1+\eps)\cinf(n,\theta,t)$ and consider the constant pooling scheme $\G(n,m,\Delta)$ with $\Delta = r^*c \log(n/k)$. Then w.h.p. $Z_{k,\ell}(\mathbf{G},\hat{\SIGMA})=0$ for all $(1- 1/\nm{\sqrt{\log n}} )k\leq \ell<k$.
\end{proposition}

The third part consists in showing that our upper bound in Lemma \ref{Lem:IntTc} is unconditional for $t=2$.
\begin{lemma}\label{Lem:T2}
For $t=2$, equation \eqref{Eq:Conj_reduced} holds for all $\alpha \in (0,1)$. 
\end{lemma}


 In Section~\ref{SSec:UpperSmall}, we prove Lemma~\ref{Lem:IntTc} by a moment calculation to establish that, \whp, there are no competing solutions exhibiting low overlap with $\SIGMA$. Subsequently, in Section~\ref{SSec:UpperLarge}, we prove Proposition~\ref{Prop:LrgU} by an expansion argument to exclude, \whp, the existence of competing solutions with high overlap. 
The strategy in \cite[Section III]{aco_2019} for classical group testing serves as a direct inspiration for our approach. Building upon this, we extend the methodology to the threshold group testing setting, introducing several important modifications. These include new expressions involving total variation distance between the probability distributions that allow us to handle more general test outcomes. \nm{In Section \ref{sec:conj2}, we prove Lemma \ref{Lem:T2}. Finally the proof of Proposition \ref{prop_upper_main} and, therefore Theorem \ref{Thm_inf}(b), is given in Section \ref{SSec:ProofUpperMain}.}

\subsection{Small Overlap - Truncated First Moment} \label{SSec:UpperSmall}

In this section, we prove Lemma~\ref{Lem:IntTc}.
We begin by studying the first moment of $Z_{k, \ell}(\G,\hat\SIGMA)$, first defined in \eqref{Eq:Zkl}, for the regime when the overlap $\ell$ is relatively small. We first establish the following upper bound, that will be analysed later.

\begin{lemma}\label{Lem:first_bound_small_conditioned}
Given $\vGamma$, let $(\vB_i(\ell))_{i \in [m]}$ be a sequence of independent random variables with $\vB_i(\ell) \sim \Bin(\vGamma_i, \ell/n)$ and let $(\vM^{(z)}_i(\ell))_{i \in [m], z \in \{0, \ldots, t-1\}}$ be a sequence of independent random variables with 
\begin{align*}
\vM^{(z)}_i(\ell) &= \bc{\vM^{(z),1}_i(\ell), \vM^{(z),2}_i(\ell), \vM^{(z),3}_i(\ell)} \\ &\sim \Mult\bc{\vGamma_i-z;\frac{k-\ell}{n-\ell}, \frac{k-\ell}{n-\ell}, \frac{n-2k+\ell}{n-\ell}}.
\end{align*}
Finally, let 
$\cP = \{i \in [m] : \vB_i(k) \geq t\}$ and $\cN=[m]\setminus \cP$.
For $i \in [m]$, let $q_i := \pr\bc{\vB_i(k)\leq t-1 \vert \vGamma}$, and let
\begin{align}\label{def_pil}
    p_{i,\ell,\vGamma} &:= \sum_{z=0}^{t-1} \Pr\bc{\vB_i(\ell) = z \vert \vGamma} \pr\bc{\vM^{(z),1}_i \leq t-1-z, \vM^{(z),2}_i \leq t-1-z \vert\vGamma}.
\end{align}
Let the events $\cM, \cM'$ be defined as $\cM:=\cbc{\vm_{\geq t}\in \brk{\bc{1 \pm n^{-\Omega(1)}} m\pr\bc{\Po(d/c)\geq t}}}$ and $\cM':=\cbc{\abs{\cP}\in \brk{\bc{1 \pm n^{-\Omega(1)}}m\pr\bc{\Po(d/c)\geq t}}}$. 
 Then for any $0 \leq \ell \leq (1- 1/\log{n})k$.
 \begin{align}\label{first_bound_small_cond}
   &  \Erw\brk{\vecone_{\cM}Z_{k,\ell}(\mathbf{G},\hat{\SIGMA})\mid \vGamma} \nonumber \\
   \leq \quad &  O\bc{\bc{\Delta k}^{3/2}} \binom{k}{\ell} \binom{n-k}{k-\ell} \Erw\brk{\vecone_{\cM'}\prod_{i \in \cN} \bc{\frac{p_{i,\ell,\vGamma}}{q_i}} \prod_{i \in \cP} \bc{\frac{1 - 2q_i + p_{i,\ell,\vGamma}}{1 - q_i}}\vert \vGamma}.
 \end{align}
\end{lemma}
\begin{proof}
The aim of the lemma is to characterise assignments $\sigma$ that have Hamming weight $k$ and overlap exactly $\ell$ with $\SIGMA$. Even though $\SIGMA$ is random and $\sigma$ depends on $\SIGMA$, due to symmetry we may assume that $\SIGMA_{x_i}=\vecone\{i\leq k\}=\vecone_{[k]}$ and $\sigma_{x_i}=\vecone\{i\leq\ell\}+\vecone\{k<i\leq 2k-\ell\}$. Then for this choice of $\sigma$,
\begin{align}\label{eq:exp_sym}
   &  \Erw\brk{\vecone_{\cM} Z_{k,\ell}(\mathbf{G},\hat{\SIGMA})\mid \vGamma}  = \binom{k}{\ell} \binom{n-k}{k-\ell} \pr\bc{\cM, \forall a \in [m]: \hat\sigma_a = (\hat\vecone_{[k]})_a \vert \vGamma}.
 \end{align}
	We will use a balls and bins argument, where we think of each test $a_i$ as a bin of capacity $\vGamma_i$, and of each clone $(x_i,h)$, $h\in[\Delta]$, of an item as a ball labelled $(\SIGMA_{x_i},\sigma_{x_i})\in\{0,1\}^2$.
	Creating the graph $\G$ is then equal in distribution to matching the $\Delta n$ balls uniformly into the bins given their capacities.
	For $i\in[m]$ and $j\in[\vGamma_i]$, we let $\vA_{i,j}=(\vA_{i,j,1},\vA_{i,j,2})\in\{0,1\}^2$ be the label of the $j$th ball that ends up in bin number $i$. Observe that we have $\ell \Delta$ balls with label $(1,1)$, $(k-\ell)\Delta$ balls with label $(1,0)$, $(k-\ell)\Delta$ balls with label $(0,1)$ and $(n-2k+\ell)\Delta$ balls with label $(0,0)$. Moreover,
    \begin{align}\label{eq:balls_bins1}
        &\pr\bc{\cM, \forall a \in [m]: \hat\sigma_a = (\hat\vecone_{[k]})_a \vert \vGamma} \nonumber \\
        =& \pr\bc{\cM, \forall i \in [m]: \vecone\cbc{\sum_{i=1}^{\vGamma_i}\vA_{i,j,1} \leq t-1} = \vecone\cbc{\sum_{i=1}^{\vGamma_i}\vA_{i,j,2} \leq t-1} \vert \vGamma}.
    \end{align}

    To facilitate the analysis of the experiment, we introduce a new set of i.i.d. random variables $(\vA_{i,j}')_{(i,j) \in \bigcup_{a=1}^m \{a\} \times [\vGamma_a]}=((\vA_{i,j,1}',\vA_{i,j,2}'))_{(i,j) \in \bigcup_{a=1}^m \{a\} \times [\vGamma_a]}$ with distribution
	\begin{align*}
	\pr\!\bc{\vA_{i,j}'=(1,1)}&=\ell/n,&
	\pr\!\bc{\vA_{i,j}'=(0,1)}&=\pr\!\bc{\vA_{i,j}'=(1,0)}=(k-\ell)/n,\\
	\pr\!\bc{\vA_{i,j}'=(0,0)}&=(n-2k+\ell)/n
	\end{align*}
	for all $i,j$. 
Moreover, let $\cT$ be the event that 
    \begin{align}  \sum_{i=1}^m\sum_{j=1}^{\vGamma_i}\vecone\cbc{\vA_{i,j}'=(1,1)}&=\ell\Delta, \qquad \sum_{i=1}^m\sum_{j=1}^{\vGamma_i}\vecone\cbc{\vA_{i,j}'=(0,0)}=(n-2k+\ell)\Delta, \nonumber \\ \sum_{i=1}^m\sum_{j=1}^{\vGamma_i}\vecone\cbc{\vA_{i,j}'=(1,0)}&= \sum_{i=1}^m\sum_{j=1}^{\vGamma_i}\vecone\cbc{\vA_{i,j}'=(0,1)}=(k-\ell)\Delta. \label{eq_def_U} \end{align}
$\cT$ describes the event that the total counts of each ball type in the independent model match the original counts, 
and given $\cT$, the vector $\vA'=(\vA_{i,j}')_{i,j}$ is distributed identically to $\vA=(\vA_{i,j})_{i,j}$ given $\vGamma$. For $u \in \{1,2\}$ let $(\hat\vA')^{(u)} \in \{0,1\}^m$ be the vector defined by the `test results' of the $u$th coordinate in the independent model. That is, for $i \in [m]$, set $(\hat\vA')^{(u)}_i = \vecone\cbc{\sum_{i=1}^{\vGamma_i}\vA_{i,j,u}' \geq t}$. Moreover, let $\cP = \cbc{i \in [m]: (\hat\vA')^{(1)}_i =1}$, $\cM':=\cbc{\abs{\cP}\in \brk{\bc{1 \pm n^{-\Omega(1)}}m\pr\bc{\Po(d/c)\geq t}}}$ and $\cN = [m]\setminus \cP$. Then
\begin{align}\label{eq:balls_bins2}
    &\pr\bc{\cM, \forall i \in [m]: \vecone\cbc{\sum_{i=1}^{\vGamma_i}\vA_{i,j,1} \leq t-1} = \vecone\cbc{\sum_{i=1}^{\vGamma_i}\vA_{i,j,2} \leq t-1} \vert \vGamma} \nonumber \\
    =&  \pr\bc{\cM', \forall i \in [m]: \vecone\cbc{\sum_{i=1}^{\vGamma_i}\vA_{i,j,1}' \leq t-1} = \vecone\cbc{\sum_{i=1}^{\vGamma_i}\vA_{i,j,2}' \leq t-1} \vert \vGamma, \cT}.
\end{align}
\nm{Given any valid degree sequence $\vGamma$ with $\sum_{i \in [m]}\vGamma_i = n\Delta$, the vector  \footnotesize
\begin{align*}
    \bc{\sum_{\substack{i\in [m],\\ j\in [\vGamma_i]}}\vecone\cbc{\vA_{i,j}'=(0,0)}, \sum_{\substack{i\in [m],\\ j\in [\vGamma_i]}}\vecone\cbc{\vA_{i,j}'=(0,1)}, \sum_{\substack{i\in [m],\\ j\in [\vGamma_i]}} \vecone\cbc{\vA_{i,j}'=(1,0)},\sum_{\substack{i\in [m],\\ j\in [\vGamma_i]}} \vecone\cbc{\vA_{i,j}'=(1,1)}}
\end{align*} \normalsize
has a $\Mult(n\Delta; (n-2k+\ell)/n, (k-\ell)/n, (k-\ell)/n, \ell/n)$-distribution. In particular, it's expectation vector is $(\Delta(n-2k+\ell), \Delta (k-\ell), \Delta(k-\ell), \Delta \ell)$. An application of Stirling's formula, which underpins the LCLT in this context, yields
}
	\begin{align}\label{eqLLT8}
	\pr\bc{\cT\vert \vGamma}=\Omega((\Delta k)^{-3/2}).
	\end{align}
    Now set
    $$\cS :=\cbc{\forall i \in [m]: \vecone\cbc{\sum_{i=1}^{\vGamma_i}\vA_{i,j,1}' \leq t-1} = \vecone\cbc{\sum_{i=1}^{\vGamma_i}\vA_{i,j,2}' \leq t-1}}.$$
Thus, by Bayes' Theorem and \eqref{eqLLT8},
\begin{align}\label{eq:balls_bins3}
    \pr\bc{\cS \cap \cM' \vert \vGamma, \cT} = \frac{\pr\bc{\cS \cap \cT \cap \cM' \vert \vGamma}}{\pr\bc{\cT \vert \vGamma}} \leq O((\Delta k)^{3/2}) \pr\bc{\cS \cap \cM' \vert \vGamma},
\end{align}
and it remains to upper bound $\pr\bc{\cS \cap \cM' \vert \vGamma}$. 
Observe that in the auxiliary model, without conditioning on $\cT$, the random variables $\vA'_{i,j}$ are independent.
Thus, the events $\cS_i$ (the outcome of the test $i$) are independent across all $i \in [m]$ when conditioned on the degree sequence $\vGamma$.
Then
\begin{align} \label{eq:balls_bins4}
& \pr\bc{\cS \cap \cM'\vert \vGamma} = \Erw\brk{\pr\bc{\cS\cap \cM' \vert \vGamma, (\hat\vA')^{(1)}}\vert \vGamma} = \Erw\brk{\mathds{1}_{\cM'}\pr\bc{\cS\vert \vGamma, (\hat\vA')^{(1)}}\vert \vGamma}  \nonumber \\
 &= \Erw\brk{\mathds{1}_{\cM'}\prod_{i \in \cN} \pr\bc{ (\hat\vA')^{(2)}_i = 0 \vert \vGamma, (\hat\vA')^{(1)}_i =0}\prod_{i \in \cP} \pr\bc{(\hat\vA')^{(2)}_i = 1 \vert \vGamma, (\hat\vA')^{(1)}_i =1}  \Big \vert \vGamma} \nonumber \\
 &= \Erw\brk{\mathds{1}_{\cM'}\prod_{i \in \cN} \frac{\pr\bc{(\hat\vA')^{(1)}_i = (\hat\vA')^{(2)}_i = 0 \vert \vGamma}}{\pr\bc{(\hat\vA')^{(1)}_i =0 \vert\vGamma}}\prod_{i \in \cP}\frac{\pr\bc{(\hat\vA')^{(1)}_i = (\hat\vA')^{(2)}_i = 1 \vert \vGamma}}{\pr\bc{(\hat\vA')^{(1)}_i =1 \vert\vGamma}}  \Big \vert \vGamma}.
\end{align}
As each $\vA'_{i,j,1} \sim \Be(k/n)$, 
\begin{align}\label{eq:balls_bins100}
     \pr\bc{(\hat\vA')^{(1)}_i =0 \vert\vGamma} = \pr\bc{\Bin\bc{\vGamma_i, \frac{k}{n}}\leq t-1 \vert \vGamma} = q_i.
\end{align}
Now observe that for any bin $i \in [m]$, the vector of ball counts of the four types 
has a Multinomial distribution with $\vGamma_i$ trials and probabilities $\ell/n, (k-\ell)/n, (k-\ell)/n, (n-2k+\ell)/n$, respectively. Conditionally on the number of $(1,1)$-balls being $z$, the numbers of the remaining balls have a $\Mult\bc{\vGamma_i-z;\frac{k-\ell}{n-\ell}, \frac{k-\ell}{n-\ell}, \frac{n-2k+\ell}{n-\ell}}$ distribution. Finally, for $z \in \{0,\ldots,t-1\}$, let 
\begin{align*}
\vM^{(z)}_i(\ell) &= \bc{\vM^{(z),1}_i(\ell), \vM^{(z),2}_i(\ell), \vM^{(z),3}_i(\ell)} \\ &\sim \Mult\bc{\vGamma_i-z;\frac{k-\ell}{n-\ell}, \frac{k-\ell}{n-\ell}, \frac{n-2k+\ell}{n-\ell}}.
\end{align*} Conditioning on the number of $(1,1)$ balls, we obtain
\begin{align}\label{eq:balls_bins5}
&\pr\bc{(\hat\vA')^{(1)}_i = (\hat\vA')^{(2)}_i = 0 \vert \vGamma} \nonumber \\
=& \sum_{z=0}^{t-1} \pr\bc{\vB_i(\ell) = z\vert\vGamma} \pr\bc{\vM^{(z),1}_i \leq t-1-z, \vM^{(z),2}_i \leq t-1-z \vert\vGamma} =  p_{i,\ell,\vGamma}.   
\end{align}
On the other hand, by the inclusion-exclusion principle 
\begin{align}\label{eq_balls_bins6}
   \pr\bc{(\hat\vA')^{(1)}_i = (\hat\vA')^{(2)}_i=1 \vert\vGamma} = 1 - 2q_i + p_{i,\ell,\vGamma}.
\end{align}
Finally, combining equations \eqref{eq:exp_sym}, \eqref{eq:balls_bins1}, \eqref{eq:balls_bins2}, \eqref{eq:balls_bins3}-\eqref{eq_balls_bins6} yields the claim.
\end{proof}

\begin{lemma} \label{Prob_Upper}
Fix $0 \leq \ell \leq (1- 1/\nm{\sqrt{\log n}} )k$, let $\alpha = \ell/k$.
Recall $q_i$ and $p_{i,\ell,\vGamma}$ from Lemma \ref{Lem:first_bound_small_conditioned}. Set
    $$ 
    \vE^* \sim \Po(d/c), \quad
    \vB^* \sim \Po(\alpha d/c), \quad 
    \vX_1 \sim \Po((1-\alpha)d/c), \quad
    \vX_2 \sim \Po((1-\alpha)d/c),
    $$
    where all random variables $\vE^*, \vB^*, \vX_1, \vX_2$ are independent. Moreover, let
    $$\vM^* \sim \max\{\vX_1, \vX_2\}.$$
Then, with probability $1-o(n^{-2})$, for all $i \in [m]$,
\begin{align*}
 q_i &=  \Pr(\vE^*\leq t-1) + O\bc{\sqrt{k/n}\log n};  \\
 p_{i,\ell,\vGamma} &= \Pr(\vB^*+\vM^*\leq t-1)   +O\bc{\sqrt{k/n}\log n}.
\end{align*}
\end{lemma}

\begin{proof}
Firstly, recall the following upper bound between the total variation distance between a binomial and Poisson variable with identical means (see \cite[Theorem 2.10]{Hofstad_2016}):
\begin{align}\label{e:bin_to_po}
     \mathrm{d}_{\mathrm{TV}} \bc{\Bin(N,p), \Po(Np)} \leq Np^2.
\end{align}
Conditionally on $\vGamma_i$, the triangle inequality and \eqref{e:bin_to_po} bound give, almost surely, 
\begin{align} \mathrm{d}_{\mathrm{TV}} \left( \Bin\left(\vGamma_i,\frac{k}{n}\right), \vE^* \right) &\leq \frac{\vGamma_i k^2}{n^2} + \mathrm{d}_{\mathrm{TV}} \left( \Po\left(\frac{\vGamma_i k}{n}\right), \vE^* \right) \nonumber\\
&\leq \frac{\vGamma_i k^2}{n^2} + \left| \frac{\vGamma_i k}{n}-\frac{d}{c} \right|. \label{Eq:DTV_Shifted} 
\end{align}
Applying Lemma \ref{Lemma_GammaMinMax} and the total variation distance bound from \eqref{Eq:DTV_Shifted}, with probability at least $1-o(n^{-2})$, 
    \begin{align} \label{eq_approx_qi}
q_i = \pr\bc{\Bin\bc{\vGamma_i, \frac{k}{n}}\leq t-1 \Big\vert \vGamma}   = \Pr\bc{\vE^* \leq t-1} + O\bc{\sqrt{k/n}\log n}.
\end{align}

To evaluate $p_{i,\ell,\vGamma}$, we first note that the first coordinate of $\vM^{(z)}_i(\ell)$ has a binomial marginal distribution:
$$\vM_i^{(z),1}(\ell) \sim \Bin\bc{\vGamma_i-z, \frac{k-\ell}{n-\ell}}.$$
Conditional on the first coordinate count $\vM_i^{(z),1}(\ell)=x$ and $\vGamma$, the second coordinate satisfies
$$\bc{\vM_i^{(z),2}(\ell) \mid \vM_i^{(z),1}(\ell)=x, \vGamma} \sim \Bin\bc{ \vGamma_i-z-x, \frac{k-\ell}{n-k}}.$$

The exact joint probability of yielding a negative test can then be written as 
\begin{align} 
p_{i,\ell,\vGamma} &= \sum_{z=0}^{t-1} \Pr\bigl(\vB_i(\ell)=z \mid \vGamma\bigr) \sum_{x=0}^{t-1-z} \Pr\bigl(\vM_i^{(z),1}(\ell)=x \mid \vGamma\bigr) \nonumber \\ 
&\qquad\times \Pr\left( \vM_i^{(z),2}(\ell)\leq t-1-z \,\middle|\, \vGamma,\vM_i^{(z),1}(\ell)=x \right). \label{e:start_p}
\end{align}
We now apply a Poisson approximation to these distributions.  Let $\vB^\ast \sim \Po(\alpha d/c)$ and define
$$ \vM^\ast = \max\left\{ \vX_1, \vX_2 \right\} ,$$
where $\vX_1, \vX_2 \sim \Po((1-\alpha )d/c)$ are independent. Using Lemma \ref{Lemma_GammaMinMax} together with the total variation bound in \eqref{Eq:DTV_Shifted}, similarly to \eqref{eq_approx_qi}, it is immediate that with probability at least $1-o(n^{-2})$, \nm{uniformly over $z, x, \ell$ in the specified range,
\begin{align}
 & \Pr(\vB_i (\ell) = z\mid \vGamma) =   \Pr(\vB^*=z) +  O\bc{\sqrt{\frac{k}{n}}\log n}, \label{e:approx_p1} \\
  &  \Pr\bc{\vM_i^{(z),1}(\ell) = x \mid \vGamma} = \Pr\bc{\vX_1=x}+  O\bc{\sqrt{\frac{k}{n}} \log n}, \label{e:approx_p2} \\
  &  \Pr\bc{\vM_i^{(z),2}(\ell) \le t-1-z \mid \vM_i^{(z),1}(\ell)=x,\vGamma} = \Pr\bc{\vX_2 \le t-1-z} +  O\bc{\sqrt{\frac{k}{n}} \log n}. \label{e:approx_p3}
\end{align}
Plugging \eqref{e:approx_p1}-\eqref{e:approx_p3} into \eqref{e:start_p} thus gives 
\begin{align}\label{eq:poisson_approx_2}
    p_{i,\ell,\vGamma}
    &= \sum_{z=0}^{t-1} \Pr(\vB^*=z)
    \Pr\bc{\vX_1 \le t-1-z}
    \Pr\bc{\vX_2 \le t-1-z}
    + O\bc{\sqrt{k/n}\log n}
    \nonumber\\
    &= \Pr(\vB^\ast+\vM^\ast \le t-1)
    + O\bc{\sqrt{k/n}\log n}.
\end{align}}

\end{proof}

This next lemma provides a bound on the binomial coefficients in the small overlap regime:
\begin{lemma}\label{Lem:BinC}
    For $\alpha \in [0, 1- 1/\nm{\sqrt{\log n}} ]$,
    \[ \binom{k}{\alpha k}\binom{n-k}{(1-\alpha )k} = \mathrm{e}^{(1-\alpha )(1-\theta)k\log n +O(k) } .\]
\end{lemma}

\begin{proof} 
Observe that $\binom{k}{\alpha k}\leq 2^k = \eul^{O(k)}$ for any $\alpha \in [0, 1- 1/\nm{\sqrt{\log n}} ]$. 
Moreover, as $(1-\alpha)k = o(n-k)$ for all  $\alpha \in [0, 1- 1/\log{n}]$, Lemma~\ref{SmallBinomial} further yields that 
    \begin{align*}
\binom{k}{\alpha k}\binom{n-k}{(1-\alpha )k} &= \exp\bc{ (1-\alpha )k\log n - (1-\alpha ) k ((1-\theta)\log n +O(k) } \\
        &= \mathrm{e}^{(1-\alpha )(1-\theta)k\log n +O(k)}.
\end{align*} 
\end{proof}

\cri{\begin{proof}[Proof of Lemma \ref{Lem:IntTc}]
Recall $q=q(r^*)$ and $p(\alpha) = p(\alpha, r^*)$ from Conjecture \ref{Conj}. 
Let the events $\cM, \cM'$ be defined as $\cM:=\cbc{\vm_{\geq t}\in \brk{\bc{1 \pm n^{-\Omega(1)}} m\pr\bc{\Po(d/c)\geq t}}}$ and $\cM':=\cbc{\abs{\cP}\in \brk{\bc{1 \pm n^{-\Omega(1)}}m\pr\bc{\Po(d/c)\geq t}}}$. Then $\cM, \cM'$ are high-probability events, where the proof for $\cM'$ is analogous to the proof for $\cM$ given in Lemma~\ref{Lemma_m0}. We evaluate the expectation of $\mathds{1}_{\cM} Z_{k,\ell}(\mathbf{G},\hat{\SIGMA})$ given $\vGamma$ using Lemmas \ref{Lem:first_bound_small_conditioned} -- \ref{Lem:BinC} for the region pertaining to $\alpha \in [0,1-1/\sqrt{\log{n}}]$ as
\begin{align}
    \Erw&\brk{\mathds{1}_{\cM} Z_{k,\ell}(\mathbf{G},\hat{\SIGMA}) \mid\vGamma} \nonumber \\
    &\leq O\bc{\bc{\Delta k}^{3/2}} \binom{k}{\ell} \binom{n-k}{k-\ell} \Erw\brk{\mathds{1}_{\cM'} \prod_{i \in \cN} \bc{\frac{p_{i,\ell,\vGamma}}{q_i}} \prod_{i \in \cP} \bc{\frac{1 - 2q_i + p_{i,\ell,\vGamma}}{1 - q_i}}\mid \vGamma} \nonumber \\
         &\leq   \exp{\bc{ (1-\alpha )(1-\theta)k\log n +O(k)}}\nonumber \\ 
         &\quad\times\exp\bc{  m\bc{ q \log\bcfr{p(\alpha)}{q} + (1-q) \log \bcfr{1-2 q + p(\alpha)}{1-q}}
          + O\bc{n^{-\Omega(1)}m}}. \label{Eq:FMM}
\end{align}    
Now let $W:[0,1] \times [0,\infty)$ be given by
$$W(\alpha,c):=1 - \alpha + c \bc{q(r^*) \log\bcfr{p(r^*,\alpha)}{q(r^*)} + (1-q(r^*)) \log\bcfr{1-2q(r^*)+p(r^*,\alpha)}{1-q(r^*)}}.$$
Then $W$ is exactly the function appearing in \eqref{Eq:Conj} within Conjecture \ref{Conj}. 
With this notation, since $(1-\theta)\log n = \log(n/k)$ and $m=ck\log(n/k)$, \eqref{Eq:FMM} becomes
\begin{align}
    \Erw&\brk{\mathds{1}_{\cM} Z_{k,\ell}(\mathbf{G},\hat{\SIGMA}) \mid\vGamma} \leq  \exp\bc{k \log(n/k) W(\alpha, c) + O(k)}. \label{Eq:FMM2}
\end{align}    
Now for all $c\geq (1+\eps)c_{\mathrm{inf}}^{\mathrm{TGT}}$, $W(0,c) = 1-cH(q) <0$ and $W(1,c)=0$. Moreover, by Lemma \ref{lem:W_derivative_boundary} 
$\left. \partial_\alpha W(\alpha,c) \right|_{\alpha=1} > 0$ and, if Conjecture \ref{Conj} holds, then $W(\alpha,c)<0$ for $\alpha \in (0,1)$.  Thus, the function $\alpha \mapsto W(\alpha,c)$ satisfies all assumptions of Lemma \ref{lem:analysis1}. Choosing $\delta_n = (\sqrt{\log n})^{-1}$ and $D=\left. \partial_\alpha W(\alpha,c) \right|_{\alpha=1}/2$, Lemma \ref{lem:analysis1} implies that $W(\alpha,c) \leq -D/\sqrt{\log n}$ for all $\alpha \in [0,1-1/\sqrt{\log n}]$. 
Plugging this upper bound
into \eqref{Eq:FMM2}, we obtain
\begin{align}
    \Erw&\brk{\mathds{1}_{\cM} Z_{k,\ell}(\mathbf{G},\hat{\SIGMA}) \mid\vGamma} \leq \exp\bc{- D \frac{m}{\sqrt{\log n}}  + O(k)}. 
\end{align}
Therefore, by Markov's inequality and a union bound over $\ell \leq k-k/\sqrt{\log n}$, $Z_{k,\ell}(\mathbf{G},\hat{\SIGMA})=0$ w.h.p. for any $m> (1+\eps)c_{\mathrm{inf}}^{\mathrm{TGT}}k \log(n/k)$.
\end{proof}}

\subsection{Large Overlap} \label{SSec:UpperLarge}

When the overlap $\ell$ comes too close to $k$, the first moment of $Z_{k,\ell}(\G,\SIGMA)$ becomes too large. To overcome this, we require a more careful analysis that guarantees every defective participates in enough informative tests, similar to the analysis conducted in \cite{aco_2019}.

\paragraph{Outline of the proof.} The proof of Proposition~\ref{Prop:LrgU} is organised into two main parts. Firstly, in Proposition~\ref{Lemma_rigid}, we define 
$\cR$  as the event that every defective item is included in a sufficient number of pivotal tests and show that for $m>(1+\eps)c_{\mathrm{inf}}^{\mathrm{TGT}}k \log(n/k)$, $\cR$ occurs \whp.  
Subsequently, in Lemma~\ref{Lemma_big_overlap}, we analyse the expectation of $Z_{k,\ell}(\G,\hat{\SIGMA})$ conditioned on the event $\cR$, which addresses the lottery phenomenon that renders the unconditional first moment path unwalkable.  
Finally, in the proof of Proposition~\ref{Prop:LrgU}, we combine these results to show that the conditional expectation is of order $o(1)$ for any $m>(1+\eps)c_{\mathrm{inf}}^{\mathrm{TGT}}k \log(n/k)$. 

Recall that $n_y(a)$ denotes the number of times item $y$ appears in test $a$.

\begin{proposition}\label{Lemma_rigid}
    For any $\eps>0$, there exists $\delta=\delta(\eps)>0$ such that for all $m>(1+\eps)c_{\mathrm{inf}}^{\mathrm{TGT}}k \log(n/k)$, the following is true. Let $\cR$ be the event that, for every $x_i$ with $\SIGMA_{x_i}=1$, there are at least $\delta \Delta$ \nm{distinct} tests $a \in \partial x_i$ such that $\sum_{y \in \partial a} n_y(a) \sigma_y =t$. Then  $\pr(\cR)=1-o(1)$.
\end{proposition}

\begin{proof}
Given $\delta >0$, let $\vec T$ be the number of defective items that participate in less than $\delta \Delta$ (distinct) tests with $t$ defective items, i.e.,
	$$ \vec T = \abs{\left\{x\in \vV_1 : \sum_{a\in\partial x}\vecone\cbc{\sum_{y \in \partial a}n_y(a)\SIGMA_y =t}<\delta\Delta\right\}}.$$
The claim follows from showing that for $m>(1+\eps)\minf$, $\delta=\delta(\eps)$ sufficiently small, $\Erw[\vec T] = o(1)$.

For any $x \in \vV_1$, let $\vD_x$ be the number of distinct tests with quantitative result $t$ in which $x$ participates (cf.~discussion in  Section \ref{SSec:repeated}). Recall that $\vm_t$ denotes the number of tests with quantitative result $t$. The total number of `defective' half-edges matching to tests with result $t$ is exactly $t\vm_t$, while the total number of defective half-edges is $k \Delta$. 

We expose the $\Delta$ half-edges of $x$ sequentially: 
Let $\cF_{i-1}$ be the $\sigma$-algebra generated by pairing the first $i-1$ half-edges of $x$. At stage $i$, at most $i-1$ distinct tests with quantitative result $t$ have been connected to $x$. Each such test has exactly $t$ half-edges matched to defective items, and at most $i-1$ half-edges belonging to these tests have already been used. Therefore, by the uniform choice of the matching, conditionally on $\cF_{i-1}, \vm_t$ and $x \in V_1$
\begin{align*}
 \pr&\bc{\text{the $i$-th half-edge of $x$ connects to a new test with result $t$}\mid \cF_{i-1}, \vm_t, x \in V_1}\\ &\geq\frac{t\vm_t-t(i-1)}{k\Delta} \geq \frac{t\vm_t-t\Delta}{k\Delta} =:p_{\min}(\vm_t).
\end{align*}

By a standard coupling argument, this uniform lower bound on the conditional success probability implies that $\vD_x$ stochastically dominates a binomial random variable
$$\vD_x \succeq \Bin(\Delta, p_{\min}(\vm_t)) := \vB.$$
Consequently, for any $\delta > 0$:
$$\pr(\vD_x < \delta\Delta \mid \vm_t, \SIGMA) \leq \pr(\vB < \delta\Delta \mid \vm_t, \SIGMA).$$
By linearity of expectation,
\begin{align}\label{eqLemma_V1delta_2}
	\Erw\brk{\vec T \mid \vm_t, \SIGMA} \leq \sum_{x \in \vV_1} \pr\bc{\vB < \delta \Delta \mid \vm_t, \SIGMA} = k \pr\bc{\vB < \delta \Delta \mid \vm_t}.
\end{align}
Applying the Chernoff bound Theorem \ref{lem_chernoff_2}, we obtain for $\delta < p_{\min}(\vm_t)$ that
\begin{align}\label{eqLemma_V1delta_3}
	\pr\bc{\vB < \delta\Delta \mid \vm_t} \leq \exp\bc{-\Delta \KL{\delta}{p_{\min}(\vm_t)}}.
\end{align}
By Lemma \ref{Lemma_m0}, for $\vGamma$ satisfying \eqref{eq_minmax_G}, using the fact that our choice of $\Delta$ and $m$ gives $r^* m= k\Delta$ and $d/c = r^*$, 
we obtain
	\begin{align}\label{eqLemma_V1delta_1}
	\pr\bc{\vm_t\in\left.\brk{\bc{1\pm n^{-\Omega(1)}} \frac{k\Delta}{t!} (r^*)^{t-1} \eul^{-r^*}}\right|\vGamma} &=1-o(n^{-7}),
	\end{align}
so that for $\vGamma$ satisfying \eqref{eq_minmax_G}, with probability $1-o(n^{-7})$,
    \begin{equation}\label{eq_pmin}
    p_{\min}(\vm_t)\in \brk{\bc{1\pm n^{-\Omega(1)}} \frac{(r^*)^{t-1}  }{(t-1)!} \eul^{-r^*}}.
    \end{equation}
Thus, on the event that $\vGamma$ satisfies \eqref{eq_minmax_G}, for $\delta < \frac{1}{(t-1)!} (r^*)^{t-1} \eul^{-r^*}$, \eqref{eqLemma_V1delta_2}, \eqref{eqLemma_V1delta_3} and \eqref{eq_pmin} imply that
\begin{align}\label{eqLemma_V1delta_4}
    \Erw\brk{\vec T \mid \vm_t} \leq \exp\bc{\theta \log n\bc{1 - \frac{1-\theta}{\theta}r^*c \KL{\delta}{\pr\bc{\Po(r)=t-1}+o(1) }}}. 
\end{align}
By assumption, 
\begin{equation} \label{Eq:Cond_R}
    c \geq (1+\eps)\max\{c_1(r^*), c_2(r^*, \theta)\} \geq (1+\eps)\frac{1}{-\frac{1-\theta}{\theta}r^*\log\bc{1-\pr\bc{\Po(r^*)=t-1}}}.
\end{equation}
Thus,
\begin{align}\label{eqLemma_V1delta_5}
    \Erw\brk{\vec T \mid \vm_t} \leq \exp\bc{\theta \log n\bc{1 - \bc{1+\eps} \frac{\KL{\delta}{\pr\bc{\Po(r^*)=t-1}+o(1) }}{-\log\bc{1-\pr\bc{\Po(r^*)=t-1}}}}}. 
\end{align}
Now, since
\begin{align*}
   & \KL{\delta}{\pr\bc{\Po(r^*)=t-1}} \\
   = & -H(\delta) - \delta\log \pr\bc{\Po(r^*)=t-1} - (1-\delta)\log\bc{1-\pr\bc{\Po(r^*)=t-1}},
\end{align*}
we may choose $\delta > 0$ small enough such that  $\KL{\delta}{\pr\bc{\Po(r)=t-1}+o(1) }$ $\geq -\log\bc{1-\pr\bc{\Po(r)=t-1}}\nm{\cdot \bc{1-\min\{\eps/3, 1/2\}} } + o(1)$. Thus, on the event that $\vGamma$ satisfies \eqref{eq_minmax_G},
\begin{align}\label{eqLemma_V1delta_6}
    \Erw\brk{\vec T \mid \vm_t} = o(1).
\end{align}
Finally, Lemma~\ref{Lemma_GammaMinMax} implies that $\Erw[\vec T] = o(1)$. 
\end{proof}

Proposition~\ref{Lemma_rigid} implies that \whp \nm{, there will be no assignments that have both large overlap with~$\SIGMA$ and produce the same test results.}
To see this, consider flipping a single item $x_i$ that has $\SIGMA_{x_i}=1$ to non-defective. On the event~$\cR$, this change alters the outcomes of at least $\delta\Delta$ tests that previously contained exactly $t$ defectives. 
To ``undo'' the $\delta\Delta$ altered outcomes, while keeping the Hamming weight of the label vector fixed, one would need to promote at least one non-defective item inside each such test to defective. However, since in our design, the tests typically have small intersection, likely many such compensating flips will be necessary and trigger further changes in order to keep the test results fixed.

The next step is therefore to bound the number of alternative assignments~$\hat{\SIGMA}$ at large overlap~$\ell$ that are consistent with the observed outcomes, conditioned on~$\cR$ and~$\vGamma$, following the strategy of \cite[Lemma III.4]{aco_2019}:

\begin{lemma} \label{Lemma_big_overlap}
    Suppose that $(1- 1/\nm{\sqrt{\log n}} )k\leq \ell<k$. 
    Then with probability $1-o(n^{-2})$, 
    \begin{align}\label{eq_big_overlap}
        \Erw&\brk{Z_{k,\ell}(\mathbf{G},\hat{\SIGMA}) \mid \vGamma, \cR} = n^{-\omega(1)}.
    \end{align}
\end{lemma}
\begin{proof}
   As in the proof of Lemma~\ref{Lem:first_bound_small_conditioned}, due to symmetry, we may assume that $\SIGMA_{x_i}=\vecone\{i\leq k\}=\vecone_{[k]}$ and $\sigma_{x_i}=\vecone\{i\leq\ell\}+\vecone\{k<i\leq 2k-\ell\}$. Then for this choice of $\sigma$, as in \eqref{eq:exp_sym},
\begin{align}\label{eq:exp_sym2}
   &  \Erw\brk{Z_{k,\ell}(\mathbf{G},\hat{\SIGMA})\mid \vGamma,\cR}  = \binom{k}{\ell} \binom{n-k}{k-\ell} \pr\bc{\forall a \in [m]: \hat\sigma_a = (\hat\vecone_{[k]})_a \vert \vGamma,\cR}.
 \end{align}

Recall the \nm{balls into bins} argument used in the proof of Lemma \ref{Lem:first_bound_small_conditioned} using the random variables $(\vA_{i,j})_{(i,j) \in \bigcup_{a=1}^m \{a\} \times [\vGamma_a]}=((\vA_{i,j,1},\vA_{i,j,2}))_{(i,j) \in \bigcup_{a=1}^m \{a\} \times [\vGamma_a]}$, as well as the  i.i.d. random variables $(\vA_{i,j}')_{(i,j) \in \bigcup_{a=1}^m \{a\} \times [\vGamma_a]}=((\vA_{i,j,1}',\vA_{i,j,2}'))_{(i,j) \in \bigcup_{a=1}^m \{a\} \times [\vGamma_a]}$. 

Let
\begin{align*}
    \cM_{t}' &= \cbc{i \in [m]: \sum_{j \in [\vGamma_i]}\vA_{i,j,1}' = t,\sum_{j \in [\vGamma_i]}\vA_{i,j,1}'\vA_{i,j,2}' \leq t-1}
\end{align*}
and define the corresponding set $ \cM_t$ with respect to the random variables $(\vA_{i,j})$.
Then in the balls-into-bins experiment, the event $\cR$ has the following consequences: Matching the $n$ $\{0,1\}^2$-labelled balls to bins (i.e. tests) creates the random variables $\vA_{i,j}$. \cri{On the event $\cR$, every item that is defective under $\SIGMA$ spawns at least $\delta\Delta$  corresponding balls that get matched to (different) tests whose corresponding bins have exactly $t$ balls with a $1$ in their first coordinate. } \nm{This in particular applies to the $k-\ell$ items spawning $\Delta$   $(1,0)$ balls each.} Since at most $t$ such balls can appear jointly in each of these tests, their total number is lower bounded by $(k-\ell)\delta\Delta/t$. On the other hand, each such test has at least one $(1,0)$ ball. Thus,
for each $\ell$,
\begin{align*}
    \cR \subseteq \cbc{|\cM_t| \geq (k-\ell)\frac{\delta}{t}\Delta}.
\end{align*}
Proposition \ref{Lemma_rigid} implies that $\cbc{|\cM_t| \geq (k-\ell)\frac{\delta}{t}\Delta}$ is a w.h.p. event, and thus that
\begin{align} 
  & \binom{k}{\ell} \binom{n-k}{k-\ell} \pr\bc{\forall a \in [m]: \hat\sigma_a = (\hat\vecone_{[k]})_a \vert \vGamma, |\cM_t| \geq (k-\ell)\frac{\delta}{t}\Delta } = n^{-\omega(1)} \label{eq_suff_cond} \\ \Longrightarrow &
  \binom{k}{\ell} \binom{n-k}{k-\ell}  \pr\bc{\forall a \in [m]: \hat\sigma_a = (\hat\vecone_{[k]})_a \vert \vGamma,\cR}  = n^{-\omega(1)}. \nonumber 
\end{align}
To establish (\ref{eq_big_overlap}), it is thus sufficient to show (\ref{eq_suff_cond}). 

We next make the step to the independent model. Recall the events $\cS$ and $\cT$ from the proof of Lemma~\ref{Lem:first_bound_small_conditioned}. 
Because $\vA'=(\vA_{i,j}')_{i,j}$ given $\cT, \vGamma$ is distributed as $\vA=(\vA_{i,j})_{i,j}$ given $\vGamma$, and the event $|\cM_t| \geq (k-\ell)\frac{\delta}{t}\Delta$ can be read off from the ball allocation, 
\begin{align}
     \pr\!\bc{\forall a \in [m]: \hat\sigma_a = (\hat\vecone_{[k]})_a \vert \vGamma, |\cM_t| \geq (k-\ell)\frac{\delta}{t}\Delta } =  \pr\!\bc{\cS \Big\vert \vGamma, \cT, |\cM_t'| \geq (k-\ell)\frac{\delta}{t}\Delta }\!.
\end{align}
Moreover, using the conditional distributional identity once more, for $\vGamma$ satisfying \eqref{eq_minmax_G}, also
	    $\pr\bc{\abs{\cM_t'}\geq(k-\ell)\frac{\delta}{t}\Delta\vert \vGamma, \cT}= 1-o(1)$. Together with the estimates
 $\pr\bc{\abs{\cM_t'}\geq(k-\ell)\frac{\delta}{t}\Delta\vert \vGamma} \leq 1$ and (\ref{eqLLT8}), Bayes' formula thus implies that for the corresponding $\vGamma$,
	\begin{align*}
	    \pr\bc{\cT | \vGamma, \abs{\cM_t'}\geq(k-\ell)\frac{\delta}{t}\Delta}=\frac{\pr\bc{\abs{\cM_t'}\geq(k-\ell)\frac{\delta}{t}\Delta\vert \vGamma, \cT}   \pr\bc{\cT | \vGamma}}{\pr\bc{\abs{\cM_t'}\geq(k-\ell)\frac{\delta}{t}\Delta\vert \vGamma} } = \Omega((\Delta k)^{-3/2}).
	\end{align*}
	\color{black}
Once more, Bayes' Theorem implies that
\begin{align}
     \pr\bc{\cS \Big\vert \vGamma, \cT, |\cM_t'| \geq (k-\ell)\frac{\delta}{t}\Delta } \leq  O((\Delta k)^{3/2})  \pr\bc{\cS \Big\vert \vGamma, |\cM_t'| \geq (k-\ell)\frac{\delta}{t}\Delta }.
\end{align}
In order to bound $ \pr\bc{\cS \Big\vert \vGamma, |\cM_t'| \geq (k-\ell)\frac{\delta}{t}\Delta }$, let
\begin{align*}
    \cS'' = \cbc{\forall i \in \cM_{t}':\sum_{j \in [\vGamma_i]}\vA_{i,j,2}' \geq t}.
\end{align*}
Then 
\begin{align}
    \cS = \cbc{\forall i \in [m]: \vecone\cbc{\sum_{j \in [\vGamma_i]}\vA_{i,j,1}' \geq t} = \vecone\cbc{\sum_{j \in [\vGamma_i]}\vA_{i,j,2}' \geq t}} \subseteq 
    \cS'',
\end{align}
so that $ \pr\bc{\cS \Big\vert \vGamma, |\cM_t'| \geq (k-\ell)\frac{\delta}{t}\Delta } \leq \pr\bc{\cS'' \Big\vert \vGamma, |\cM_t'| \geq (k-\ell)\frac{\delta}{t}\Delta }$. Since we work in the independent model, given $\cM_{t}'$ and $\vGamma$, 
\begin{align*}
   \pr\!\bc{\cS''\mid \vGamma, \cM_t'}&=\!\!\!\prod_{i\in \cM_t'}\!\! \pr\!\bc{\sum_{j \in [\vGamma_i]}\vA_{i,j,2}' \geq t \Big \vert  \sum_{j \in [\vGamma_i]}\vA_{i,j,1}' = t, \!\!\sum_{j \in [\vGamma_i]}\vA_{i,j,1}'\vA_{i,j,2}' \leq t-1, \vGamma}\!. 
\end{align*}
In each of the probabilities, we condition on the event that the first coordinate of $\vA_{i,j}'$ is one for exactly $t$ $j \in [\vGamma_i]$, out of which at least one has a zero second coordinate. Thus, if the sum over the second coordinates is at least $t$, then \emph{at least one} of the other $\vGamma_i-t$ random variables must have a $1$ in its second coordinate. These considerations lead to the upper bound
	\begin{align*}
	\pr\bc{\cS''\mid \vGamma, \cM_t'}&\leq \prod_{i\in \cM_t'} \bc{1-\bc{1-\frac{k-\ell}{n-k}}^{\vGamma_{i} -t}}\leq
	\bc{1-\bc{1-\frac{k-\ell}{n-k}}^{\vGamma_{\max}}}^{\abs{\cM_t'}}.
	\end{align*}
    The tower property now implies that
	\begin{align*}
	 \pr\bc{\cS'' \Big\vert \vGamma, |\cM_t'| \geq (k-\ell)\frac{\delta}{t}\Delta } &= \Erw\brk{\Erw\brk{\vecone \cS''\vert\vGamma,\cM_t'}\vert \vGamma, |\cM_t'| \geq (k-\ell)\frac{\delta}{t}\Delta}\\ 
     &\leq \bc{1-\bc{1-\frac{k-\ell}{n-k}}^{\vGamma_{\max}}}^{(k-\ell)\delta \Delta/t}.
	\end{align*}    
    Thus,
    \begin{align} 
  & \binom{k}{\ell} \binom{n-k}{k-\ell} \pr\bc{\forall a \in [m]: \hat\sigma_a = (\hat\vecone_{[k]})_a \vert \vGamma, |\cM_t| \geq (k-\ell)\frac{\delta}{t}\Delta } \nonumber \\ 
  \leq &
  O((\Delta k)^{3/2})\binom{k}{\ell} \binom{n-k}{k-\ell} \bc{1-\bc{1-\frac{k-\ell}{n-k}}^{\vGamma_{\max}}}^{(k-\ell)\delta \Delta/t}. \label{eq_almost_final_bound}
\end{align}
Thus, it remains to show that (\ref{eq_almost_final_bound}) is $n^{-\omega(1)}$. For $\vGamma$ satisfying \eqref{eq_minmax_G}, using a logarithm expansion and the definition of $\Delta$,
\begin{align*}
    \bc{1-\frac{k-\ell}{n-k}}^{\vGamma_{\max}}\!\!\!\! &= \exp\!\bc{\frac{\Delta n}{m}\log\bc{1-\frac{k-\ell}{n-k}} \!+\! O\!\bc{\sqrt{\frac{\Delta n}{m}}\log n \cdot \log\bc{1-\frac{k-\ell}{n-k}}\!}\!} \\
    &=\! \exp\bc{-r^*\bc{1-\frac{\ell}{k}} + O\bc{\nm{\frac{(k-\ell)}{\sqrt{kn}} \log n}}}.
\end{align*}
Thus,
\begin{align*}
     &\bc{1-\bc{1-\frac{k-\ell}{n-k}}^{\vGamma_{\max}}}^{(k-\ell)\delta \Delta/t} \\ &\qquad\qquad\qquad\leq \exp\bc{\frac{k\delta \Delta}{t} \log\bc{1- \exp\bc{-r^*\bc{1-\frac{\ell}{k}} + O\bc{\nm{\frac{(k-\ell)}{\sqrt{kn}} \log n}}}} }.
\end{align*}

\cri{We can upper bound the two binomial coefficients as
\begin{align*}
    \binom{k}{\ell}\binom{n-k}{k-\ell} = \binom{k}{k-\ell}\binom{n-k}{k-\ell} \leq \bc{\frac{\eul k}{(k-\ell)}}^{k-\ell} \bc{\frac{\eul (n-k)}{(k-\ell)}}^{k-\ell} \leq \bc{\frac{\eul^2 k(n-k)}{(k-\ell)^2}}^{k-\ell}. 
\end{align*}}
Putting everything together, we obtain that
\begin{align*}
    O((\Delta k)^{3/2})& \binom{k}{\ell} \binom{n-k}{k-\ell} \bc{1-\bc{1-\frac{k-\ell}{n-k}}^{\vGamma_{\max}}}^{(k-\ell)\delta \Delta/t} \\
    \leq O((\Delta k)^{3/2}) &\cdot \exp\bc{(k-\ell)\log\bc{\frac{\eul^2 k(n-k)}{(k-\ell)^2}}} \\ &\cdot \exp\bc{ \frac{\nm{(k-\ell)}\delta \Delta}{t} \log\bc{1- \exp\bc{-r^*\bc{1-\frac{\ell}{k}} + O\bc{\nm{\frac{(k-\ell)}{\sqrt{kn}} \log n}} }}}.
\end{align*}
\nm{We next take a closer look at the exponent
\begin{align}\label{eq_problematic_exponent}
    \log\bc{\frac{\eul^2 k(n-k)}{(k-\ell)^2}} + \frac{\delta \Delta}{t} \log\bc{1- \exp\bc{-r^*\bc{1-\frac{\ell}{k}} + O\bc{\nm{\frac{(k-\ell)}{\sqrt{kn}} \log n}} }}
\end{align}
 (up to multiplication by $k-\ell$). As $r^*(1-\ell/k) +  O\bc{\nm{\frac{(k-\ell)}{\sqrt{kn}} \log n}} \to 0$, 
\begin{align}\label{eq_some_expo_expa}
   & \exp\bc{-r^*\bc{1-\frac{\ell}{k}} +   O\bc{\frac{(k-\ell)}{\sqrt{kn}} \log n}} \nonumber \\
   &\qquad\qquad= 1 - r^*\bc{1-\frac{\ell}{k}} + O\bc{ \frac{(k-\ell)}{\sqrt{kn}} \log n + \bc{1-\frac{\ell}{k}}^2}.
\end{align}
Equation \eqref{eq_some_expo_expa} thus implies that
\begin{align}\label{eq_errors1}
     & \log\bc{1- \exp\bc{-r^*\bc{1-\frac{\ell}{k}} + O\bc{\nm{\frac{(k-\ell)}{\sqrt{kn}} \log n}} }} \nonumber \\
     &\qquad\qquad=  \log\bc{r^*\bc{1-\frac{\ell}{k}}\bc{1 + O\bc{ \sqrt{\frac{k}{n}} \log n + \bc{1-\frac{\ell}{k}}}}}.
\end{align}
Thanks to \eqref{eq_errors1}, we obtain the following upper bound on \eqref{eq_problematic_exponent}:
\begin{align}
   &\log\bc{\frac{nk}{(k-\ell)^2}}+ O(1) + \frac{\delta \Delta}{t} \bc{\log\bc{\frac{k-\ell}{k}} + \log r^* + O\bc{ \sqrt{\frac{k}{n}} \log n + \bc{1-\frac{\ell}{k}}}} \nonumber \\
   &= \bc{\frac{\delta\Delta}{t}-2}\log\bc{k-\ell} - \bc{\frac{\delta\Delta\theta }{t}-(1+\theta)}\log n + \frac{\delta \Delta}{t} \log r^* + O(1) \nonumber \\
   & \leq \frac{\delta\Delta \theta}{t}\log n - \frac{\delta\Delta}{2t}\log\log n - \bc{\frac{\delta\Delta\theta }{t}-(1+\theta)}\log n + \frac{\delta \Delta}{t} \log r^* + O(1).
\end{align}
In the last line, we have used that $k-\ell \leq k/\nm{\sqrt{\log n}}$. Finally, cancelling terms and using that $\Delta = \Theta(\log n)$, \eqref{eq_problematic_exponent} is upper bounded by
\begin{align}
    - \frac{\delta\Delta}{2t}\log\log n + O(\log n).
\end{align}
}
Putting everything together, we obtain, for $\vGamma$ satisfying \eqref{eq_minmax_G},
\begin{align} 
  O((\Delta k)^{3/2})\binom{k}{\ell} \binom{n-k}{k-\ell}  &\pr\bc{\forall a \in [m]: \hat\sigma_a = (\hat\vecone_{[k]})_a \vert \vGamma,\cR} \nonumber \\  
  &=  \nm{\exp\bc{-\Omega\bc{\log n \cdot \log\log n}} = n^{-\omega(1)}.  }  
  \nonumber 
\end{align}
	
\end{proof}

We have now established that the event $\cR$ occurs \whp{} for the specific assumption on $m$ in Lemma \ref{Lemma_rigid} (see \eqref{Eq:Cond_R}). Further, Lemma \ref{Lemma_big_overlap} has provided an upper bound for the expected number of competing solutions when $\cR$ occurs for large overlap. Thus, the final step is to prove that the sum of these expectations is $o(1)$ regardless of the choice of $c$.

\begin{proof}[Proof of Proposition \ref{Prop:LrgU}]
By Lemma~\ref{Lemma_big_overlap}, for all $k-k/\nm{\sqrt{\log n}} \leq \ell <k$ and $\vGamma$ satisfying \eqref{eq_minmax_G}
\begin{align*}
     \Erw&\brk{Z_{k,\ell}(\mathbf{G},\hat{\SIGMA}) \mid \vGamma, \cR} = n^{-\omega(1)}.
\end{align*}
On the other hand, for $k-k/\nm{\sqrt{\log n}} \leq \ell <k$, by Markov's inequality and Lemma~\ref{Lemma_big_overlap},
\begin{align*}
    & \pr\bc{ \exists k-k/\nm{\sqrt{\log n}} \leq \ell <k: Z_{k,\ell}(\mathbf{G},\hat{\SIGMA}) \geq 1} \\
     = &\pr\bc{\sum_{\ell =  k-k/\nm{\sqrt{\log n}} }^{k-1} Z_{k,\ell}(\mathbf{G},\hat{\SIGMA}) \geq 1} \\
    \leq &\sum_{\Gamma \text{ satisfying \eqref{eq_minmax_G}}} \pr\bc{\sum_{\ell =  k-k/\nm{\sqrt{\log n}} }^{k-1}  Z_{k,\ell}(\mathbf{G},\hat{\SIGMA}) \geq 1 \vert \vGamma=\Gamma, \cR} \pr\bc{\vGamma=\Gamma, \cR} + o(1) \\
    \leq & \sum_{\Gamma \text{ satisfying \eqref{eq_minmax_G}}} \sum_{\ell =  k-k/\nm{\sqrt{\log n}}}^{k-1} \Erw\brk{Z_{k,\ell}(\mathbf{G},\hat{\SIGMA}) \vert \vGamma=\Gamma, \cR} \pr\bc{\vGamma=\Gamma, \cR} + o(1) = o(1).
\end{align*}
\end{proof}

\subsection{Proof of Lemma \ref{Lem:T2}}\label{sec:conj2}

\begin{proof}
For $r>0$, let $q(r)$ and $p(\alpha, r)$ be defined as in equations \eqref{var:q} and \eqref{var:p} and set $\beta=2-\alpha$. As observed in the paragraph above \eqref{Eq:ToStudy}, to establish \eqref{Eq:Conj_reduced} for $t=2$, it is sufficient to establish \eqref{Eq:ToStudy} for all $r \in [1,2]$, which we will do. 
Writing out the different cases, we obtain
$$q(r) = \eul^{-r}\bc{1+r} \quad \text{and} \quad p(\beta,r) =\nm{\e^{-\beta r}\bc{\bigl(1+(\beta-1)r\bigr)^2+(2-\beta)r}}.$$
Substituting these expressions into \eqref{Eq:ToStudy}, and cancelling the common factor $e^{-\beta r}$, yields
$$ (1+r)^\beta \geq \bigl(1+(\beta-1)r\bigr)^2+(2-\beta)r. $$
We thus aim to show that
\begin{align}\label{eq:Conj2}
h(\beta,r) :=(1+r)^\beta -\bigl(1+(\beta-1)r\bigr)^2 -(2-\beta)r > 0
\end{align}
for all $r\in[1,2]$ and $\beta\in (1,2)$. We will consider $h$ as a function of $\beta$ for fixed $r$.
First observe that
\begin{align*}
h(1,r) =& (1+r)^1 - \left[ (1 + 0 \cdot r)^2 + (2-1)r \right] = (1+r) - (1 + r) = 0, \\
h(2,r) =& (1+r)^2 - \left[ (1 + (2-1)r)^2 + (2-2)r \right] = (1+r)^2 - (1+r)^2 = 0,
\end{align*}
and that 
\begin{align*}
\partial_{\beta} h(\beta,r) &= (1+r)^\beta \log(1+r) - \bc{ r + 2r^2(\beta-1) }, \\
\partial_{\beta}^2 h(\beta,r) &= (1+r)^\beta (\log(1+r))^2 - 2r^2, \\
\partial_{\beta}^3 h(\beta,r) &= (1+r)^\beta (\log(1+r))^3.
\end{align*}
From this, since $r>0$, we conclude that, for all $\beta \in [1, 2]$,
$$\partial_{\beta}^3 h(\beta, r) > 0.$$
This implies that for all $r \in [1,2]$, the second derivative $\partial_{\beta}^2 h(\beta,r)$ is a strictly increasing function in $\beta$. A strictly increasing function can cross the zero-axis at most once. Thus, $\partial_{\beta}^2 h(\beta, r)$ has at most one zero in $[1, 2]$, which in turn implies that $\partial_{\beta} h(\beta, r)$ has at most two zeros in $[1,2]$.
 Furthermore, we can evaluate $\partial_{\beta}h(1,r)$ and $\partial_{\beta} h(2,r)$ for $r\in [1,2]$ as
\begin{align*}
\partial_{\beta} h(1,r) =& (1+r) \log(1+r) - r>0; \text{ and} \\
\partial_{\beta} h(2,r)=&(1+r)^2 \log(1+r) -(r+2r^2)<0.    
\end{align*}
This implies that $\beta\mapsto h(\beta,r)$ is positive both in a neighbourhood of $1$ and of $2$.
Now for contradiction suppose that $m:=\min_{\beta \in [1,2]}h(\beta,r) <0$. Since $\beta \mapsto h(\beta,r)$ is continuous, there exists a local maximum $n_1$ of $\beta\mapsto h(\beta,r)$ in $(1,m)$ and a local maximum $n_2$ of $\beta\mapsto h(\beta,r)$ in $(m,2)$.  Then $\partial_{\beta} h(m, r)=\partial_{\beta} h(n_1, r)=\partial_{\beta} h(n_2, r)=0$, and we have found three distinct zeros of $\beta\mapsto \partial_{\beta}h(1,r)$, which leads to a contradiction. Therefore $m=0$ for any $r \in [1,2]$, which establishes \eqref{Eq:Conj_reduced} via \eqref{eq:Conj2}.
\end{proof}

\subsection{Proof of Proposition~\ref{prop_upper_main}} \label{SSec:ProofUpperMain}

Choose $\Delta = r^*c \log(n/k)$. Lemma \ref{Lem:IntTc} (under the assumption of Conjecture \ref{Conj}) and Proposition \ref{Prop:LrgU}  readily imply that $Z_k(\G,\hat\SIGMA)=1$ \whp\ if $m>(1+\eps)c_{\mathrm{inf}}^{\mathrm{TGT}}k\log(n/k)$. Moreover, Lemma \ref{Lem:T2} shows that for $t=2$, the conclusion of Lemma \ref{Lem:IntTc} holds true without the assumption of Conjecture~\ref{Conj}. Hence, for $t=2$ or $t\geq 3$ assuming Conjecture \ref{Conj}, \Cor~\ref{Cor_Nishi}(2) shows that there exists an inference algorithm that given $\G,\hat\SIGMA$ and $k$ outputs $\SIGMA$ \whp.

\section{Linear Number of Defectives} \label{Sec:LinearK}
 In Section \ref{SSec:LinFixRow}, we prove Theorem \ref{Thm:LinMain} on the linear regime of TGT with i.i.d. defectives.  Finally, in Section \ref{ssec:ach_linear}, we present a simple, `genuinely TGT' pooling scheme, and corresponding efficient inference algorithm, that requires $O(n\log n)$ tests.

\subsection{More than Linear Number of Tests Required: Proof of Theorem \ref{Thm:LinMain}}\label{SSec:LinFixRow}
Recall that we assume that each item is defective independently with probability $\alpha\in(0,1)$ and we work under the following assumption on the test design:
\begin{assumption}\label{ass_item_degrees}
    Assume that $\vec G$ is a test design on $n$ items such that the degree of each item is at most $\log^4 n$, and the degree of each test is at most $M$ for some $M>0$ not depending on $n$.
\end{assumption}

Throughout this section, we follow the strategy used to derive the information-theoretic lower bound for noisy linear CGT in Hintze et al. \cite{hintze2024noisy}. This yields a $\log n$-improvement over an application of the naive counting bound, and allows us to consider reasonably general pooling schemes.

 Hintze et al. \cite{hintze2024noisy} do not have a corresponding assumption on the test degrees, as they can remove items of large degree from any test design $G$ while making the inference problem for the remaining items only easier. This makes use of the fact that in CGT, all tests that are adjacent to a defective item do not contain information about their remaining items in the first place. On the other hand, non-defective items can always safely be removed. In our more general setting however, tests adjacent to defective items actually do contain information about their remaining items.

\subsubsection{Optimal Estimators}
Once again, our aim is to show that \emph{any} estimator fails \whp{} unless supplied with a sufficient number of tests, which can first be reduced to considering the Maximum A Posteriori (MAP) estimator.
\begin{definition}[MAP Estimator]\label{def:map}
For any test design $G$ and any observed outcomes $\hat{\sigma}\in \cbc{0,1}^m$,
the \emph{MAP estimator} $\tilde{\sigma}_{{\rm MAP}}(G, \hat{\sigma}) = \tilde{\sigma}_{{\rm MAP}}^{G,\hat{\sigma}}$ of $\SIGMA$ is defined as
\begin{align}\label{eq:map-argmax}
\tilde{\sigma}_{{\rm MAP}}^{G,\hat{\sigma}}
= \arg \max_{\tilde{\sigma}\in \cbc{0,1}^n} \Pr \bc{\SIGMA = \tilde{\sigma} \mid G,\hat{\sigma}},
\end{align}
where, in the event of a non-unique $\arg \max$, we select any $\tilde{\sigma}$ containing the greatest number of zeros.
\end{definition}
Indeed, the MAP estimator is the optimal estimator for exact recovery:
\begin{fact}[{\cite[Section 3.2]{Abbe17}}]\label{lem:map-best}
For any test design $G$ and any estimator $\cA:G\times\cbc{0,1}^m\to\cbc{0,1}^n$, the estimator $f$ performs no better than the {\rm MAP} estimator, i.e.,
\begin{align*}
\Pr\bc{\SIGMA = \tilde{\sigma}_{{\rm MAP}}^{G,\hat{\SIGMA}}}
\geq
\Pr\bc{\SIGMA = \cA(G,\hat{\SIGMA}) }.
\end{align*}
\end{fact}
Thus, in the following, it is sufficient to consider the performance of the MAP estimator. We will compare it to a (non-practical) algorithm that has oracle-access to additional information about the ground truth, which, \nm{following \cite{hintze2024noisy},}  we call the {\em genie} estimator:

\begin{definition}[Genie estimator]\label{def:genie}
    Given a test design $G$ and the observed results $\hat{\sigma}$, with $\SIGMA_{-i}$ the ground truth for every other item except $x_i$, the \emph{genie-based estimator} (or just \emph{genie estimator}) is given by
    \begin{align}\label{eq:genie}
        \tilde{\sigma}_{{\rm gen}}^{G,\hat{\sigma}}(i)=\tilde{\sigma}_{{\rm gen}}^{G,\hat{\sigma},\SIGMA_{-i}}(i) = \arg \max_{s\in \{0,1\}} \Pr\bc{\SIGMA(i) = s \mid G,\hat{\sigma},\SIGMA_{-i}}.
    \end{align}
    If both $0$ and $1$ maximise $\Pr\bc{\SIGMA(i) = s \mid G,\hat{\sigma},\SIGMA_{-i}}$, then we set $\tilde{\sigma}_{{\rm gen}}^{\vec G,\hat{\sigma}}(i)=0$.
\end{definition}

Indeed, given that the genie estimator has access to additional information, it is not surprising that it performs even better than the MAP estimator. The following lemma is a direct generalisation of Lemma~4.2 in \citep{hintze2024noisy}:
\begin{lemma}\label{lem_bests}
 For every test design $G$,
    \begin{align*}
        \pr\bc{\SIGMA = \tilde{\sigma}_{{\rm gen}}^{G,\hat{\SIGMA}} } \geq \pr\bc{\SIGMA = \tilde{\sigma}_{{\rm MAP}}^{G,\hat{\SIGMA}} }.
    \end{align*}
\end{lemma}

We next characterise the genie estimator in our setting of TGT without noise. For this, recall the notion of a \emph{pivotal test} from Definition~\ref{Def:Piv}. For $i \in [n]$, let $g_i = g_i(G,\SIGMA)$ denote the number of pivotal tests involving item $x_i$. If the labels of all items except of $x_i$ are known, and $x_i$ does not participate in any pivotal test, then its label can be $0$ or $1$ without changing the test outcomes. In this case, the genie estimator cannot do better than guess the most prevalent status. On the other hand, if $x_i$ has a pivotal test, then the status of $i$ determines the outcome of that pivotal test and thus can be read off from its outcome. Let $g_i^+(\SIGMA)$ and $g_i^-(\SIGMA)$ be the number of positive and negative pivotal tests for item $i$, respectively. Observing that non-defective items cannot participate in positive pivotal tests, and defective items cannot participate in negative pivotal tests, we have just argued the following: 

\begin{claim} \label{lem:geniethr} 
    For any test design $G$ and $i \in [n]$, the genie estimator equals
    \begin{align*}
        \tilde{\sigma}_{{\rm gen}}^{G,\hat{\SIGMA}}(i) = \begin{cases}
            0 & \quad \textup{if $g^{+}_i(\SIGMA)=0$ and $g^{-}_i(\SIGMA)>0$ or if  $g^{+}_i(\SIGMA)=g^{-}_i(\SIGMA)=0$ and $\alpha \leq \frac{1}{2}$},\\
            1 & \quad \textup{if $g^{-}_i(\SIGMA)=0$ and $g^{+}_i(\SIGMA)>0$ or if  $g^{+}_i(\SIGMA)=g^{-}_i(\SIGMA)=0$ and  $\alpha > \frac{1}{2}$}.
        \end{cases}
    \end{align*}
\end{claim}

Claim~\ref{lem:geniethr} immediately implies the following claim:
\begin{claim} \label{lem:genie_problow}
Fix any test design $G$. If $\alpha\leq \tfrac12$, then
    \begin{align}
         \mathds{1}\cbc{ \tilde{\sigma}_{\mathrm{gen}}^{G, \hat{\SIGMA}}(i) \neq \SIGMA(i)} \geq \mathds{1}\cbc{\vec g_i^+(\SIGMA) =0, \SIGMA(i)=1}.
    \end{align}
 If $\alpha > \tfrac12$, then
    \begin{align}
         \mathds{1}\cbc{ \tilde{\sigma}_{\mathrm{gen}}^{G, \hat{\SIGMA}}(i) \neq \SIGMA(i)} \geq \mathds{1}\cbc{\vec g_i^-(\SIGMA) =0, \SIGMA(i)=0}.
    \end{align} 
\end{claim}

\subsubsection{Modifying the Test Design}
Observe that any test that has degree less than $t$ in $G$ is uninformative, as its outcome will always be $0$, irrespectively of $\SIGMA$.
This last observation motivates the following definition:
\begin{definition}[Modified test design]\label{def:G_eta}
    Given a test design $G$,
    we construct the \emph{modified test design} $G'$ by removing all tests in $L = \cbc{a \in F \mid \abs{\partial_{G} a} \leq t-1}$. We denote the vector of test results in $G'$ by $\hat{\SIGMA}'$.
\end{definition}

It is clear that the modified test set-up is not harder than the original one:

\begin{claim}
\label{lem_mod_easier} Using Lemma~\ref{lem_bests},
    \begin{align}\label{eq_mod_easier}
        \pr\bc{\SIGMA =  \tilde{\sigma}_{\mathrm{gen}}^{G', \hat{\SIGMA}'}} =  \pr\bc{\SIGMA =  \tilde{\sigma}_{\mathrm{gen}}^{G, \hat{\SIGMA}}} \geq  \pr\bc{\SIGMA =  \tilde{\sigma}_{\mathrm{MAP}}^{G,\hat{\SIGMA}}}. 
    \end{align}
\end{claim}

\subsubsection{Proof of Theorem \ref{Thm:LinMain}}
In this section, $G$ denotes an arbitrary test design satisfying Assumption \ref{ass_item_degrees} and $G'$ its modification as in Definition~\ref{def:G_eta}.  
Restricting both item and test degrees will allow us to find a large set of items that do not share tests, which in turn will enable us to treat their genie estimates as independent. 

Following \cite{hintze2024noisy}, we call a set $D \subseteq [n]$ of items \emph{distant} in $G'$ if the second neighbourhoods of no two items in $D$ in the bipartite graph $G'$ intersect.
The degree restrictions that all items have degree at most $\log^4{n}$ and all tests have degree at most $M$ ensure that following a greedy approach of selecting an item $i$, and afterwards removing all items in the fourth neighbourhood of $i$, we remove a maximum of $M^2 \log^8 n$ items. By following this process, then it is possible to construct distant sets of size at least $n/(M^2\log^{9}n)$ in $G'$.

For a given distant set $D \subseteq [n]$, let 
\begin{align}
  \vD_{\mathrm{err}} = \cbc{i \in D:  \tilde{\sigma}_{\mathrm{gen}}^{G', \hat{\SIGMA}'} (i) \not= \SIGMA(i)}
\end{align}
be the (random) set of misclassified items in $D$. The following lemma now relates correctness of the genie estimator to the expected number of misclassified items in any given fixed distant set:

\begin{lemma}[{\cite[Lemma 4.7]{hintze2024noisy}}]\label{lem:genie_correct_upper_bound}
    For any set of distant items $D \subseteq [n]$,
        $$\Pr \bc{\SIGMA = \tilde{\sigma}_{\mathrm{gen}}^{G', \hat{\SIGMA}'}\mid \SIGMA}             \leq \bc{\Erw\brk{\vD_{\mathrm{err}}\mid \SIGMA}}^{-1}.$$
\end{lemma}
The proof of \cite[Lemma 4.7]{hintze2024noisy} also applies to our set-up, which is why we omit it here. Finally, recall the lower bound on the error probabilities from Claim \ref{lem:genie_problow}. The next main lemma, which is in the same spirit as \cite[Lemma 4.9]{hintze2024noisy}, bounds these error probabilities 
from below:

\begin{lemma}\label{lem:main_lower}
For any 
$\delta''>0$, there \nm{exists} an integer $n_0=n_0 (\delta'')$ \nm{and a constant $C>0$}
with the following property. For all $n \geq n_0$, every test design $G$ satisfying Assumption \ref{ass_item_degrees} \cri{on $m \leq Cn \log n$ tests has} a set of distant items $D \subseteq [n]$ such that
\begin{enumerate}[a)]
\item if $\alpha\leq \tfrac12$, then
    \begin{align}\label{eq:main_lower00}
        \pr\bc{\sum_{i \in D}\mathds{1}\cbc{\vec g_i^+(\SIGMA) =0, \SIGMA(i)=1} > \frac{1}{2\delta''}} \geq 1-4\delta'';
    \end{align}
\item if $\alpha > \tfrac12$, then 
    \begin{align}\label{eq:main_lower01}
        \pr\bc{\sum_{i \in D}\mathds{1}\cbc{\vec g_i^-(\SIGMA) =0, \SIGMA(i)=0} > \frac{1}{2\delta''}} \geq 1-4\delta''.
    \end{align}
\end{enumerate}
\end{lemma}

\begin{proof}
 First, assume that $\alpha \leq \tfrac12$ and that $D$ is an arbitrary distant set in $G'$. For $i \in D$, let us define the indicator of the event that $i$ is both defective and disguised by
 $$\vA_i := \mathds{1}\cbc{\vec g_i^+(\SIGMA) =0, \SIGMA(i)=1},$$
 as well as $\vA :=\sum_{i \in D} \vA_i$. 
Since $D$ is a distant set, it follows that $\vA_i$ and $\vA_j$ are independent for $i,j\in D$ with $i\neq j$, i.e., 
    \begin{equation}
    \label{cond-independence-Ai}
    \Erw\brk{\vA_i\vA_j}=\Erw\brk{\vA_i}\Erw\brk{\vA_j}.
    \end{equation}

We first bound $\Erw\brk{\vA_i} =  \pr\bc{\vec g_i^+(\SIGMA) =0, \SIGMA(i)=1}$ by a loose bound that allows application of the FKG inequality from below. First observe that if $|\partial_{G'}i|=0$, then $i$ appears in no tests in $G'$, so that, trivially,
    \begin{align}
        \log\Erw\brk{\vA_i} = \log \alpha =: L_i.
    \end{align}
In the following, we thus assume that $|\partial_{G'}i|>0$. In this case, for $a \in  \partial_{G'}i$, let 
\begin{align}
    \vX_a := \abs{\cbc{j \in \partial_{G'}a \setminus \{i\}: \SIGMA(j)=1}}.
\end{align}
To simplify the argument we bound $\Erw\brk{\vA_i}$ by considering only the scenario where every neighbouring test of $i$ is negative. Specifically, an item $i$ is disguised if for every test $a \in \partial i$, the number of other defective items $\vX_a$ satisfies $\vX_a + \SIGMA(i) < t$. 

Given $\SIGMA(i)=1$, this requires $\vX_a \leq t-2$ for all $a \in \partial i$. Thus,
 \begin{align*}
  \Erw\brk{\vA_i}  &\geq   \pr\bc{ \bigcap_{a \in \partial i} \cbc{\vX_a \leq t-2} \cap \cbc{  \SIGMA(i) = 1}} = \alpha \cdot \pr\bc{ \bigcap_{a \in \partial i} \cbc{\vX_a \leq t-2}}.
 \end{align*}
By construction, the events $\{\vX_a \leq t-2\}$ are decreasing, as then the only randomness is in the defective statuses, which are independent random variables. Applying the FKG inequality, we obtain
 \begin{align}
      \Erw\brk{\vA_i} \geq \alpha \prod_{a \in \partial i} \pr\bc{\vX_a \leq t-2}.
 \end{align}
Because we have enforced the condition that $|\partial i| \leq \log^4 n$
and $|\partial a| \leq M$ for all $i \in D$, the number of terms in the product is at most $\log^4 n$, and each individual probability is a positive constant independent of $n$ (specifically, $\pr\bc{\vX_a \leq t-2}\geq \pr(\Bin(|\partial a|-1, \alpha) \le t-2) \ge \pr(\Bin(M-1, \alpha) \le t-2) > 0$). Consequently, \invisible{there exists a constant $L$ such that,} for all $i \in D$,
 \begin{align}\label{eq_lower_Ai}
      \log \Erw\brk{\vA_i} \geq \log \alpha +  \sum_{a \in \partial i} \log \pr\bc{\vX_a \leq t-2} =: L_i.
 \end{align}
 
 We thus have a lower bound $L_i$ for any $i\in D$ which depends on the degrees of the neighbours of $i$ in $G'$. By the same case distinction into three cases corresponding to the neighbourhood of $i$ in $G'$, we can define $L_i$ for any $i \in [n].$ The next key observation is that we can choose a distant set $D$, in which no pairs share a test, such that for each $i \in D$, its lower bound is almost at least as big as the average over the lower bounds over all $j \in [n]$. As this construction works exactly as the corresponding construction in \cite{hintze2024noisy} as we have stricter degree bounds, we omit the details here
 (note that both sides in \eqref{Li-bound} are non-positive):

\begin{claim}[{\cite[Claim D.1]{hintze2024noisy}}] \label{claim_distant}
For any $\delta >0$, a distant set $D$ of size $\lfloor n/\log^{9}n \rfloor$ and $n_0=n_0(\delta)$ exist such that, for all $n \geq n_0$ and $i \in D$,
\begin{align}
\label{Li-bound}
  L_i  \geq (1+\delta) \frac{\sum_{j \in [n]} L_j }{n}.
\end{align}
\end{claim}

Let now $D$ be a specific distant set as in Claim~\ref{claim_distant}. From \eqref{eq_lower_Ai}, we obtain
\begin{align}
  \log \Erw\brk{\vA_i}   &\geq (1+\delta) \frac{\sum_{j \in [n]} L_j }{n} \nonumber \\
  & = (1+\delta)\log \alpha + (1+\delta)\frac{1}{n} \sum_{a \in F_{G'}} |\partial_{G'}a| \log \pr\bc{\vX_a \leq t-2}, \nonumber
\end{align}
where we recall that $F_{G'}$ is the collection of tests in $G'$. 
We proceed by bounding this sum. Let 
\begin{align}
    c_{\mathrm{na}} = c_{\mathrm{na}}(\alpha,M) := \min_{t \leq \Gamma \leq M} \Gamma \log \pr\bc{\Bin(\Gamma-1, \alpha) \leq t-2}.
\end{align}
Since the minimum is taken over a set where the probabilities are bounded away from 0 and 1, we obtain $c_{\mathrm{na}}\in(-\infty,0).$ Note that it is here where we use that no multi-edges are present. 
We use this to bound
\begin{align}
  \log \Erw\brk{\vA_i}   &\geq (1+\delta)\log \alpha + (1+\delta)\frac{1}{n} \sum_{a \in F_{G'}} c_{\mathrm{na}} \geq (1+\delta)\log \alpha + (1+\delta)c_{\mathrm{na}} m, 
\end{align}
where we recall that $m$ denotes the number of tests.

We are now ready to complete the proof of Lemma \ref{lem:main_lower}.
Now assume that the modified test design has at most $Cn\log n$ tests for a constant $C=C(\alpha,M)>0$ that we will choose small enough later, we get that
\begin{align}
\label{exp-Ai-LB}
   &\log \Erw\brk{\vA_i}   \geq   (1+\delta)\log \alpha + (1+\delta)\frac{c_{\mathrm{na}} Cn\log n}{n}.
\end{align}

Fix $\delta''>0$. We first argue that to prove the lemma, it is sufficient to show that $\Erw\brk{\vA}$ tends to infinity as $n \to \infty$: Indeed, using the independence in \eqref{cond-independence-Ai}, which is due to the fact that $D$ is a distant set, we get that 
    \begin{equation}
        \Var\bc{\vA}=\sum_{i\in D} \Var\bc{\vA_i}.
    \end{equation}  Further, $\Var\bc{\vA_i}\leq \Erw\brk{\vA_i}$ since $\vA_i\in \{0,1\}$. Therefore, an application of Chebyshev's inequality along the same lines as in the derivation of \cite[(D.26)]{hintze2024noisy} yields that, for every $y>0$,  \begin{align}
    \label{eq:main_lower1}
        \Pr\bc{\abs{\vA-\Erw\brk{\vA}} \geq y \Erw\brk{\vA}} &\leq \frac{\Var\bc{\vA}}{y^2\Erw\brk{\vA}^2}    = \frac{\sum_{i\in D} \Var\bc{\vA_i}}{y^2\Erw\brk{\vA}^2}
        \leq         \frac{1}{y^2\Erw\brk{\vA}}.
    \end{align}
    If $1/\Erw\brk{\vA} \leq \delta''$, then \eqref{eq:main_lower1} with the choice $y=\tfrac12$ implies that
    \[\Pr\bc{\vA \leq \frac{1}{2\delta''}}
        \leq \Pr\bc{\vA \leq \frac{1}{2}\Erw\brk{\vA}}
        \leq \frac{4}{\Erw\brk{\vA}}
        \leq 4\delta''.\]
Thus, to prove \eqref{eq:main_lower00}, it is enough to show $1/\Erw\brk{\vA} \leq \delta''$ for $n$ large enough, which we will now do.
 
For this, we use \eqref{exp-Ai-LB} to obtain\
\begin{align}
    \Erw\brk{\vA} \geq \frac{n}{M^2 \log^{9}n}\exp\bc{ (1+\delta)\brk{\log \alpha + c_{\mathrm{na}} C\log n}} \to \infty,
\end{align}
if $\delta$ and $C$ are chosen small enough. This completes the proof of Lemma \ref{lem:main_lower} a). 

The proof of part b) follows by an almost identical argument, since the events $\{\vX_a \leq t-2\}$ are also disguising for a non-defective item. Applying the FKG-inequality therefore yields the lower bound
 \begin{align}
      \Erw\brk{\vA_i} \geq (1-\alpha) \prod_{a \in \partial i} \pr\bc{\vX_a \leq t-2}.
 \end{align}
The remainder of the proof proceeds analogously, with only minor modifications arising from this new bound.

\end{proof}

\begin{proof}[Proof of Theorem \ref{Thm:LinMain}]
Fix $\eps>0$. By  Fact~\ref{lem:map-best} and Claim~\ref{lem_mod_easier}, to prove Theorem \ref{Thm:LinMain}, it is sufficient to show that $ \pr\bc{\SIGMA =  \tilde{\sigma}_{\mathrm{gen}}^{G', \hat{\SIGMA}'}}<\eps$.
      To begin, let us combine Lemma \ref{lem:genie_correct_upper_bound} and Claim \ref{lem:genie_problow} to form a bound on the probability that the genie estimator is correct, for any set of distant items $D \subseteq [n]$ and when $\alpha \leq \tfrac12$,
    \begin{align}
    \Pr \bc{\SIGMA = \tilde{\sigma}_{\mathrm{gen}}^{G', \hat{\SIGMA}'}\mid \SIGMA}
        \leq \biggl(\sum_{i \in D} \mathds{1}\cbc{\vec g_i^+(\SIGMA) =0, \SIGMA(i)=1}\biggr)^{-1} =: \mathbf U_D.\end{align}
    Lemma \ref{lem:main_lower} states that for any $\delta''>0$, we can choose $D$ so that $\Pr\bc{\mathbf U_D \leq 2 \delta''} \geq 1 - 4 \delta''$.
    Since trivially $\Pr \bc{\SIGMA = \tilde{\sigma}_{\mathrm{gen}}^{G'}\mid \SIGMA} \leq 1$,
    \begin{align}
    \Pr \bc{\SIGMA = \tilde{\sigma}_{\mathrm{gen}}^{G'}}
       &= \Erw\brk{\Pr \bc{\SIGMA = \tilde{\sigma}_{\mathrm{gen}}^{G'}\mid \SIGMA}}
        \leq \Erw\brk{\min\cbc{1, \mathbf U_D}}
    \\ &\leq \Erw \brk{\min\cbc{1, \mathbf U_D}\mid \mathbf U_D \leq 2 \delta''} + \Pr\bc{\mathbf U_D > 2\delta''}
        = 2\delta'' + 4\delta'', \nonumber
    \end{align}
        which yields Theorem \ref{Thm:LinMain} with $\delta'' = \eps/12$. The case $\alpha >1/2$ is analogous.
\end{proof}

\subsection{Achievability} \label{ssec:ach_linear}

In this section, we show that in the linear regime with i.i.d. prior, there exist a simple randomised test design in $O(n\log n)$ tests and an efficient algorithm that guarantee \whp{} exact recovery of $\SIGMA$.

\begin{theorem}\label{thm:linear_upper}
For any $0<\alpha <1$ and $\eps>0$, there is $n_0=n_0(\alpha,\eps)$ such that for every $n>n_0$ there exist a  randomised test
design $\vec G$ comprising 
\begin{align}
    m \leq (1+\eps) \frac{1}{t\alpha^{t-1}}n\log n
\end{align}
tests and a polynomial time algorithm TDD that given $\vec G$ and the test results $\hat{\SIGMA}_{\vec G}$ outputs $\SIGMA$ w.h.p.
\end{theorem}

We proceed to describe the randomised test design $\vec G$ and the TDD algorithm from Theorem \ref{thm:linear_upper}. Both are based on 
the observation that in TGT, the only way to obtain \textit{definite} information is from positive tests of size exactly $t$. Each such positive test will allow us to classify exactly $t$ items as defective. Thus, let $\vec G$ denote the random test design where each of the $m$ tests has degree $t$ and chooses its $t$ items independently from the other tests and uniformly without replacement. Given this test design $\vec G$, the TDD algorithm is then defined as follows:

\begin{algorithm}
\KwIn{ $\vec G, \hat \SIGMA$}
\KwOut{ An estimate $\tilde \sigma$ of $\SIGMA$}
Set $\tilde \sigma = (0) \in \cbc{0,1}^{n}$; \\ 
    \For{$j = 1, \ldots, m$}{
    \If{$\hat{\sigma_j}=1$}{\For{all $x_i \in \partial a_j$}{Set $\tilde \sigma_i = 1$;}}}
Return $\tilde{\sigma}$.\\
\caption{Threshold Definitive Defectives (TDD).}
\label{Algo_Naive}
\end{algorithm}

\begin{proof}[Proof of Theorem \ref{thm:linear_upper}]
Fix $\eps>0$ and let $m = \lfloor (1+\eps) \frac{1}{t\alpha^{t-1}}n\log n\rfloor$ be as in the statement of the theorem. The proof is based on a simple coupon-collection argument. Condition on the event $\{|\vec V_1|=k\}$, where we assume that $k \in [\alpha n - \sqrt{n}\log n, \alpha n + \sqrt{n} \log n]$. For each defective item and each fixed test, the probability that the test chooses this item along with $t-1$ other defective items is $\binom{k-1}{t-1}/\binom{n}{t}$. Thus, 
\begin{align}
 \pr\bc{\text{For all } a \in \partial_{\vec G}x_i: \hat\SIGMA_{a}=0 \Big\vert \SIGMA_i=1, \sum_{j=1}^k\SIGMA_j=k} =  \bc{1 - \frac{\binom{k-1}{t-1}}{\binom{n}{t}}}^m.
\end{align}
For $k \in [\alpha n - \sqrt{n}\log n, \alpha n + \sqrt{n} \log n]$, the probability that there is a defective item that is not revealed as defective by any of the $m$ tests is at most
\begin{align}
   n \exp\bc{-m \frac{t\alpha^{t-1}}{n}\bc{1 + o(1)}} =   \exp\bc{-(1+\eps)\bc{1 + o(1)}\log n + \log n} = o(1).
\end{align}
This holds for any number of defectives in the range $[\alpha n - \sqrt{n}\log n, \alpha n + \sqrt{n} \log n]$. However, by the Chernoff bound, the probability that the number of defective items is not in that range is $o(1)$. 

Finally, if all defective items are included in a test with $t-1$ other defective items, then the TDD algorithm will classify them correctly as defective (and by design, no false positives can be introduced). As TDD is a polynomial-time algorithm, this concludes the proof. 
\end{proof}

\begin{acks}[Acknowledgements]
\vspace{-0.3cm}
\cri{ We are deeply grateful to an anonymous reviewer for an exceptionally careful and constructive review of an earlier version of this work. Their detailed comments identified several important issues and led to substantial improvements in both the results and their presentation.
}
\end{acks}
\begin{funding}
\vspace{-0.8cm}
\phantom{a}
\InsertBoxR{0}{\includegraphics[scale=0.1]{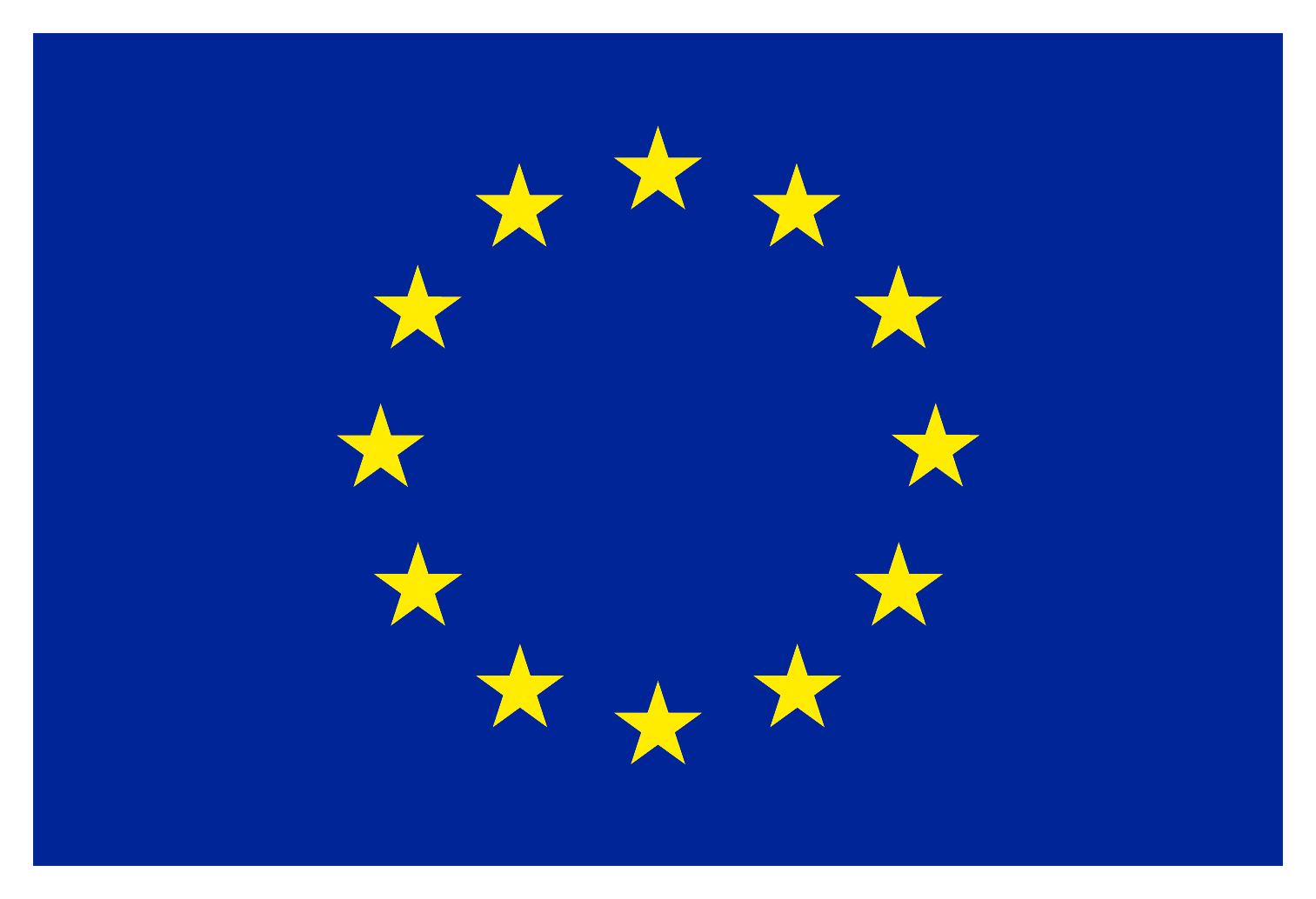}}
\noindent The first and second authors are supported in part by the NWO Gravitation project NETWORKS under grant no. 024.002.003.
\end{funding}

\bibliographystyle{imsart-nameyear} 
\bibliography{bibliography_j}       


\begin{appendix}
    \section{Concentration Inequalities} \label{Appendix_inequalities}

\begin{theorem}[Chernoff bound for the binomial distribution {\cite[Theorem 2.1]{janson2011random}}] 
\label{lem_chernoff}
Let $\vX \sim \Bin \bc{n, p}$. Then for any $\eps>0$, 
\begin{align*}
    & \Pr\bc{ \vX \geq (1 + \eps) \Erw \brk{\vX}} \leq \exp\bc{-\frac{\eps^2}{2 + 2\eps/3} \Erw\brk{\vX}} \qquad \text{and}  \\
    & \Pr\bc{\vX \leq (1 - \eps) \Erw\brk{\vX}} \leq \exp\bc{-\frac{\eps^2}{2} \Erw\brk{\vX}}.
\end{align*}
\end{theorem}

\begin{theorem}[Chernoff bound for the binomial distribution alternative \cite{janson2011random}] \label{lem_chernoff_2}
Let $\vX \sim \Bin \bc{n, p}$. Then
\begin{align}\label{eqChernoff1}
    \Pr\bc{X \geq {qN}} &\leq \exp \bc{-N\KL{q}{p}} \quad \text{for $p<q<1$,} \\
    \Pr\bc{X \leq {qN}} &\leq \exp \bc{-N\KL{q}{p}} \quad \text{for $0<q<p$.}\label{eqChernoff2}
\end{align}
\end{theorem}

\begin{theorem}[Chernoff bound for the hypergeometric distribution {\cite[Theorem 2.1]{janson2011random}}] \label{lem_chernoff_hyp}
    Let $\vX$ have the hypergeometric distribution with parameters $N, n$ and $m$. Then Theorem \ref{lem_chernoff} holds with $\Erw[\vX] = mn/N$.
\end{theorem}

\begin{theorem}[Chernoff bound for the Poisson-binomial distribution {\cite[Theorem 2.8]{janson2011random}}] 
\label{lem_chernoff_pb}
If $\vX_i \sim \Be(p_i), i=1, \dots, n$ are independent and $\vX = \sum_{i=1}^n \vX_i$ then the Chernoff bound for the binomial distribution holds. 
\end{theorem}

\begin{theorem}[Binomial coefficient $\tbinom{n}{k}$ for small $k$ {\cite[Section 5.1 - Case 5]{spencer2014asymptopia}}]\label{SmallBinomial}
Assume $k = o(n)$ then  
\[ \binom{n}{k} \sim \bcfr{n\mathrm{e}}{k}^k \bc{2 \pi k}^{-1/2} \mathrm{e}^{- \frac{k^2}{2n} (1+o(1))}. \]
\end{theorem}

\begin{theorem}[Binomial coefficient $\tbinom{n}{k}$ for $k$ linear in $n$ {\cite[Section 5.3]{spencer2014asymptopia}}]\label{LinearBinomial}
Assume $k \sim cn$ where $0<c<1$, then  
\[ \binom{n}{k} = \mathrm{e}^{n \bc{H(c)+o(1)}}, \]
 where $H$ is the entropy function \[ H(c) = -c \log c - (1-c) \log(1-c). \]
\end{theorem}

\begin{lemma}[Coupon collector: w.h.p. form {\cite[Theorem 5.13]{mitzenmacher2017probability}}] \label{Lem:CC} \phantom{a}
\begin{enumerate}[(1)]
  \item For any fixed $\varepsilon>0$, after $(1+\varepsilon)n\log n$ draws all coupons are collected \whp
  \item For any fixed $\varepsilon>0$, after $(1-\varepsilon)n\log n$ draws at least one coupon remains missing \whp
\end{enumerate}
\end{lemma}

Let $E$ be a finite set, and consider the configuration space $\{0,1\}^E$, whose elements $\omega$ are binary-valued configurations indexed by $E$. We equip this space with the natural coordinatewise partial order: for $\omega,\omega' \in \{0,1\}^E$,
we write
\[
\omega \le \omega' \quad \text{if and only if} \quad \omega_e \le \omega'_e \text{ for all } e \in E.
\]

This order allows us to formalise the notion of monotonicity for events and functions in $\{0,1\}^E$.

\begin{definition}[Increasing events and random variables]
An event $\cA \subseteq \{0,1\}^E$ is called \emph{increasing} if it is upward-closed, that is,
for all $\omega, \omega' \in {0,1}^E$ with $\omega \le \omega'$, the inclusion $\omega \in \cA$ implies $\omega' \in \cA$.

Similarly, a function (or random variable) $\vf : \{0,1\}^E \to \mathbb{R}$ is called \emph{increasing} if for all $\omega, \omega' \in {0,1}^E$ with $\omega \le \omega'$, one has $\vf(\omega) \le \vf(\omega')$.
\end{definition}

We next introduce the probabilistic setting. Let $\Pr_p$ denote the product measure on $\{0,1\}^E$ under which the coordinates $(\vX_e : e \in E)$ are independent Bernoulli random variables with parameter $p \in [0,1]$.

\begin{theorem}[FKG inequality {\cite[Section~2.2]{GrimmettPercolation}}] \label{FKG}
Under the measure $\Pr_p$. If $\cA, \cB \subseteq \{0,1\}^E$ are increasing events, then
\[
\Pr_p(\cA \cap \cB) \ge \Pr_p(\cA)\,\Pr_p(\cB).
\]

Equivalently, for any increasing random variables
$\vf, \vg : \{0,1\}^E \to \mathbb{R}$,
\[
\Erw_p[\vf \vg] \ge \Erw_p[\vf]\,\Erw_p[\vg].
\]
\end{theorem}

\section{Proof of Lemmas in Section \ref{SSec:Prelim}} \label{SSec:PrelimProofs}

\begin{proof}[Proof of Lemma \ref{lem_a0}]
Observe that
\begin{align}
    c_1'(r) = \frac{\pr\bc{\Po(r)= t-1}}{H(\pr\bc{\Po(r)\leq t-1})^2}\cdot \bc{\log\pr\bc{\Po(r)>t-1} - \log \pr\bc{\Po(r)\leq t-1}}.
\end{align}
In particular, the sign of $c_1'(r)$ only depends on the sign of $\log\pr\bc{\Po(r)>t-1} - \log \pr\bc{\Po(r)\leq t-1}$. \nm{Since  the function $\lambda \mapsto \pr\bc{\Po(\lambda)\leq t-1}$ is strictly decreasing from $1$ to $0$ in $\lambda$,} the latter has a unique zero in $r_1 \in (0,\infty)$, is negative on $(0,r_1)$ and positive on $(r_1,\infty)$. Moreover, $r_1$ satisfies $\pr\bc{\Po(r_1)\leq t-1} = \frac 1 2$. \nm{To obtain the stated bounds on $r_1$, we use known facts about the median of a Poisson random variable. Recall that for a $\Po(\lambda)$-random variable, its median $m$ is defined as the least integer $x$ for which $\pr\bc{\Po(\lambda)\leq x} \geq \frac{1}{2}$. Moreover,  $\lambda-\log 2 \leq m \leq \lambda +\frac 1 3$ (see \cite[Theorem 2]{choi1994medians}). As $t-1$ is an integer, the latter bounds imply that the median of a $\Po(t-1)$ random variable is $t-1$, so that $\pr\bc{\Po(t-1)\leq t-1} \geq 1/2$. Now, the function $\lambda \mapsto \pr\bc{\Po(\lambda)\leq t-1}$ is strictly decreasing in $\lambda$, from which we conclude that $r_1 \geq t-1$.
On the other hand, the lower bound on the median implies that for any $\lambda>t-1+\log 2$, $\pr\bc{\Po(\lambda) \leq t-1} < 1/2$, from which it follows that $r_1 \in \brk{t-1, t-1 + \log 2}$. }
\end{proof}

\begin{proof}[Proof of Lemma \ref{lem_a}]
Let $h: [0,\infty) \to \RR$ be given by
$$h(r) = r \log\!\big(1 - \exp(-r) r^{t-1}/(t-1)!\big).$$
Then the claim follows if we show that $h$ has a unique minimum $r_2$ on $(0,\infty)$, is strictly decreasing on $(0,r_2)$ and strictly increasing on $(r_2,\infty)$. Observe that 
\begin{align} \label{eq:derivative_c2}
    h'(r) = \log\bc{1- \exp(-r) \frac{r^{t-1}}{(t-1)!}} + \frac{\exp(-r) \frac{r^{t-1}}{(t-1)!}}{1-\exp(-r) \frac{r^{t-1}}{(t-1)!}}\bc{r-(t-1)}.
\end{align}
From this, it is straightforward to see that $h'$ is negative on the interval $(0,t-1]$. Now assume that $r \in [t-1,\infty)$. The equation $h'(r) = 0$ is equivalent to
\begin{align}
   h_1(r):= r-(t-1) = -  \frac{1-\exp(-r) \frac{r^{t-1}}{(t-1)!}}{\exp(-r) \frac{r^{t-1}}{(t-1)!}} \log\bc{1- \exp(-r) \frac{r^{t-1}}{(t-1)!}} =: h_2(r).
\end{align}
Clearly, $h_1$ is strictly increasing on $[t-1, \infty)$ with $h_1(t-1)=0$ and $h_1(r) \uparrow \infty$, $r \uparrow \infty$. On the other hand, $h_2(t-1) >0$, and we will show $h_2$ is strictly decreasing on $[t-1,\infty)$. This will give the existence of a unique $z$ with $h_1(z)=h_2(z)$, or equivalently, $h'(z) =0$. For this, we compute
\begin{align}
    h_2'(r) = 
   \bc{1 - \frac{t-1}{r}} \bc{\frac{1}{\exp(-r) \frac{r^{t-1}}{(t-1)!}}\log\bc{1- \exp(-r) \frac{r^{t-1}}{(t-1)!}} +1 }
\end{align}
For $r>t-1$, the first factor is positive, while using $\log(1+r) \leq r$ for $r> -1$, we see that the second factor is negative. Thus $h_2'(r) <0$ on $(t-1,\infty)$.

Thus $h$ has a unique local minimum, say $r_2$. The previous argument also shows that $h'$ is negative on $(0,r_2)$ and positive on $(r_2,\infty)$.

Finally, observe that $h'(t) =  \log\bc{1- \pr\bc{\Po(t)=t-1}} + \frac{\pr\bc{\Po(t)=t-1} }{1-\pr\bc{\Po(t)=t-1}}$. However the function $p \mapsto \log(1-p) + p/(1-p)$ is non-negative on $(0,1)$. Therefore $h'(t)>0$ and by the previous argument, $r_2 < t$.
\qedhere

\end{proof}

\begin{proof}[Proof of Lemma \ref{lem_b}]
Firstly, Lemma \ref{lem_a0} states the value of $r_1$ \nm{is} unique. Recall the definition of $h(r)$ from \nm{the proof of} Lemma \ref{lem_a} and its derivative \eqref{eq:derivative_c2}
\begin{align*}
    h'(r) = \log\bc{1- \exp(-r) \frac{r^{t-1}}{(t-1)!}} + \frac{\exp(-r) \frac{r^{t-1}}{(t-1)!}}{1-\exp(-r) \frac{r^{t-1}}{(t-1)!}}\bc{r-(t-1)}. 
\end{align*}
We 
\nm{show that $h'(r_1)<0$.}
Using $\log(1-x)< -x$ for $x \in (0,1)$
\begin{align*}
    h'(r) &\leq - \exp(-r) \frac{r^{t-1}}{(t-1)!} + \frac{\exp(-r) \frac{r^{t-1}}{(t-1)!}}{1-\exp(-r) \frac{r^{t-1}}{(t-1)!}}\bc{r-(t-1)} \\ 
    &\leq - \frac{\exp(-r) \frac{r^{t-1}}{(t-1)!}}{1-\exp(-r) \frac{r^{t-1}}{(t-1)!}}\bc{t-r-\exp(-r) \frac{r^{t-1}}{(t-1)!}}.
\end{align*}
Therefore, what remains to show is
$$ t-r_1-\exp(-r_1) \frac{r_1^{t-1}}{(t-1)!}>0. $$

For the case $t>2$, \nm{by Lemma \ref{lem_a0},} $t-r_1 \geq 1-\log 2 = 0.306$. 
Therefore, it suffices to bound the final term. Consider
$$f(x)=\exp(-x)x^{t-1} \qquad \text{whose derivative is} \qquad f'(x)=\exp(-x)x^{t-2}\bc{(t-1)-x}.$$
The function $f$ attains its maximum at $x=t-1$. Hence
$$\exp(-r_1)\frac{r_1^{t-1}}{(t-1)!} \leq \exp(-(t-1))\frac{(t-1)^{t-1}}{(t-1)!}.$$
Using Stirling's lower bound, and since $t\geq 3$, we find
$$\exp(-r_1) \frac{r_1^{t-1}}{(t-1)!} \leq \exp(-(t-1)) \frac{(t-1)^{t-1}}{(t-1)!} \leq \frac{1}{\sqrt{2 \pi (t-1)}} \leq \frac{1}{\sqrt{4 \pi}} =0.2821. $$
Combining the two bounds then yields the result.

For $t=2$ the bounds are not strong enough, but if we directly evaluate the function using $r_1=1.6783$ found in Table \ref{tab:values} then the relation holds.
\end{proof}

\begin{proof}[Proof of Lemma \ref{claim_unique_r}]
 Fix $\theta \in (0,1)$ and  let $C(r) = \max\{c_1(r), c_2(r,\theta)\}$. By Lemmas \ref{lem_a0} and \ref{lem_a}, $C$ is strictly decreasing on $(0, r_1]$ and strictly increasing on $[r_2,\infty)$. By continuity, $C$ attains a global minimum in the interval $[r_1, r_2]$. 

 On $(r_1,r_2]$, the function $c_1$ is strictly increasing, whereas $c_2$ is strictly decreasing. In particular, $c_1-c_2$ is continuous and strictly increasing on $[r_1,r_2]$. Therefore, there exists at most one $r\in[r_1,r_2]$ such that $c_1(r)=c_2(r,\theta).$
 If there exists no such $r$, one of the two functions, say $c_i$, dominates the other on the interval, and we have $r^*=r_i$. If there is exactly one such $r$, we have $C(s) \geq c_1(s) > c_1(r)$ for all $s>r$ and $C(s) \geq c_2(s) > c_2(r)$ for all $s<r$. Thus, $r=r^*$ is the unique global minimiser of $C$.
\end{proof}

\begin{proof}[Proof of Lemma~\ref{claim_opt_eq}]
\nm{Recall $r^*$ from \eqref{def_rstar} and that it is unique by Claim \ref{claim_unique_r}.} We first show that $ c_{\mathrm{inf}}^{\mathrm{TGT}} \geq  c_{\mathrm{inf}}^{\mathrm{TGT},*}$. 
Let $d^*=r^*c_1(r^*)$ and w.l.o.g. assume that $c_1(r^*) \geq c_2(r^*, \theta)$ (the other case follows by swapping indices in the following argument). Then $c_1(d^*/c_1(r^*))=c_1(r^*)$ and \nm{$c_2(d^*/c_1(r^*), \theta)=c_2(r^*, \theta) \leq c_1(r^*)$} by assumption. This implies that 
\begin{align*}
    c_1^*(d^*) \leq c_1(r^*) \qquad \text{and} \qquad c_2^*(d^*, \theta) \leq c_1(r^*),
\end{align*}
so that $c_{\mathrm{inf}}^{\mathrm{TGT}} \nm{= c_1(r^*)} \geq \max\cbc{c_{1}^*(d^*),  c_{2}^*(d^*,\theta) } \nm{\geq c_{\mathrm{inf}}^{\mathrm{TGT}, *}}$.

For the reverse bound, we make a case distinction as to whether \nm{$r^* \in \{r_1, r_2\}$ or not. First, assume that $r^* \in \{r_1, r_2\}$ and let $d>0$. By definition of $c_1^*(d), c_2^*(d,\theta)$, $c_1^*(d) \geq c_1(r_1)$ and $c_2^*(d,\theta)\geq c_2(r_2,\theta)$, as these are the minimum values the functions on the right hand sides of the inequalities $c \geq c_1(d/c), c \geq c_2(d/c, \theta)$ can take. Thus,
\begin{align}
    \max\cbc{c_1^*(d), c_2^*(d,\theta)} \geq \max\cbc{c_1(r_1), c_2(r_2, \theta)}.
\end{align}
In the case where $r^* \in \{r_1, r_2\}$, this implies that $c_{\mathrm{inf}}^{\mathrm{TGT}} \leq c_{\mathrm{inf}}^{\mathrm{TGT}, *}$.} 




Finally, assume that $c_1(r^*)=c_2(r^*,\theta) = c_{\mathrm{inf}}^{\mathrm{TGT}}$. \nm{First consider the case $d \geq r^*c_1(r^*)$. For all $c < c_1(r^*)$, $d/c > r^*$. Since $c_1$ is strictly increasing on $[r^*, \infty)$,
\begin{align*}
    c_1(d/c) > c_1(r^*) > c.
\end{align*}
The definition of $c_1^*$ implies that for all $d \geq r^*c_1(r^*)$, $c_1^*(d) \geq c_1(r^*)$. Thus for all $d \geq r^*c_1(r^*)$, $\max\cbc{c_1^*(d), c_2^*(d,\theta)} \geq c_{\mathrm{inf}}^{\mathrm{TGT}}$.}

\nm{Next, assume that $d \leq r^*c_1(r^*)$ and again consider $c < c_1(r^*)$. In this case, we make an additional case distinction. If $d/c \leq r^*$, then similarly as before (but now considering $c_2$),
\begin{align*}
    c_2(d/c, \theta) \geq c_2(r^*, \theta) = c_1(r^*) > c.
\end{align*}
Thus $c_2^*(d,\theta) \geq c_1(r^*)$. If $d/c >r^* > t-1$, then we consider the function $f$ given by $f(r) = r \cdot c_2(r,\theta)$. It is immediate to verify that $f'(r)$ has the opposite sign as $\frac{\mathrm{d}}{\mathrm{d}r}\pr\bc{\Po(r)=t-1}$. The latter is negative on the interval $(t-1,\infty)$, from which we conclude that
\begin{align*}
    r \mapsto r \cdot c_2(r,\theta) \quad \text{is strictly increasing on} \quad (t-1,\infty).
\end{align*}
Thus, under our assumption that $d/c >r^* > t-1$, we conclude that
\begin{align}
    \frac{d}{c} \cdot c_2\bc{\frac{d}{c}, \theta} > r^* c_2(r^*,\theta) \quad \Longrightarrow c_2\bc{\frac{d}{c}, \theta} > \frac{r^* c}{d} c_1(r^*) \geq c,
\end{align}
thanks to our assumption on $d$. Thus $c_2^*(d,\theta) \geq c_1(r^*)$ also in this case. Thus also for all $d \leq r^*c_1(r^*)$, $\max\cbc{c_1^*(d), c_2^*(d,\theta)} \geq c_{\mathrm{inf}}^{\mathrm{TGT}}$, which concludes the proof.
}
\end{proof}

\section{Proof of Lemma~\ref{Lemma_V_all}(iii)}
\label{app:Condition3}
To estimate $|\vV_0^+|$, again we first work without the multiplicity restriction. 
Let
\begin{align*}
    \vU_0 = \sum_{i=1}^n \mathds{1}\cbc{\SIGMA_i=0, \forall a \in \partial x_i: \vY_a \not=t-1} = \sum_{x \in \vV_0} \vecone\left\{ \sum_{a \in \partial x} \vecone \{\vY_a =t-1\} =0\right\}
\end{align*}
be the number of non-defective items that do not participate in any test with exactly $t-1$ defective items. 
Since every disguised non-defective item is counted by the corresponding indicator in $\vU_0$,  
$\vU_0$ provides an upper bound on the number
of disguised non-defectives.
 
       Let $\vGamma^{t-1} = (\vGamma_1^{t-1}, \dots, \vGamma_{\vm_{t-1}}^{t-1})$ denote the test degrees of those tests whose quantitative outcome equals $t-1$. 
    The w.h.p. bounds from Lemma~\ref{Lemma_GammaMinMax} also apply to $\vGamma^{t-1}$. Together with the bound from Lemma \ref{Lemma_m0} on $\vm_{t-1}$, we obtain
	\color{black}
	\begin{align*}
		\Erw[\vU_0\mid\vGamma^{t-1},\vm_{t-1}]&=(n-k)\binom{(n-k)\Delta - \sum_{i=1}^{\vm_{t-1}}\bc{\vGamma_i^{t-1} \!\!- (t-1)}}\Delta\binom{(n-k)\Delta}\Delta^{-1} \\ &= (n-k) \frac{\bc{(n-k)\Delta - \sum_{i=1}^{\vm_{t-1}}\bc{\vGamma_i^{t-1} - (t-1)}}_{\Delta}}{((n-k)\Delta)_{\Delta}}\\ &= \bc{1+n^{-\Omega(1)}} (n-k) \bc{\frac{(n-k)\Delta - \sum_{i=1}^{\vm_{t-1}}\bc{\vGamma_i^{t-1}\!\! - (t-1)}}{(n-k) \Delta}}^\Delta\\
		&=\bc{1+n^{-\Omega(1)}}(n-k)\bc{1-\frac{((1+o(1))(\Delta n/m) \vm_{t-1}}{n\Delta}}^\Delta \\
        &=\bc{1+n^{-\Omega(1)}}(n-k)\bc{1-\frac{1}{(t-1)!} \bcfr{d}{c}^{t-1}\exp(-d/c)}^\Delta.
	\end{align*}
	Similarly, we obtain
	\begin{align*}
		\Erw&[\vU_0^2\mid\vGamma^{t-1},\vm_{t-1}]\\&=(n-k)\binom{(n-k)\Delta - \sum_{i=1}^{\vm_{t-1}}\bc{\vGamma_i^{t-1} - (t-1)}}\Delta\binom{(n-k)\Delta}\Delta^{-1} \\ &\,\,\,\,\,\,+ (n-k)(n-k-1)\binom{(n-k)\Delta - \sum_{i=1}^{\vm_{t-1}}\bc{\vGamma_i^{t-1} - (t-1)}}{2\Delta}\binom{(n-k)\Delta}{2\Delta}^{-1}\\
        &=\Erw[\vU_0\mid\vGamma^{t-1},\vm_{t-1}]+\bc{1+n^{-\Omega(1)}}\Erw[\vU_0\mid\vGamma^{t-1},\vm_{t-1}]^2.
	\end{align*}
	In the present case, we assume that $k\bc{1-(d/c)^{t-1}\exp(-d/c)/(t-1)!}^{\Delta} \geq n^{\Omega(1)}$. Therefore, Chebyshev's inequality shows that \whp{}
	\begin{align}\label{eqLemma_V0plus_1}
		\vU_0&\in\brk{\bc{1\pm n^{-\Omega(1)}}n\bc{1-\frac{1}{(t-1)!} \bcfr{d}{c}^{t-1}\exp(-d/c)}^\Delta}.
	\end{align}	
We now impose the multiplicity restriction. Let now $\vR_0$ be the number of non-defective items that are contained in at least one test with multiplicity greater than one so that $\vU_0 \geq \abs{\vV_0^+} \geq \vU_0-\vR_0$. 
 Let $\kappa>0$ be such that $\vU_0 \geq n^{1-\theta} \cdot n^\kappa$ \whp\.~ It is possible to find such $\kappa$ thanks to \eqref{eqLemma_V0plus_1} and $k\bc{1-(d/c)^{t-1}\exp(-d/c)/(t-1)!}^{\Delta} \geq n^{\Omega(1)}$. 
By Markov’s inequality,
as in case (i), 
$$\Pr \bc{\vR_0 \geq \tfrac12 n^{1-\theta}n^\kappa} \leq \frac{2\,\Erw[\vR_0]}{n^{1-\theta} n^\kappa} = o(1).$$
Hence, \whp, simultaneously
$\vU_0 \geq n^{1-\theta} n^\kappa$ and $\vR_0 \le \tfrac12 n^{1-\theta} n^\kappa$, and therefore
$$\abs{\vV_0^+} \geq \vU_0 - \vR_0 \geq \tfrac12 n^{1-\theta}n^\kappa = n^{\Omega(1)} \quad \text{\whp.}$$

\section{{Results related to Conjecture~\ref{Conj}}}\label{appendix:conj}

In this short section, we derive properties of the function $W:[0,1] \times [0,\infty)$, 
$$W(\alpha,c):=1 - \alpha + c \bc{q(r^*) \log\bcfr{p(r^*,\alpha)}{q(r^*)} + (1-q(r^*)) \log\bcfr{1-2q(r^*)+p(r^*,\alpha)}{1-q(r^*)}}$$
appearing in \eqref{Eq:Conj} within Conjecture \ref{Conj}. To simplify the notation, we write $q = q(r^*)$ and $p(\alpha)=p(\alpha,r^*)$ below.

\subsection{Derivative at $\alpha=1$}

\begin{lemma}[Derivative at $\alpha=1$] \label{lem:W_derivative_boundary}
For any integer $t \geq 2$ and $c > c_{\mathrm{inf}}^{\mathrm{TGT}}$, 
$$\left. \partial_\alpha W(\alpha,c) \right|_{\alpha=1} > 0.$$
\end{lemma}

\begin{proof}
Recall the random variables $\mathbf{B}^*$ and $\mathbf{X}_1$ from Conjecture \ref{Conj}. Differentiating $W$ with respect to $\alpha$ yields
$$\partial_\alpha W(\alpha,c) = -1 + c r^* \sum_{j=0}^{t-1}\pr\bc{B^*=j}\pr\bc{X_1=t-1-j}^2 \bc{\frac{q}{p(\alpha)} + \frac{1-q}{1-2q+p(\alpha)}}.$$
Observe that for $\alpha=1$, $p(1) = q$ and the sum above reduces to $\Pr(\Po(r^*)=t-1)$. Thus
$$\partial_\alpha W(\alpha,c)|_{\alpha=1} = -1 + 2 c r^* \Pr(\Po(r^*)=t-1).$$
Let $g_t:(0,\infty) \to [0,\infty)$ be defined by $g_t(r) :=  r \Pr(\Po(r)=t-1)$. Differentiating $g_t$ by $r$ gives
$$g_t'(r) = (t-r) \Pr(\Po(r)=t-1).$$
In particular, $g_t$ is monotonically increasing on the interval $[t-1, t]$.
The proof of Lemma \ref{claim_unique_r} shows that $r^\ast \in [t-1, t]$, so that $g_t(r^\ast) \geq g_t(t-1)$
and thus
$$\partial_\alpha W(\alpha,c)|_{\alpha=1} = -1 + 2 c r^* \Pr(\Po(r^*)=t-1) \geq -1 + 2 c g_t(t-1).$$
We next argue that choosing $t=2$ provides a $t$-independent lower bound for the right hand side by showing that $g_2(1) < g_3(2) < \ldots$. For $s \geq 2$, 
$$\frac{g_{s+1}(s)}{g_s(s-1)} = \bcfr{s^{s+1} \e^{-s}}{s!} \bcfr{(s-1)!}{(s-1)^{s} \e^{-(s-1)}} = \frac{1}{\e} \bcfr{s}{s-1}^s > 1,$$
since $(1+1/n)^{n+1} \downarrow \eul$. Thus, for all $t \geq 2$, $g_t(t-1) \geq g_2(1) = \eul^{-1}$. On the other hand, $c > c_{\mathrm{inf}}^{\mathrm{TGT}} \geq c_1(r^*) \geq c_1(r_1) = (\log 2)^{-1}$. Consequently,
$$\partial_\alpha W(\alpha,c)|_{\alpha=1} > -1 + \frac{2}{\e \log 2} \approx 0.06 > 0.$$
\end{proof}

\begin{lemma}[Uniform bound on subinterval] \label{lem:analysis1}
Let $W:[0,1]\to\mathbb{R}$ be continuous with $W(x)<0$ for all $x \in [0,1)$ and $W(1)=0$. Moreover, assume  that $W$ is differentiable at $1$ with $W'(1)>0.$ Let $(\delta_n)_n$ be any sequence of positive real numbers such that $\delta_n\to 0.$  Then, for every constant $D<W'(1)$, there exists $N\in\mathbb{N}$ such that for all $n\ge N$,  \[ W(x)\le -D\delta_n \quad \text{for all} \quad  x\in [0,1-\delta_n].\quad \] 
\end{lemma}

\begin{proof}
Fix $D\in(0,W'(1))$ and a sequence $\delta_n\to 0.$  Since $W$ is differentiable at $1$ and $W(1)=0$,
$$
\lim_{x\uparrow 1}\frac{W(x)}{x-1}=W'(1).
$$
As $D<W'(1)$, there exists $\eta>0$ such that
$$
\frac{W(x)}{x-1}>D
$$
for every $x\in(1-\eta,1)$. Since $x-1<0$, multiplication by $x-1$ yields
\begin{align}
W(x)<D(x-1)=-D(1-x) \quad \text{on} \quad (1-\eta, 1). \label{eq:der1}
\end{align}
By continuity, $W$ attains its maximum on the compact interval $[0,1-\eta]$. Set
$$
M:=\max_{x\in[0,1-\eta]}W(x).
$$
Since $W(x)<0$ for every $x<1$, we have $M<0$. 
Because $\delta_n\to0$, there exists $N\in\mathbb{N}$ such that, for every $n\geq N$,
$$
\delta_n<\min\left\{\eta,\frac{-M}{D}\right\}.
$$
Fix $n\geq N$ and $x\in[0,1-\delta_n]$. 
If $x\in[0,1-\eta]$, then the definition of $W$ ensures that
$$
W(x)\leq M<-D\delta_n.
$$
If $x\in(1-\eta,1-\delta_n]$, then $1-x\geq\delta_n$, and therefore by \eqref{eq:der1}
$$
W(x)<-D(1-x)\leq-D\delta_n.
$$
Thus, in both cases,
\[
W(x)\leq-D\delta_n
\qquad\text{for all }x\in[0,1-\delta_n]\text{ and all }n\geq N.
\]
\end{proof}

\subsection{Proof of Proposition \ref{prop:conj_reduction}} \label{appendix:conj2}
Observe that for all $\alpha \in [0,1]$, $p(r^*,\alpha) \leq q(r^*)$, in the definition of $W$, the term multiplied by $c$ is negative. An upper bound for $W(\alpha,c)$ for all $c>c_{\mathrm{inf}}^{\mathrm{TGT}}$ and $\alpha \in (0,1)$ can therefore be obtained by fixing $c = c_1(r^\ast) \leq c_{\mathrm{inf}}^{\mathrm{TGT}}$. 

To simplify the notation, we now abbreviate $q = q(r^*)$ and $p(\alpha)=p(\alpha,r^*)$. Then
\begin{align*}
   W(\alpha, c_1(r^\ast))   & = 1-\alpha + \frac{1}{H(q)} \bc{q \log\bcfr{p(\alpha)}{q} + (1-q) \log\bcfr{1-2q+p(\alpha)}{1-q}} \\
    &= 2- \alpha - \frac{q \log p(\alpha)+  (1-q) \log\bc{1-2q+p(\alpha)}}{q \log q +(1-q) \log (1-q)}.
\end{align*}
For $W(\alpha, c_1(r^\ast))$ to be negative for $\alpha \in (0,1)$ (as proposed in Conjecture~\ref{Conj}), it is thus clearly sufficient that
\begin{align*}
    (2-\alpha) q \log q &> q \log p(\alpha) \qquad \text{and} \\
    (2-\alpha) (1-q) \log (1-q) &> (1-q)\log\bc{1-2q+p(\alpha)}.
\end{align*}
Further simplification gives
\begin{align}
    q^{2-\alpha} &> p(\alpha), \label{cond1a} \\
    (1-q)^{2-\alpha} &> 1-2q+p(\alpha). \label{cond1b}
\end{align}

\begin{lemma}[Condition \eqref{cond1b} follows from Condition \eqref{cond1a}] \label{lem:cond_dominance}
For any $\alpha \in [0, 1]$, if \eqref{cond1a} holds, then also \eqref{cond1b} holds. 
\end{lemma}
\begin{proof}[Proof of Lemma~\ref{lem:cond_dominance}]
First recall that the proof of Lemma \ref{claim_unique_r} shows that $r^* \in [r_1, r_2]$. Thus $r^* \geq r_1$ and hence $q = \pr\bc{\Po(r^*) \leq t-1} \leq \pr\bc{\Po(r_1) \leq t-1} = 1/2$.  For fixed $\alpha \in [0,1]$, let $F_{\alpha}: [0, 1/2] \to \mathbb{R}$ be defined by
\begin{equation}
    F_{\alpha}(x) = 2x - 1 + (1-x)^{2-\alpha} - x^{2-\alpha}.
\end{equation}
 As $0 \leq q \leq 1/2$, the claim follows if $F_{\alpha}(q)\geq 0$ for all $\alpha \in [0,1]$, as in this case,
$$(1-q)^{2-\alpha} - \bc{1 - 2q + p(\alpha)} \geq q^{2-\alpha} - p(\alpha).$$
We are thus left to show that $F_{\alpha}$ is positive on $[0,1/2]$. Differentiating $F_\alpha$ twice with respect to $x$ yields
$$F_{\alpha}''(x) = (2-\alpha)(1-\alpha) \bc{\frac{1}{(1-x)^{\alpha}} - \frac{1}{x^{\alpha}}}.$$
For $x \in [0, 1/2]$, we have $1-x \geq x$, which implies $(1-x)^{-\alpha} \leq x^{-\alpha}$ for $\alpha \in [0,1]$. Thus $F_{\alpha}''(x) \leq 0$, meaning $F_{\alpha}$ is concave on $[0, 1/2]$. Together with the boundary values $F_{\alpha}(0) = F_{\alpha}(1/2) = 0$, concavity implies that $F_{\alpha}(x) \geq 0$ for all $x \in [0, 1/2]$. 
\end{proof}
The derivation prior to the previous lemma has shown that if \eqref{cond1a} and \eqref{cond1b} hold for all $\alpha \in (0,1)$, then Conjecture~\ref{Conj} holds true. Lemma~\ref{lem:cond_dominance} states that indeed the hold of only \eqref{cond1a} implies that Conjecture~\ref{Conj} is true. 
This concludes the proof of Proposition~\ref{prop:conj_reduction}.



\section{Justification for the Absence of High-Degree Items}
\begin{proposition}\label{prop:bounded_item_degrees}
Let $M,C>0$ be fixed. If $\vec G$ is a randomised test design on $m \leq Cn\log n$ tests where each test has degree at most $M$ and chooses its items uniformly without replacement, then \whp, each item participates in at most $\log^4 n$ tests.
\end{proposition}

\begin{proof}
Denote the degree of test $j \in [m]$ in $\vec G$ by $\Gamma_j$. For each test $j \in [m]$, let $\vX_{ij}$ be the indicator that item $i \in [n]$ is included in test $j$. Then
$$\vX_i := \sum_{j=1}^m \vX_{ij}$$
counts the number of tests in which item $i$ participates, where $\vX_{ij} \sim \Be(\Gamma_j /n)$. 
By linearity of expectation,
$$\mathbb{E}[\vX_i] = \sum_{j=1}^m \frac{\Gamma_j}{n} \leq m\frac{M}{n} = CM \log n.$$
Applying the Chernoff bound for Poisson–binomial variables (Theorem \ref{lem_chernoff_pb}), we find
$$\Pr\bc{\vX_i > \log^4 n} \leq \exp\bc{-\Omega(\log^4 n)}.$$
Finally, taking a union bound over all $n$ items yields
$$\Pr\bc{\max_i \vX_i > \log^4 n} \leq n \exp\bc{-\Omega(\log^4 n)} = o(1).$$
Thus, with high probability, every item participates in at most $\log^4 n$ tests.
\end{proof}

\end{appendix}

\end{document}